\documentclass[journal, onecolumn, draftcls]{IEEEtran}

\usepackage[utf8]{inputenc} 
\usepackage[T1]{fontenc}
\usepackage{url}
\usepackage{ifthen}
\usepackage{cite}
\usepackage[cmex10]{amsmath} 

\usepackage{amsthm,amssymb,amsfonts}
\usepackage{multirow}
\usepackage{multicol}
\usepackage{enumitem}
\usepackage{braket}

\usepackage{pgfplots}
\pgfplotsset{compat=1.12}
\usepackage{tikz}
\usepackage{verbatim}
\usepackage{graphicx}
\usetikzlibrary{shapes,arrows, arrows.meta, positioning, decorations.markings, decorations.pathreplacing,calc,math} 
\usepackage{mathtools}
\DeclarePairedDelimiter{\ceil}{\lceil}{\rceil}

\theoremstyle{definition}
\newtheorem{theorem}{Theorem}
\newtheorem{assumption}{Assumption}

\newtheorem{lemma}{Lemma}
\newtheorem{claim}{Claim}
\newtheorem{proposition}{Proposition}

\newtheorem{example}{Example}
\newtheorem{remark}{Remark}

\newtheorem{definition}{Definition}
\newtheorem{observation}{Observation}

\newcommand{\vs}{\vspace{-0.1in}}

\newcommand{\calH}{\mathcal{H}}
\newcommand{\calI}{\mathcal{I}}

\newcommand{\calL}{\mathcal{L}}

\newcommand{\calX}{\mathcal{X}}
\newcommand{\calY}{\mathcal{Y}}
\newcommand{\calZ}{\mathcal{Z}}

\newcommand{\bfy}{\mathbf{y}}

\newcommand{\bfX}{\mathbf{X}}
\newcommand{\bfY}{\mathbf{Y}}

\newcommand{\bfx}{\mathbf{x}}

\newcommand{\bfz}{\mathbf{z}}

\newcommand{\bbC}{\mathbb{C}}

\newcommand{\bbN}{\mathbb{N}}

\newcommand{\bbR}{\mathbb{R}}

\newcommand{\frL}{\mathfrak{L}}

\newcommand{\Mod}[1]{\ \mathrm{mod}\ #1}

\newcommand{\Tr}{\textrm{Tr}}

\pgfmathsetmacro{\xW}{0}
\pgfmathsetmacro{\yW}{0}
\pgfmathsetmacro{\xX}{0}
\pgfmathsetmacro{\yX}{-1}
\pgfmathsetmacro{\xY}{sqrt(3)/2}
\pgfmathsetmacro{\yY}{0.5}
\pgfmathsetmacro{\xZ}{-sqrt(3)/2}
\pgfmathsetmacro{\yZ}{0.5}

\pgfmathsetmacro{\xM}{-sqrt(3)/3}
\pgfmathsetmacro{\yM}{5/3}
\pgfmathsetmacro{\xL}{sqrt(3)/3}
\pgfmathsetmacro{\yL}{5/3}
\pgfmathsetmacro{\xN}{-sqrt(3)}
\pgfmathsetmacro{\yN}{-1/3}
\pgfmathsetmacro{\xP}{-2*sqrt(3)/3}
\pgfmathsetmacro{\yP}{-4/3}
\pgfmathsetmacro{\xR}{sqrt(3)}
\pgfmathsetmacro{\yR}{-1/3}
\pgfmathsetmacro{\xQ}{2*sqrt(3)/3}
\pgfmathsetmacro{\yQ}{-4/3}

\pgfmathsetmacro{\xA}{0}
\pgfmathsetmacro{\yA}{4}
\pgfmathsetmacro{\xB}{-2*sqrt(3)}
\pgfmathsetmacro{\yB}{-2}
\pgfmathsetmacro{\xC}{2*sqrt(3)}
\pgfmathsetmacro{\yC}{-2} 

\pgfmathsetmacro{\Eps}{0.25}
\pgfmathsetmacro{\xBC}{0}
\pgfmathsetmacro{\yBC}{-2+\Eps}
\pgfmathsetmacro{\xCB}{0}
\pgfmathsetmacro{\yCB}{-2-\Eps} 

\pgfmathsetmacro{\xBA}{-sqrt(3) - sqrt(3) * \Eps / 2}
\pgfmathsetmacro{\yBA}{1 +  \Eps / 2}
\pgfmathsetmacro{\xAB}{-sqrt(3) + sqrt(3) * \Eps / 2}
\pgfmathsetmacro{\yAB}{1 -  \Eps / 2} 

\pgfmathsetmacro{\xAC}{sqrt(3) + sqrt(3) * \Eps / 2}
\pgfmathsetmacro{\yAC}{1 + \Eps / 2} 
\pgfmathsetmacro{\xCA}{sqrt(3) - sqrt(3) * \Eps / 2}
\pgfmathsetmacro{\yCA}{1 -  \Eps / 2}

\pgfmathsetmacro{\xMM}{-sqrt(3)}
\pgfmathsetmacro{\yMM}{7}
\pgfmathsetmacro{\xLL}{sqrt(3)}
\pgfmathsetmacro{\yLL}{7}
 
\pgfmathsetmacro{\xML}{0}
\pgfmathsetmacro{\yML}{7+\Eps}
\pgfmathsetmacro{\xLM}{0}
\pgfmathsetmacro{\yLM}{7-\Eps} 

\pgfmathsetmacro{\xMA}{-sqrt(3)/2 + sqrt(3) * \Eps / 2}
\pgfmathsetmacro{\yMA}{5.5 + \Eps / 2}
\pgfmathsetmacro{\xAM}{-sqrt(3)/2 - sqrt(3) * \Eps / 2}
\pgfmathsetmacro{\yAM}{5.5 - \Eps / 2} 

\pgfmathsetmacro{\xLA}{sqrt(3)/2 - sqrt(3) * \Eps / 2}
\pgfmathsetmacro{\yLA}{5.5 +  \Eps / 2}
\pgfmathsetmacro{\xAL}{sqrt(3)/2 + sqrt(3) * \Eps / 2}
\pgfmathsetmacro{\yAL}{5.5 -  \Eps / 2}

\pgfmathsetmacro{\xNN}{-4*sqrt(3)}
\pgfmathsetmacro{\yNN}{-2}
\pgfmathsetmacro{\xPP}{-3*sqrt(3)}
\pgfmathsetmacro{\yPP}{-5}

\pgfmathsetmacro{\xNB}{-3*sqrt(3)}
\pgfmathsetmacro{\yNB}{-2+\Eps}
\pgfmathsetmacro{\xBN}{-3*sqrt(3)}
\pgfmathsetmacro{\yBN}{-2-\Eps} 

\pgfmathsetmacro{\xNP}{-3.5*sqrt(3) + sqrt(3) * \Eps / 2}
\pgfmathsetmacro{\yNP}{-3.5 + \Eps / 2}
\pgfmathsetmacro{\xPN}{-3.5*sqrt(3) - sqrt(3) * \Eps / 2}
\pgfmathsetmacro{\yPN}{-3.5 - \Eps / 2} 

\pgfmathsetmacro{\xBP}{-2.5*sqrt(3) - sqrt(3) * \Eps / 2}
\pgfmathsetmacro{\yBP}{-3.5 +  \Eps / 2}
\pgfmathsetmacro{\xPB}{-2.5*sqrt(3) + sqrt(3) * \Eps / 2}
\pgfmathsetmacro{\yPB}{-3.5 -  \Eps / 2}

\pgfmathsetmacro{\xRR}{4*sqrt(3)}
\pgfmathsetmacro{\yRR}{-2}
\pgfmathsetmacro{\xQQ}{3*sqrt(3)}
\pgfmathsetmacro{\yQQ}{-5}

\pgfmathsetmacro{\xRC}{3*sqrt(3)}
\pgfmathsetmacro{\yRC}{-2+\Eps}
\pgfmathsetmacro{\xCR}{3*sqrt(3)}
\pgfmathsetmacro{\yCR}{-2-\Eps} 

\pgfmathsetmacro{\xRQ}{3.5*sqrt(3) + sqrt(3) * \Eps / 2}
\pgfmathsetmacro{\yRQ}{-3.5 - \Eps / 2}
\pgfmathsetmacro{\xQR}{3.5*sqrt(3) - sqrt(3) * \Eps / 2}
\pgfmathsetmacro{\yQR}{-3.5 + \Eps / 2} 

\pgfmathsetmacro{\xCQ}{2.5*sqrt(3) + sqrt(3) * \Eps / 2}
\pgfmathsetmacro{\yCQ}{-3.5 +  \Eps / 2}
\pgfmathsetmacro{\xQC}{2.5*sqrt(3) - sqrt(3) * \Eps / 2}
\pgfmathsetmacro{\yQC}{-3.5 -  \Eps / 2}

\tikzset{
    node_1/.style={circle, draw=black, minimum size=6mm, inner sep=0pt, font = \large},
    node_2/.style={rectangle, draw=black, minimum size=4, inner sep=2pt, font = \large},
    edge_1/.style={draw, -},
    edge_2/.style={draw, postaction={decorate, decoration={markings, mark=at position 0.3 with {\arrow{Stealth[scale=1]}}}}}
}

\begin{document}
\title{Quantum advantage in zero-error function computation with side information}

\author{ 
  \IEEEauthorblockN{Ruoyu Meng and Aditya Ramamoorthy}\\
  \IEEEauthorblockA{Department of Electrical and Computer Engineering\\ 
                    Iowa State University, Ames, IA , USA\\
                    Email: \{rmeng, adityar\}@iastate.edu}
}

\maketitle

\begin{abstract} 
We consider the problem of zero-error function computation with side information. Alice and Bob have correlated sources $X,Y$ with joint p.m.f. $p_{XY}(\cdot, \cdot)$. Bob wants to calculate $f(X,Y)$ with zero error. Alice encodes $m$-length blocks $(m \geq 1)$ of her observations to Bob over error-free channels, which can be classical or quantum. 
We consider two classical settings. (i) Alice communicates via a fixed length code (FLC), and (ii) Alice communicates via a variable length code (VLC). In the FLC scenario, the minimum communication rate depends on the asymptotic growth of the chromatic number of an appropriately defined $m$-instance ``confusion graph'' $G^{(m)}$. In the VLC scenario, the corresponding rate is characterized by the asymptotics of the chromatic entropy of $G^{(m)}$.
In the quantum setting, we only consider fixed length codes; the corresponding rate depends on the asymptotic growth of the orthogonal rank of the complement of $G^{(m)}$.  The behavior of the communication rates depends critically on $G^{(m)}$, which is shown to be sandwiched between $G^{\boxtimes m}$ ($m$-times strong product) and $G^{\lor m}$ ($m$-times OR product) respectively. Our work presents necessary and sufficient conditions on the function $f(\cdot, \cdot)$ and joint p.m.f. $p_{XY}(\cdot,\cdot)$ such that $G^{(m)}$ equals either $G^{\boxtimes m}$  or $G^{\lor m}$.  Our work explores the multitude of possible behaviors of the quantum and classical (FLC/VLC) rates in the single-instance case and the asymptotic (in $m$) case for several classes of confusion graphs. We show that there exist confusion graphs for which the quantum rate is exponentially smaller than the FLC and VLC rates.  With respect to FLCs, we show, e.g., that the quantum advantage may exist in the single-instance but disappear in the asymptotic rate. With respect to VLCs, we show that there are scenarios where the VLC rate can be smaller than the quantum rate and vice versa. 





\end{abstract}


\section{Introduction }\label{sec:introduction}

The problem of function computation over networks has applications in various settings, such as sensor networks \cite{akyildiz2002survey} and network monitoring. 
\begin{figure}[!b] 
\begin{center} 
\begin{tikzpicture}[node distance=2cm, >=latex]
    \node (source) [xshift=-0.5cm] {x};
    \node (encoder) [right=0.5cm of source, rectangle, draw] {Encoder};
    \node (decoder) [right=2cm of encoder, rectangle, draw] {Decoder};
    \node (fxy) [right=0.5cm of decoder] {\( f(x,y) \)};
    
    \node (y) at ($(decoder.south west)!0.5!(decoder.south east) + (0,-0.75cm)$) {y};

    \draw[->] (source) -- (encoder);
    \draw[-{Triangle[fill=white]}] (encoder) -- node[midway,above=0.15cm] {R} node[midway,sloped,anchor=center] {/} (decoder);
    \draw[->] (y) --  ($(decoder.south west)!0.5!(decoder.south east)$);
    \draw[->] (decoder) -- (fxy);

    \draw [decorate,decoration={brace,amplitude=5pt}] (encoder.north west) -- node[above, yshift=0.2cm] {Alice} (encoder.north east);
    \draw [decorate,decoration={brace,amplitude=5pt}] (decoder.north west) -- node[above, yshift=0.2cm] {Bob} (decoder.north east);

\end{tikzpicture}
\end{center}
 \caption{Function computation with side information
 \label{fig:func_computing_diag}} 
\end{figure}
In these settings, one is typically interested in the statistics of the data gathered at remote locations, e.g., mean/median/variance, and the like, rather than the data itself. Furthermore, in several of these settings the functions are computed over data that is correlated, e.g., in sensor networks, the readings from different sensors are likely to be similar. In such scenarios, a fundamental question of interest is - {\it ``how to transmit information efficiently so that the function of interest can be computed''.}

From an information-theoretic perspective, a model that captures several of the inherent subtleties and nuances of the problem is the one depicted in Fig.~\ref{fig:func_computing_diag} and is referred to as the problem of {\it ``function computation with side information''}. Alice observes a source $X$. Bob observes a source $Y$ that is correlated with $X$; $Y$ is referred to throughout as the side information. Bob wishes to compute a function $f(X,Y)$, and for this Alice needs to communicate a certain amount of information to him. The problem is evidently feasible, since Alice can always simply transmit the source $X$. However, determining the minimum amount of information that Alice needs to send is nontrivial in general. This rate depends on the function properties, the correlation structure of $X$ and $Y$ and on the requirements on the recovery of $f(X,Y)$. The case when Bob wants to recover $X$, i.e., when $f(X,Y) = X$, is called ``source coding with side information''.  

The classical version of this problem where Alice communicates classical bits has a long history. The typical setting considers $m$-length blocks of the sources $X_i, Y_i$, for $i \in [m]$ where $[m]= \{1, \dots, m\}$, and the schemes allow the calculation of $f(X_i, Y_i)$ for $i \in [m]$. Suppose that Bob is content with asymptotically error-free reconstruction. This means that the probability that there exists an index in $[m]$ such that the function computation is incorrect can be made arbitrarily small with increasing $m$. In this case, the work of \cite{OrlitskyR_01} provides a precise characterization of the required rate of transmission. In this scenario, for the source coding with side information setting, the solution follows from the Slepian-Wolf theorem \cite{slepian1973noiseless}. On the other hand, if Bob insists on recovering the function computation with zero-error, the problem is significantly more complicated. There are several works in this setting as well, e.g.,
\cite{Witsenhausen_76,FergusonB_75,ahlswede1979coloring,AlonO_95,LinialV89,OrlitskyR_01,KornerO_98,Charpenay_23}. Also, see  Part III of \cite{KornerO_98} for an extensive survey.

In this work, our focus is on the zero-error version of the function computation with side information. 
The study of classical zero-error source coding with side information was initiated in \cite{Witsenhausen_76}. This work introduced the notion of the confusion graph ({\it cf.} Definition \ref{defn:confusion_graph}). Furthermore, it showed that when using fixed length codes, the optimal classical transmission strategy for Alice is to color the confusion graph and transmit the color to Bob. Variable length codes were considered in \cite{AlonO_96} and the corresponding rate was characterized as the chromatic entropy of the confusion graph. Several variants of these problems in the zero-error setting have been examined \cite{CharpenayTR23ITW, CharpenayTR23ISIT, AlonO_96, Shayevitz_14, Malak_22}. 


\subsection{Problem Formulation: Classical and quantum settings}\label{subsec:problem_formulation}
The function computation with side information problem can be formally specified as follows. Alice observes a sequence of independent and identically distributed (i.i.d.)~observations of a random variable $X$ (taking values in a discrete alphabet $\calX$). Bob has access to i.i.d.~observations of a side information random variable $Y$ (taking values in a discrete alphabet $\calY$). $X$ and $Y$ are correlated such that their joint probability mass function (p.m.f.) is $p_{XY}(x,y), x \in \calX, y \in \calY$. Bob seeks to compute a function $f: \calX \times \calY \to \calZ$.  We assume that the channel from Alice to Bob is error-free and can support either classical or quantum transmission depending on the considered scenario.

The aim is to understand the minimum rate at which Alice can communicate information to Bob such that the function computation is successful. We consider $m$-length blocks of the sources $X_i, Y_i$, for $i \in [m]$, and the schemes allow the calculation of $f(X_i, Y_i)$ for $i \in [m]$. We consider $m$-length blocks, since it is often the case that computing multiple instances of the function at the same time allows Alice to send less information ``per'' computation than computing them one by one \cite{shannon1956zero}.

In this work, we consider the {\it zero-error} version of this problem. Thus, for any $m$-instance scheme, we require that Bob should be able to recover $f(X_i, Y_i)$ for all $i \in [m]$.
Suppose that Alice has symbols $x$ and $x'$ such that for all $y$ that occur with joint non-zero likelihood, i.e., $p_{XY}(x,y) > 0$ and $p_{XY}(x',y) > 0$, it holds that $f(x,y) = f(x',y)$. Then, from Bob's perspective, $x$ and $x'$ are equivalent and can be given the same description by Alice. Thus, Alice's symbols that need to be given ``different'' descriptions can be formalized by the following definition \cite{Witsenhausen_76}. 
\begin{definition}
\label{defn:confusion_graph}
{\it $f$-confusion graph.} The $f$-confusion graph of $X$ given $Y$ is a graph $G=(V,E)$ where $V=\calX$ and $(x_1,x_2) \in E$ if there exists $y \in \calY$ such that $p_{XY}(x_1,y) p_{XY}(x_2, y) > 0$ and $f(x_1, y) \neq f(x_2, y)$. 

\begin{figure}[!t]
\begin{center}


\resizebox{0.25\linewidth}{!}{
\begin{tikzpicture}
  \foreach \a in {5}{
    \node [regular polygon, regular polygon sides=\a, minimum size=3cm, draw] at (\a*4,0) (A) {};
      \node [circle, label=90:0, fill=black, inner sep=3pt] at (A.corner 1) {};
      \node [circle, label=90+72:1, fill=red, inner sep=3pt] at (A.corner 2) {};
      \node [circle, label=90+72*2:2, fill=black, inner sep=3pt] at (A.corner 3) {};
      \node [circle,label=90+72*3:3, fill=red, inner sep=3pt] at (A.corner 4) {};
      \node [circle, label=90+72*4:4, fill=green, inner sep=3pt] at (A.corner 5) {};
  }
\end{tikzpicture}
}

\end{center}
\vs
\caption{\label{fig:pentagon} {\small The figure shows the $f$-confusion graph $G$ of $f$ with p.m.f. defined in (\ref{eq:corr_prop_eq}). A coloring of $G$ can be performed with three colors. This is shown in the figure with red, black and green. When communicating over one instance, the classical rate can therefore be $\log_2 3$. It can be shown that the graph over two instances $G^{(2)}$ can be colored with five colors. This leads to a per-computation classical rate of $\frac{1}{2} \log_2 5$. This is in fact optimal \cite{Witsenhausen_76,lovasz1979shannon}.}}
\vs
\end{figure}


Similarly, if we consider computations over $m$ instances, then we can define the $f$-confusion graph over $m$ instances denoted $G^{(m)}$ analogously. Let $\bfx$ and $\bfy$ denote $m$-length Alice and Bob sequences, respectively. The vertex set corresponds to all $m$-length Alice sequences. $(\bfx_1,\bfx_2) \in E(G^{(m)})$ if there exists $\bfy$ such that $\Pi_{i=1}^m p_{XY}(x_{1i},y_i) p_{XY}(x_{2i}, y_i) > 0$ and there exists $j \in [m]$ such that $f(x_{1j},y_j) \neq f(x_{2j},y_j)$. 
We use shorthand $f^{(m)}(\bfx,\bfy):= \{f(x_i,y_i)\}_{i=1}^m$, and $p_{\bfX\bfY}(\bfx,\bfy)= \prod_{i=1}^mp_{XY}(x_i,y_i)$.
\end{definition} 

\begin{example}
Consider a scenario where  $X$ and $Y$ take values in $\{0, 1, \dots, 4\}$ and are correlated such that 
    \begin{align}
       p_{XY}(x,y)  &= \begin{cases}
            \frac{1}{10}, \text{~if~} y=x\text{~or~}y=(x+1)\Mod 5,\\
            0,\text{~otherwise.}
        \end{cases} \label{eq:corr_prop_eq}
    \end{align}
Suppose that Bob wishes to compute the equality function, i.e., $f(X,Y) = 1$ if $X=Y$ and $f(X,Y) =0$ if $X\ne Y$. We observe that Alice's symbols $0$ and $1$ are connected since $p_{XY}(0,1) p_{XY}(1,1) = 1/100 > 0$ and $f(0,1) \neq f(1,1)$; the reasoning for the other edges is similar. It can be observed that this $f$-confusion graph is a pentagon (see Fig. \ref{fig:pentagon}).   
\end{example}

In general, the structure of the $m$-instance confusion graph $G^{(m)}$ can be quite complicated. Graph products are classes of graphs that are more structured and extensively studied \cite{HammackGraphProduct}. It turns out that studying graph products allows us to also characterize general $m$-instance confusion graphs at some level.
In particular, we will deal extensively with the strong product and the OR product of graphs. These are defined below.
%
%

\begin{definition}\label{defn:AND_product_OR_product}
     The strong product of $G$ and $H$, denoted $G\boxtimes H$, has vertex set $V(G)\times V(H)$. $\big( (u_1,v_1),(u_2,v_2) \big )\in E(G\boxtimes H)$ if and only if
    \begin{align*}
        & \big (u_1=u_2\text{ and }(v_1,v_2)\in E(H)\big )\text{ or}\\
        & \big ((u_1,u_2)\in E(G)\text{ and }v_1=v_2\big )\text{ or}\\
        & \big ((u_1,u_2)\in E(G)\text{ and }(v_1,v_2)\in E(H)\big ).
    \end{align*}
    The $m$-fold strong product is written as $G^{\boxtimes m} := \underbrace{G\boxtimes G\boxtimes \dots \boxtimes G}_{m\text{ times}}.$
    
\noindent The OR product of $G$ and $H$, denoted $G\lor H$, has vertex set $V(G)\times V(H)$. $\big((u_1,v_1),(u_2,v_2)\big)\in E(G\lor H)$ if and only if 
    \begin{align*} 
            (v_1,v_2)\in E(H)&\text{ if }u_1=u_2,\\
            (u_1,u_2)\in E(G)&\text{ if }v_1=v_2,\\
            (u_1,u_2)\in E(G) \text{ or }(v_1,v_2)\in E(H)&\text{ if }u_1\ne u_2\text{ and }v_1\ne v_2. 
    \end{align*}
    The $m$-fold OR product is written as $G^{\lor m} :=\underbrace{ G\lor G\lor \dots \lor G}_{m\text{ times}}.$
\end{definition}

\begin{figure}[t]
    \centering
    \scalebox{1}{%
    \begin{tikzpicture} 
[
  node distance=0.5cm,
  every node/.style={circle, draw, minimum size=4mm, inner sep=0pt, font = \normalsize},
  every edge/.style={draw, -} 
]

\node[draw=none] (times) at (-2,2) {$\boxtimes$}; 

\node[draw=black] (001) at (-3,1) {1};
\node[draw=black] (002) at (-3,2) {2};
\node[draw=black] (003) at (-3,3) {3}; 

\node[draw=black] (01) at (-1,1) {1};
\node[draw=black] (02) at (-1,2) {2};
\node[draw=black] (03) at (-1,3) {3}; 

\node[draw=none] (equal) at (0,2) {$=$};

\node[draw=black] (11) at (1,1) {1,1};
\node[draw=black] (12) at (1,2) {1,2};
\node[draw=black] (13) at (1,3) {1,3}; 

\node[draw=black] (21) at (2,1) {2,1};
\node[draw=black] (22) at (2,2) {2,2};
\node[draw=black] (23) at (2,3) {2,3}; 

\node[draw=black] (31) at (3,1) {3,1};
\node[draw=black] (32) at (3,2) {3,2};
\node[draw=black] (33) at (3,3) {3,3};

\draw (01) edge (02);
\draw (02) edge (03); 

\draw (11) edge (12);
\draw (12) edge (13); 

\draw (21) edge (22);
\draw (22) edge (23); 

\draw (31) edge (32);
\draw (32) edge (33); 

\draw (001) edge (002);
\draw (002) edge (003);

\draw (11) edge (21);
\draw (21) edge (31); 

\draw (12) edge (22);
\draw (22) edge (32); 

\draw (13) edge (23);
\draw (23) edge (33);

\draw (11) edge (22);
\draw (21) edge (32); 

\draw (12) edge (23);
\draw (22) edge (33); 

\draw (12) edge (21);
\draw (22) edge (31); 

\draw (13) edge (22);
\draw (23) edge (32);

\node[draw=none] (comma) at (4,1) {$,$}; 

\node[draw=none] (times) at (6,2) {$\lor$}; 

\node[draw=black] (661) at (5,1) {1};
\node[draw=black] (662) at (5,2) {2};
\node[draw=black] (663) at (5,3) {3}; 

\node[draw=black] (61) at (7,1) {1};
\node[draw=black] (62) at (7,2) {2};
\node[draw=black] (63) at (7,3) {3}; 

\node[draw=none] (equal) at (8,2) {$=$};

\node[draw=black] (71) at (9,1) {1,1};
\node[draw=black] (72) at (9,2) {1,2};
\node[draw=black] (73) at (9,3) {1,3}; 

\node[draw=black] (81) at (10,1) {2,1};
\node[draw=black] (82) at (10,2) {2,2};
\node[draw=black] (83) at (10,3) {2,3}; 

\node[draw=black] (91) at (11,1) {3,1};
\node[draw=black] (92) at (11,2) {3,2};
\node[draw=black] (93) at (11,3) {3,3};

\draw (61) edge (62);
\draw (62) edge (63); 

\draw (71) edge (72);
\draw (72) edge (73); 

\draw (81) edge (82);
\draw (82) edge (83); 

\draw (91) edge (92);
\draw (92) edge (93); 

\draw (661) edge (662);
\draw (662) edge (663);

\draw (71) edge (81);
\draw (81) edge (91);

\draw (72) edge (82);
\draw (82) edge (92); 

\draw (73) edge (83);
\draw (83) edge (93);

\draw (71) edge (82);
\draw (81) edge (92);

\draw (72) edge (83);
\draw (82) edge (93); 

\draw (72) edge (81);
\draw (82) edge (91); 

\draw (73) edge (82);
\draw (83) edge (92);

\draw (71) edge (92);
\draw (71) edge (83); 

\draw (93) edge (81);
\draw (73) edge (81);

\draw (91) edge (72);
\draw (91) edge (83); 

\draw (73) edge (92);
\draw (93) edge (72); 

\end{tikzpicture}%
    }
    \caption{Strong product and OR product of $G$ with itself where $G$ is a path of two edges.}
    \label{fig:strong_or_product_P3}
\end{figure} 

\noindent Examples of strong and OR products appear in Fig. \ref{fig:strong_or_product_P3}.
In this work, we examine the zero-error function computation with side information problem under three different settings. In the first setting, Alice transmits   fixed length classical codes to Bob,  (i.e., the setting of \cite{Witsenhausen_76}). In the second setting, Alice transmits   variable length classical codes to Bob,  (i.e., the setting of restricted inputs in  \cite{AlonO_96}). In the third setting, Alice transmits fixed dimension quantum states to Bob. 
\begin{definition}
    Let $G=(V,E)$ be a graph. A $k$-coloring of $G$ is a labeling $f:V\mapsto S$, where $S$ is a set of $k$ elements, such that adjacent vertices have different labels. The chromatic number $\chi(G)$ is the minimum $k$ such that there exists a $k$-coloring of $G$.
\end{definition}

\subsubsection{\bf Fixed Length Classical Code (FLC) Setting}




In the fixed length classical code setting, based on the work of Witsenhausen \cite{Witsenhausen_76}, it can be observed that Alice's strategy corresponds to coloring $G^{(m)}$ with the fewest possible colors and transmitting the color of her realization\cite{Witsenhausen_76}. Thus, the FLC rate is defined as follows. 
\begin{definition} 
\label{defn:witsenhausen_classical_rate} {\it Optimal rate in the FLC setting.} The FLC rate is given by 
\begin{align*}
    R_{\text{FLC}} &= \lim_{m \to \infty} \frac{\log_2 \chi(G^{(m)})}{m} = \inf_m  \frac{\log_2 \chi(G^{(m)})}{m}.
\end{align*}
\end{definition}
The second equality above, follows from the sub-additivity of $\log_2 \chi(G^{(m)})$ \cite{NayakTR06}. In the sequel, we will use $R_{\text{FLC}}^{(1)} = \log_2 \chi(G)$ to denote the single-instance rate. 


\subsubsection{\bf Variable Length Classical Code (VLC) Setting} \label{subsubsec:VLC_setting}


Let $\{0, 1, \dots, D-1\}$ be a $D$-ary alphabet from which the symbols in a variable length code are chosen, and let $ \{0,\dots, D-1\}^*$ denote the set of all finite-length strings from $\{0, 1, \dots, D-1\}$. A graph $H$  is said to be an induced subgraph of $G$ if its vertex set $V(H)$ is a subset of $V(G)$  and edge set $E(H)$ consists of edges in $E(G)$ whose two endpoints are in $V(H)$. For $\bfy$ of length-$m$, let $G^{(m)}_\bfy$ denote the induced subgraph of $G^{(m)}$ consisting of the subset of $\bfx$ vertices that occur along with $\bfy$ with strictly positive probability, i.e., $V(G^{(m)}_\bfy) = \{\bfx:\bfx\text{ such that }p_{\bfX\bfY}(\bfx,\bfy)>0\}$ and $E(G^{(m)}_\bfy) = \{(\bfx_1,\bfx_2):(\bfx_1,\bfx_2)\in E(G^{(m)})\text{ and }\bfx_1,\bfx_2\in V(G^{(m)}_\bfy)\}$. 



\begin{definition}{\it Non-singular Code.} \label{def:ns_code}
A $m$-instance variable length code consists of an encoding function $\phi_{e}:V(G^{(m)}) \mapsto \{0,\dots, D-1\}^*$ such that $\phi_e$ is a coloring of $G^{(m)}_\bfy$ for all $\bfy$. This means that if $(\bfx_1,\bfx_2)$ is an edge in $G^{(m)}_\bfy$, then $\phi_{e}(\bfx_1) \neq \phi_{e}(\bfx_2)$. In this case, we call the variable length code, non-singular.
\end{definition} 

It is not too hard to see that a non-singular variable length code as described above corresponds to a coloring of the $m$-instance graph $G^{(m)}$. Indeed, $(\bfx_1, \bfx_2) \in E(G^{(m)})$ then there exists a $\bfy$ such that $(\bfx_1, \bfx_2) \in E(G^{(m)}_\bfy)$ and hence $\phi_e(\bfx_1) \neq \phi_e(\bfx_2)$. On the other hand, suppose that we consider a variable length non-singular source code as defined in \cite{CoverT05} (Chapter 5), that encodes a given coloring of $G^{(m)}$. Then, by setting $\phi_e(\bfx)$ to be the encoding of its color, we observe that we arrive at a valid non-singular variable length code for our problem. 

\begin{definition}{\it Prefix-free code.} \label{def:pf_code}
Consider a $m$-instance non-singular variable length code with encoding function $\phi_e: V(G^{(m)}) \mapsto \{0, \dots, D-1\}^*$.
Let $\phi_e(G^{(m)}_\bfy) \triangleq \{\phi_e(\bfx):\text{~where~} \bfx \in V( G^{(m)}_\bfy)\}$ denote the subset of all encodings corresponding to vertices in $G^{(m)}_\bfy$, i.e., $\phi_e(G^{(m)}_\bfy)$ is the image of the map $\phi_e(\cdot)$ when $\bfx$ is restricted to $V( G^{(m)}_\bfy)$. Suppose that for all $\bfy$ we have the following property. Let $\zeta_i \in \phi_e(G^{(m)}_\bfy), i =1,2$ such that $\zeta_1 \neq \zeta_2$. Then, we have that $\zeta_1$ is not a prefix of $\zeta_2$. In this case, we call the variable length code, prefix-free. 
\end{definition}  

\begin{remark}\label{remark:AlonO_prefix-free}
In \cite{AlonO_96}, the authors define a prefix-free variable length code for the problem of source coding with side information, where they enforce that $\phi_e(\bfx_1)$ is not a prefix of $\phi_e(\bfx_2)$ if $(\bfx_1, \bfx_2) \in E^{(m)}$. As discussed in Section \ref{sec:VLC_rate_graph_entropy}, Example \ref{example:C5_VLC}, this definition is not appropriate for the general function computation problem.
%
\end{remark}

\textbf{Extension of a code:} Consider a prefix-free variable length code for a given value of $m$. Then, this code can be extended to all multiples of $m$ by simply concatenating the encodings; this is referred to as an extension of the original code (Chapter 5 of \cite{CoverT05}). For example, the encoding for a length $2m$ vector $[\bfx~ \bfx']$ will be the concatenation $\phi_e(\bfx)\phi_e(\bfx')$. This implies, for instance, that a prefix-free variable length code for $m=1$ can be employed as a code for arbitrary $m$ by extension, such that Bob can instantaneously decode $\phi_e(x_i), i = 1, \dots, m$ from the string $\phi_e(x_1)~\phi_e(x_2)~\dots~\phi_e(x_m)$.

On the other hand, a non-singular code for a given $m$ cannot necessarily be extended so that it is even uniquely decodable \cite{CoverT05} (see Example \ref{example:C5_VLC} in Section \ref{sec:VLC_rate_graph_entropy}). For this reason, henceforth, we will only consider prefix-free variable length codes and define rates accordingly.



\begin{definition}{\it Optimal Rate in VLC setting.} \label{def:VLC_rate}
Consider a prefix-free variable length code for our problem. Let  $|\phi_e(\bfx)|$ denote the length of the encoded $D$-ary  string corresponding to $\phi_e(\bfx)$. Then, the rate of this code is defined as $R_{VLC}^{(m)} \triangleq \frac{1}{m}\sum_{\bfx}p_{\bfX}(\bfx)|\phi_e(\bfx)|$ where $p_{\bfX}(\cdot)$ is the marginal distribution of $\bfX$.
The optimal rate is defined as
\begin{align}
    R_{VLC} &=  \inf_m \min_{\phi_e} R^{(m)}_{VLC},
\end{align}
where $\min_{\phi_e} R^{(m)}_{VLC}$ is the rate of the optimal $m$-instance encoding. 
\end{definition}

\begin{remark}\label{remark:VLC_depends_on_p_X}
    Unlike the fixed length code rate, the variable length code rate depends on the marginal distribution $p_X(\cdot)$ of the random variable $X$.  
\end{remark}
\subsubsection{\bf Quantum Setting} 
The following notions are standard within quantum information theory (see \cite{wilde_17}). A quantum bit (qubit) or more generally a qudit represented by $\ket{\psi}$ (Dirac's bra-ket notation) is a unit-norm column vector in a complex Hilbert space $\calH$. A closed quantum system evolves only via unitary transformations $U$, i.e., $\ket{\psi} \mapsto U \ket{\psi}$, where $U U^\dagger = U^\dagger U = I$ ($\dagger$ denotes conjugate transpose). Two qudits $\ket{\psi_1} \in \calH_A$ and $\ket{\psi_2} \in \calH_B$ lie in the (composite) tensor product space $\calH_A \otimes \calH_B$. Quantum states can be ``pure'' in which case they can be represented as vectors. Quantum states can also be mixed, e.g., pure state $\ket{\psi_x}$ may be prepared with probability $p_X(x)$. For a Hilbert space $\calH$, let $\calL(\calH)$ denote the set of linear operators from $\calH$ to itself; treated as matrices, these are square matrices. When states are mixed, the description of the state is given by a density operator in $\calL(\calH)$,  $\rho_X = \sum_{x \in \calX} p_X(x) \ket{\psi_x} \bra{\psi_x}$. A density operator (or quantum state) is a positive Hermitian operator\footnote{A positive operator $A$ over a complex Hilbert space, denoted $A \succeq 0$ is Hermitian, i.e., $A = A^\dagger$ and has non-negative eigenvalues. Treated as a matrix, it is a Hermitian, positive semi-definite matrix.} with unit-trace. The density operator of a pure state $\ket{\psi_x}$ is written as $\ket{\psi_x}\bra{\psi_x}.$ A quantum measurement is a fundamental operation by which information can be extracted from a quantum system. While there are many different measurement formalisms, here we discuss the positive operator-valued measure (POVM) measurements which are relevant to our work. These are specified by a set of positive operators $\{\Lambda_j\}_{j=1}^{k}$ where $\Lambda_j \succeq 0$ and $\sum_{j=1}^k \Lambda_j = I$ 
 ($I$ denotes the identity operator). The probability of obtaining measurement value $i$ is given by $\Tr \left(\Lambda_i \rho_X \right)$ for $i \in [k]$. Quantum states $\rho$ and $\sigma$ are said to be orthogonal, denoted $\rho \perp \sigma$ if $\Tr (\rho^{\dagger} \sigma) = 0$. This is equivalent to $\rho$ and $\sigma$ having orthogonal supports\footnote{The support of an operator is the orthogonal complement of its kernel. For Hermitian operators (as we consider) the support is its image.}. Two  quantum states $\rho$ and $\sigma$ are said to be perfectly distinguishable if there exists a POVM $\{\Lambda_j\}_{j=1}^{k}$ such that $\Tr \left(\Lambda_i \rho \right)\Tr \left(\Lambda_i \sigma \right) = 0$ for $i \in [k]$, i.e., there exists a POVM $\{\Lambda_j\}_{j=1}^{k}$ which distinguishes $\rho$ and $\sigma$ with zero-error. Such a POVM exists if and only if $\rho$ and $\sigma$ are orthogonal \cite{MEDEIROSD05}. Note that for this case $k=2$ suffices. However, in our discussion later, we need to perfectly distinguish multiple states.

 

Suppose that Alice and Bob have $m$-length sequences $\bfx$ and $\bfy$ and Alice communicates to Bob through an error-free quantum channel that supports the transmission of operators on $\calH$. The quantum code for our problem setting operates as follows.
\begin{itemize}[leftmargin=2mm]
    \item Alice picks a quantum state $\rho_\bfx \in \calL(\calH)$ (where $\dim \calH = d$), that depends on the sequence $\bfx$ and transmits it to Bob.  
    \item Bob performs a POVM measurement\cite[pg. 134]{wilde_17} on the received state. The POVM, which is dependent on $\bfy$, is specified by $\{\Lambda^\bfz_{f(\cdot,\cdot), \bfy}\}_{\bfz \in \calZ^m}$, where the indices $\bfz$ are measurement results and the matrices $\Lambda^\bfz_{f(\cdot,\cdot), \bfy}$ are positive semi-definite matrices that sum to the identity, i.e., $\Lambda^\bfz_{f(\cdot,\cdot), \bfy} \succeq 0$, for $\bfz \in \calZ^m$, and $\sum_{\bfz \in \calZ^m}  \Lambda^\bfz_{f(\cdot,\cdot), \bfy} = I$.
    

\end{itemize}
The code is deemed successful if $z_i = f(x_i, y_i)$, for $i \in \{ 1, \dots, m\}$ for all possible pairs $(\bfx, \bfy)$. 
 Alice's rate of transmission is defined as $\frac{1}{m} \log_2 d$.

It can be shown ({\it cf.} Section \ref{sec:orth_rep}), that this code is zero-error if and only if for $(\bfx, \bfx') \in G^{(m)}$ then $\rho_\bfx \perp \rho_{\bfx'}$. If this condition is satisfied, depending on the $\bfy$ sequence observed by Bob, it can be shown that he can prepare the measurement such that   the function value is recovered with zero error. Furthermore, Alice needs to find the suitable set of states in as small of a dimension $d$ as possible, so that the task can be achieved by transmitting the fewest number of quantum bits. The quantum rate is defined formally in Section \ref{sec:orth_rep} (Theorem \ref{thm:witsenhausen_quantum_rate}). As in the classical case, we make a distinction between the single-instance rate and the asymptotic rate.

We point out that the classical FLC code can be considered as an instance of the proposed quantum code, simply by operating in a large enough vector space and labeling the nodes of the $f$-confusion graph by canonical basis vectors (binary vectors with all components zero except one). However, considering general unit-norm vectors provides much more flexibility and therefore the dimension $d$ in which we need to operate in, can be lower; concrete examples can be found ({\it cf.} Section \ref{sec:analysis_bahavior_quantum_advantage}).

We emphasize that there are significant difficulties with defining variable length quantum codes. In particular, as pointed out in \cite{BraunsteinFGH98}, determining the quantum message boundaries will itself require measurement, which will, in general, destroy superposition. Thus, in this work we do not consider variable length quantum codes. 

\subsection{Related work}\label{subsec:related_work}
A natural conditional entropy defined on the confusion graph provides the precise rate characterization for function computation with side information when the computation is only required to be lossless, i.e., the distortion is asymptotically vanishing; it is not a zero-error scenario \cite{OrlitskyR_01}. The Slepian-Wolf theorem \cite{slepian1973noiseless} addresses the case of source coding with side information.

The exhaustive survey paper \cite{KornerO_98} overviews several aspects of classical zero-error information theory. We first discuss the case of  classical source coding with side information. 
Witsenhausen \cite{Witsenhausen_76} considered the usage of fixed length codes in this setting and showed the optimal rate could be analyzed by considering the chromatic number of the corresponding $f$-confusion graph. We note here that in general, a single-letter characterization of the optimal rate is not available. 
Alon and Orlitsky  \cite{AlonO_96} considered variable length codes and related the rate characterization to an appropriately defined chromatic entropy of $f$-confusion graph. It is well-known that the multiple instance rate can be much less than the single instance rate \cite{KornerO_98}. For instance, it can be observed that the $f$-confusion graph in the pentagon example ({\it cf.} Fig. \ref{fig:pentagon}) can be colored with three colors so that the rate can be $\log_2 3$. However, if we allow computation over two instances, then it can be shown that the corresponding graph can be colored with five colors so that the rate equals $\frac{1}{2} \log_2 5$ \cite{Witsenhausen_76,lovasz1979shannon}. It turns out that this rate is optimal for this problem; this follows from a celebrated result of Lov{\'a}sz\cite{lovasz1979shannon}.  Orlitsky \cite{orlitsky_worst_case90} also considered the case of interactive communication between Alice and Bob. For the variable length codes, \cite{KoulgiTRR_03} expresses the optimal rate as complementary graph entropy of $f$-confusion graph. While this quantity is still hard to compute in general, \cite{CharpenayTR23ISIT} consider certain cases where it can actually be computed. Reference \cite{CharpenayTR23ITW}, considers a zero-error function computation model where side information is also available at the encoder. For this model, under specific conditions on the joint p.m.f. and Alice's side information, they derive a single-letter characterization of the optimal rate.

Function computation has also been studied in the scenario where both $X$ and $Y$ are transmitted by remote entities and Bob seeks to compute $f(X,Y)$ with asymptotically vanishing error. The work of \cite{HanK87} determined a necessary and sufficient condition on $f(\cdot, \cdot)$ such that the optimal rate region coincides with the Slepian-Wolf region. In these scenarios, structured codes often outperform random codes, as pointed out in \cite{KornerM79} for the problem of recreating the modulo-2 sum of binary sources. The work of \cite{NazerG07, wilson_etal10} consider function computation over multiple-access channels. More recently, function computation over quantum many-to-one networks has been considered in \cite{AllaixLYPHCJ24}.


There have been a few works that consider quantum generalizations of classical zero-error problems. The work of \cite{cubitt2010improving} considers communicating with zero error over a noisy discrete memoryless channel (DMC), denoted $W$ in the presence of shared entanglement between the source and the terminal. Let $C_0(W)$ and $C_{SE}(W)$ denote the zero-error capacity of the channel in the absence and in the presence of shared entanglement, respectively. Their work demonstrates that there exist DMCs such that $C_0(W) < C_{SE}(W)$; these are constructed by leveraging Kochen-Specker sets \cite{held2009kochen}. This is rather interesting since the work of Bennett et al. \cite{bennett2002entanglement} shows that the presence of shared entanglement does not increase the capacity of a classical channel when considering asymptotically error-free communication. Following this work, there were several works that considered entanglement-assisted zero-error source-channel coding. The work of \cite{BrietBL_15} considers shared entanglement, but the sources and channels are still classical. The work of Stahlke \cite{stahlke2015quantum} considers quantum sources and quantum channels. However, none of these works consider general functions.

Prior work that is most closely related to our work includes \cite{BuhrmanCW_98, GavinskyKKRd_07, Bar-YossefJK08}. For function computation, these works demonstrate that there exists an exponential separation between the classical and the quantum rate. However, these results are for the single-instance case, i.e., they do not discuss the asymptotic rate that considers multiple instances. As discussed above, the classical and quantum rates exhibit interesting behavior when considering multiple instances. This is the central goal of our work.  In the recent work of \cite{GuptaSXCA_23}, they consider a specific function that for which they do consider the asymptotic behavior of the rate. In contrast, in our work, we consider general functions and consider the different possible behaviors of the single-instance and asymptotic rates.

\subsection{Contributions and Organization}
For the source coding with side information problem, it is not hard to see that $G^{(m)} = G^{
\boxtimes m}$ ({\it cf.} Definition \ref{defn:confusion_graph} and Definition \ref{defn:AND_product_OR_product}). In our setting, for a given $f$-confusion graph $G$, we show that the $m$-instance $f$-confusion graph $G^{(m)}$is sandwiched between the strong product ($G^{\boxtimes m}$) and OR-product ($G^{\lor n}$), i.e., $G^{\boxtimes m} \subseteq G^{(m)} \subseteq G^{\lor n}$ ({\it cf.}   Definition \ref{defn:AND_product_OR_product}). Our first contribution is a characterization of necessary and sufficient conditions on the function $f(\cdot, \cdot)$ and joint p.m.f. $p_{XY}(\cdot,\cdot)$ such that $G^{(m)}$ equals  $G^{\boxtimes m}$ or $G^{\lor m}$. 

We use this condition to study classes of $f$-confusion graphs that reveal interesting behavior of the classical (FLC and VLC) and quantum rates for our problem. As discussed above, e.g., when considering the classical FLC rate, a single-instance rate is given by $\chi(G)$. Unfortunately, the behavior of the sequence of chromatic numbers $\chi(G^{(2)}), \chi(G^{(3)}), \dots, \chi(G^{(m)}), \dots$ is not very well understood. However,  if $G^{(m)} = G^{\lor m}$, then a single-letter characterization is available. When considering the classical VLC rate, the $m$-instance rate is characterized by the so-called chromatic entropy of $G^{(m)}$, and a single letter characterization in terms of K\"orner's graph entropy \cite{Korner_73} exists if $G^{(m)} = G^{\lor m}$. The quantum rate for the $m$-instance case can be expressed in terms of the orthogonal rank of the complement of $G^{(m)}$ (see Section \ref{sec:orth_rep}). We define the quantum advantage in the single or asymptotic settings as the ratio of the single-instance or asymptotic classical and quantum rates, respectively.

Our second contribution is to present examples of $f$-confusion graphs that demonstrate a multitude of possible behaviors of the quantum and classical (FLC/VLC) rates when considering the single-instance rate and the asymptotic rate. 
%

When considering FLCs, it is not hard to see that the quantum rate will always be lower. We demonstrate that there exist (i) problems where there is no quantum advantage in either single-instance or asymptotic rates, (ii)  problems where there are quantum advantages in both single-instance and asymptotic rates, and (iii) problems where there is a quantum advantage in the single-instance rate but no quantum advantage in the asymptotic rates. The situation is summarized compactly in Table \ref{table:Main_contribution}.


The comparison with VLC is more complicated. It turns out that the VLC rate can be lower than the quantum rate in some cases; we refer to this as the VLC advantage. Our main findings are that the quantum advantage continues to hold for certain problems. Nevertheless, there are specific problems that demonstrate the VLC rate (single-instance and asymptotic) can be smaller than the corresponding quantum rate. Section \ref{sec:VLC_rate_graph_entropy} has the complete details. We summarize our findings in Table \ref{table:VLC}.

\begin{table}[t]   
    \centering 
    \caption{Quantum advantage within different FLC scenarios} 
    \label{table:Main_contribution}
    \begin{tabular}{|c|c|c|}
        \hline
        Single-instance advantage & Multiple-instance advantage & $f$-Confusion Graphs \\
        \hline
        Yes    & Yes   & $G_{13}^{\lor m} \text{({\it cf.} Sec. \ref{subsubsec:G13})},H_n^{\boxtimes m},H_n^{\lor m} \text{({\it cf.} Sec. \ref{subsubsec:Hn})}$ \\
        \hline
        Yes    & No   & $\frL(G_{13})^{\lor m},\frL(G_{13})^{\boxtimes m} \text{({\it cf.} Sec. \ref{sec:line_graph_eg})}$ \\
        \hline
        No    & Yes   & Unknown \\
        \hline
        No    & No   & $C_5^{\lor m},C_5^{\boxtimes m}$ ({\it cf.} Sec. \ref{subsec:C5})\\
        \hline
    \end{tabular}
    \vspace{5mm}
    \caption{Type of advantage within different VLC scenarios}\label{table:VLC}
    \begin{tabular}{|c|c|c|c|}
        \hline
        Single-instance advantage & Multiple-instance advantage & Is $p_X(x)$ uniform? & $f$-Confusion Graphs   \\
        \hline
        VLC    & VLC & No  & $K_3^{\boxtimes m}  \text{({\it cf.} Sec. \ref{sec:VLC_k3})}$   \\ 
        \hline
          VLC  & Equal   & Yes  & $C_5^{\boxtimes m} \text{({\it cf.} Sec. \ref{sec:VLC_C5_AND})}$ \\
        \hline
         Quantum   & Equal & Yes  & $C_5^{\lor m} \text{({\it cf.} Sec. \ref{sec:VLC_C5_OR})}$ \\
        \hline
        Quantum    & Quantum & Yes  & $G_{13}^{\lor m} \text{({\it cf.} Sec. \ref{sec:VLC_G13})}, H_n^{\boxtimes m},H_n^{\lor m} \text{({\it cf.} Sec. \ref{sec:VLC_Hn})}$ \\
        \hline
        Quantum    & VLC & Yes  & $\frL(G_{13})^{\lor m}$ ({\it cf.} Sec. \ref{sec:VLC_line_graph})\\
        \hline
    \end{tabular}
\end{table}
This paper is organized as follows. Preliminary definitions and results from graph theory appear in Section \ref{sec:prelim}. Section \ref{sec:struc_conf_graph} presents the structural characterization of the $m$-instance $f$-confusion graph and conditions under which it equals either $G^{\boxtimes m}$ or $G^{\lor m}$. We discuss orthogonal representations of graphs in Section \ref{sec:orth_rep}, and provide a discussion of various examples that demonstrate the behavior of the quantum advantage compared with FLC in Section \ref{sec:analysis_bahavior_quantum_advantage}. We introduce chromatic entropy and K\"orner's graph entropy in Section \ref{sec:VLC_rate_graph_entropy}, and   characterize  the asymptotic VLC rate in terms of these entropies. In Section \ref{sec:analysis_bahavior_quantum_advantage_vlc}, we present a comparison of VLC rate and quantum rate for several examples. We summarize the conclusions of the paper in Section \ref{sec:conclusion}. 



 
\section{Preliminaries} 
\label{sec:prelim}
Let $\bbN,\bbN_{\ge1},\bbC$ be the set of natural numbers, the set of  positive natural numbers   and the set of complex numbers respectively. In our discussion, we consider finite directed and undirected graphs. 
Our discussion will involve directed graphs only in Section \ref{sec:line_graph_eg}. 
Thus, unless otherwise specified, henceforth when we refer to a graph, we are considering undirected graphs. The following definitions are standard within graph theory and can be found, e.g., in \cite{west_2000}. A simple graph is a graph that has no loops  (edges connecting a vertex to itself) and no multiple edges (more than one edge between the same two vertices).  A graph $H=(V',E')$ is said to be a subgraph of $G = (V,E)$ if $V'\subseteq V$ and $E'\subseteq E'$. We say that $G$ is a spanning subgraph of $H$, denoted by $G\subseteq H$  if $V(G) = V(H)$ and $E(G)\subseteq E(H)$. Suppose that $S$ is a subset of $V$. The induced subgraph $G[S]$ is the graph whose vertex set is $S$ and edge set consists of those edges whose endpoints are in $S$.  For a graph $G$, $\overline{G}$ denotes the complementary graph. For $G$, $\alpha(G)$ denotes its independence number (the size of its largest independent set), $\omega(G)$ denotes its clique number (the size of its largest clique) and $\chi(G)$ its chromatic number (minimum number of colors used in a valid vertex coloring).   We will use the following known results (proved in Appendix \ref{app:appendix_A} for completeness) for graphs $G,H$.

\begin{proposition}\label{prop:chi}
\begin{enumerate} [label=(\alph*)]
    \item $\overline{G\boxtimes H} = \overline{G}\lor \overline{H}$.\label{prop:de_morgan}
    \item $\alpha(G\lor H)=\alpha(G)\alpha(H) \le \alpha(G\boxtimes H)$. \label{prop:alpha_AND_OR_product}
    \item $\chi(G\boxtimes H)\le \chi(G\lor H)\le \chi(G)\chi(H)$.\label{prop:chi_product}
\end{enumerate}
\end{proposition}
The fractional chromatic number is defined as follows\cite{ScheinermanU97FGT}.
\begin{definition}\label{defn:b-fold_coloring} 
%

Suppose that there is an underlying set $[a]=\{1,\dots, a\}$ and that each vertex $v$ is associated with a subset $S(v)\subseteq [a]$ with $|S(v)| = b$ such that if $u,v\in V(G)$ are adjacent, then $S(u)$ and $S(v)$ are disjoint, i.e., $S(u) \cap S(v) = \emptyset$. In this case, we say that $G$ has a $a:b$ coloring.
%
%
The least $a$ for which $G$ has a $a:b$ coloring is called the $b$-fold chromatic number of $G$, denoted $\chi_b(G)$. The fractional chromatic number $\chi_f(G)$ is defined as
\begin{align*}
    \chi_f(G):= \inf_b \frac{\chi_b(G)}{b}.
\end{align*}
\end{definition}
 Since $\chi_1(G) = \chi(G)$, we have 
\begin{equation}\label{eqn:chi>=chi_f}
    \chi(G) \ge \chi_f(G).
\end{equation}


 \begin{proposition}  \label{prop:chi_frac}
Let $G,H$ be simple graphs. Then,
\begin{enumerate} [label=(\alph*)]
    \item $\chi_f(G\lor H) = \chi_f(G)\chi_f(H)$ (Corollary 3.4.2 of \cite{ScheinermanU97FGT}).\label{prop:OR_chi_f}
    \item $\chi_f(G) = \inf_m \chi(G^{\lor m})^{1/m}$ (Corollary 3.4.3 of \cite{ScheinermanU97FGT}).\label{prop:OR_Witsenhausen_rate}
\end{enumerate} 
 \end{proposition} 

 
\section{Structure of $f$-confusion graph}
\label{sec:struc_conf_graph} 

\begin{table*}[t]
\centering 
\begin{minipage}{.32\textwidth}
  \centering
  \begin{tabular}{c|c|ccccc}
    \multicolumn{2}{c|}{} & \multicolumn{5}{c}{\( x \)} \\
    \cline{3-7}
    \multicolumn{2}{c|}{\multirow{-2}{*}{\( \tilde f(x,y) \)}} & 1 & 2 & 3 & 4 & 5 \\
    \hline
    \multirow{5}{*}{\( y \)} 
    & 1 & 1 & 0 & * & * & * \\
    & 2 & * & 1 & 0 & * & * \\
    & 3 & * & * & 1 & 0 & * \\
    & 4 & * & * & * & 1 & 0 \\
    & 5 & 0 & * & * & * & 1 
  \end{tabular}
  \caption{Function $\tilde f$}
  \label{table:simple_eg1}
\end{minipage}%
\begin{minipage}{.32\textwidth}
  \centering
  \begin{tabular}{c|c|ccccc}
    \multicolumn{2}{c|}{} & \multicolumn{5}{c}{\( x \)} \\
    \cline{3-7}
    \multicolumn{2}{c|}{\multirow{-2}{*}{\( \tilde g(x,y) \)}} & 1 & 2 & 3 & 4 & 5 \\
    \hline
    \multirow{5}{*}{\( y \)} 
    & 1 & 1 & 0 & 1 & * & * \\
    & 2 & * & 1 & 0 & 1 & * \\
    & 3 & * & * & 1 & 0 & 1 \\
    & 4 & 1 & * & * & 1 & 0 \\
    & 5 & 0 & 1 & * & * & 1 
  \end{tabular}
  \caption{Function $\tilde g$}
  \label{table:simple_eg2}
\end{minipage} 
\begin{minipage}{.32\textwidth}
  \centering
  \begin{tabular}{c|c|ccccc}
    \multicolumn{2}{c|}{} & \multicolumn{5}{c}{\( x \)} \\
    \cline{3-7}
    \multicolumn{2}{c|}{\multirow{-2}{*}{\(\tilde  h(x,y) \)}} & 1 & 2 & 3 & 4 & 5 \\
    \hline
    \multirow{5}{*}{\( y \)} 
    & 1 & 1 & 0 & 1 & * & * \\
    & 2 & * & 1 & 0 & * & * \\
    & 3 & * & * & 1 & 0 & * \\
    & 4 & * & * & * & 1 & 0 \\
    & 5 & 0 & * & * & * & 1 
  \end{tabular}
  \caption{Function $\tilde h$}
  \label{table:simple_eg3}
\end{minipage}%
\end{table*}  

In this section, we explore the dependence of the $m$-instance $f$-confusion graph $G^{(m)}$ on the underlying joint p.m.f. $p_{XY}(\cdot, \cdot)$ and the function $f(\cdot, \cdot)$. Our main result is a compact characterization of the conditions under which $G^{(m)}$ equals either $G^{\boxtimes m}$ or $G^{\lor m}$.

Towards this end, consider the following three scenarios, all of which have identical single instance $f$-confusion graphs but very different $m$-instance $f$-confusion graphs. These examples appear in  Table \ref{table:simple_eg1}-\ref{table:simple_eg3}. The functions are denoted $\tilde f,\tilde g$ and $\tilde h$. The rows corresponding to \( y \) and columns corresponding to \( x \). Within the table, numerical entries represent the value of the function with input $(x,y)$ and $*$ entries indicate that $p_{XY}(x,y)=0$. It can be verified that $G_{\tilde f}=G_{\tilde g}=G_{\tilde h}=C_5$ (pentagon graph). However, we have distinct $m$-instance $f$-confusion graphs: $G_{\tilde f}^{(m)}=C_5^{\boxtimes m}, G_{\tilde g}^{(m)}=C_5^{\lor m},C_5^{\boxtimes m}\ne G_{\tilde h}^{(m)}\ne C_5^{\lor m}$. See Appendix \ref{app:2_fold_graphs} for two-fold $f$-confusion graphs of $\tilde f,\tilde g, \tilde h$. 

In what follows, we assume the following condition on the joint p.m.f.; this rules out uninteresting cases.
\begin{assumption}\label{assume:irreducible}
    $\forall x \in \calX, \exists y\in \calY\text{ such that  }p_{XY}(x,y)>0.$
\end{assumption} 
For a $f$-confusion graph $G$, the following proposition states that $G^{(m)}$ is sandwiched between $G^{\boxtimes m}$ and $G^{\lor m}$. The proof appears in Appendix \ref{app:pf_prop5}.
\begin{proposition}\label{prop:AND_power_m-instance_graph_OR_power}
  $G^{\boxtimes m} \subseteq G^{(m)} \subseteq G^{\lor m}$ for $m\ge 1$.
\end{proposition}
\begin{remark}
Combining Proposition \ref{prop:AND_power_m-instance_graph_OR_power}  and Proposition \ref{prop:chi_frac} \ref{prop:OR_Witsenhausen_rate}, we note  $\log_2\chi_f(G)$ is an upper bound on the asymptotic FLC rate for an arbitrary confusion graph. 
%
\end{remark}
Consider two trivial cases. If $G$ is edge-less, i.e., $G$ is a set of isolated vertices, then $G^{\lor m}$ is edge-less and thus $G^{( m)}$ is edge-less. Similarly, if $G$ is complete, then $G^{\boxtimes m}$ is complete and thus $G^{( m)}$ is complete. In both cases, we have $G^{\boxtimes m}= G^{(m)}=G^{\lor m}$.  If $G$ is neither complete nor edge-less, we say that $G$ is nontrivial. Accordingly, we now investigate conditions under which $G^{(m)} = G^{\boxtimes m}$ and $G^{(m)} = G^{\lor m}$ when $G$ is nontrivial.

\begin{observation}
Let $G$ be a nontrivial $f$-confusion graph. Then, for distinct $x, x'\in V(G)$, $(x,x')\notin E(G)$   can be due to either one of the  following mutually exclusive conditions.
\begin{itemize}[leftmargin=5mm]
    \item $C1:$ there is no  $y \in \calY$  such that  $p_{XY}(x,y)p_{XY}(x',y)> 0$.
    \item  $C2: \exists y$ such that $p_{XY}(x,y)p_{XY}(x',y)> 0$ but for all such $y$,  we have $f(x,y) = f(x',y)$.
\end{itemize}    
\end{observation}
Examples of $C1,C2$ can be observed in the functions $\tilde f,\tilde g,\tilde h$ from Table \ref{table:simple_eg1}-\ref{table:simple_eg3}. For instance, consider $x=1,x'=3$ in Table \ref{table:simple_eg1}. Since there does not exist $y$ such that $p_{XY}(x,y)p_{XY}(x',y)>0$, $x,x'$ are non-adjacent due to $C1$. In contrast, for the same $x,x'$ in Table \ref{table:simple_eg2}, there exists $y$ such that $p_{XY}(x,y)p_{XY}(x',y)>0$. For such $y$, i.e. $y=1$, we have   $\tilde f(x,y)=\tilde f(x',y)=1 $.   $x,x'$ are non-adjacent due to $C2$.  Upon inspection, we can conclude that all non-adjacent $(x,x')$s in Table \ref{table:simple_eg1} are due to $C1$.  Similarly, all non-adjacent $(x,x')$s in Table \ref{table:simple_eg2} are due to $C2$. In contrast, in Table \ref{table:simple_eg3} non-adjacent $(x,x')$ occur both due to $C1$ and $C2$. 



\begin{theorem}\label{thm:theorem4}
%
Let $G$ and $G^{(m)}$ be $f$-confusion graphs over one-instance and $m$-instances respectively for a function $f$ and source pair $(X,Y)$ with joint p.m.f. $p_{XY}(\cdot, \cdot)$. Suppose that $G$ is nontrivial. Then,
    \begin{enumerate}[label=(\alph*)]
\item  $G^{(m)} = G^{\boxtimes m}$ for all $m\in\bbN_{\ge 1}$ if and only if all $x\ne x'$ and $(x,x')\notin E(G)$ are due to condition $C1$.\label{thm:theorem4_part_a}
    \item $G^{(m)} = G^{\lor m}$ for all $m\in\bbN_{\ge 1}$ if and only if all $x\ne x'$ and $(x,x')\notin E(G)$ are due to condition $C2$. \label{thm:theorem4_part_b}
\end{enumerate}
\end{theorem}

\begin{proof}[Proof of Theorem \ref{thm:theorem4} (a)]
$\Rightarrow:$ We show the contrapositive, i.e., if there exists $(x,x')\notin E(G)$ such that  it is due to condition $C2$, then $G^{(m)}\ne G^{\boxtimes m}$. In particular, we show $G^{(2)}\ne G^{\boxtimes 2}$ by showing that $E(G^{(2)})\ne E(G^{\boxtimes 2})$.

Let $x_1\ne x_1'$ and $(x_1,x_1')\notin E(G)$ due to condition $C2$ and $(x_2,x_2')\in E(G)$. Such $(x_2,x_2')$ exists as $G$ is nontrivial. Denote $\bfx=[x_1,x_2],\bfx'=[x_1',x_2']$.    Note that $(\bfx,\bfx')\notin E(G^{\boxtimes 2})$ as $(x_1,x_1')\notin E(G)$.

By condition $C2$, there exists $y_1$ such that $$p_{XY}(x_1,y_1)p_{XY}(x_1',y_1)>0\text{ and }f(x_1,y_1)= f(x_1',y_1). $$
Moreover, $(x_2,x_2')\in E(G)$ implies the existence of $y_2$ such that  $$
p_{XY}(x_2,y_2)p_{XY}(x_2',y_2)>0\text{ and } f(x_2,y_2) \ne f(x_2',y_2). 
$$
Denote $\bfy=[y_1,y_2]$.  Then it follows that  
\begin{align*}
    &\prod_{i=1}^2p_{XY}(x_i,y_i)p_{XY}(x_i',y_i)>0,\text{ and}\\
    &f^{(2)}(\bfx,\bfy)\ne f^{(2)}(\bfx',\bfy),
\end{align*}because $f(x_2,y_2)\ne f(x_2',y_2)$. Therefore, we have $(\bfx,\bfx')\in E(G^{(2)})$, which implies $E(G^{(2)})\ne E(G^{\boxtimes 2})$. 

$\Leftarrow:$ By Proposition \ref{prop:AND_power_m-instance_graph_OR_power} and the fact that $V(G^{\lor m})  = V(G^{\boxtimes m})$, it suffices to show $E(G^{(m)}) \subseteq E(G^{\boxtimes m})$. We prove it by showing $E(\overline{G^{(m)}}) \supseteq E(\overline{G^{\boxtimes m}}).$

Now let $\bfx\ne \bfx'$ and $ (\bfx,\bfx') \notin E(G^{\boxtimes m})$ be arbitrary. Then, there is an index $i$ such that   $(x_i,x_i')\notin E(G)$. Since all $(x,x')\notin E(G)$ are due to condition $C1$,   $p_{XY}(x_i,y_i)p_{XY}(x_i',y_i)=0$ for all $y_i\in \calY$, which implies  $p_{\bfX\bfY}(\bfx,\bfy)p_{\bfX\bfY}(\bfx',\bfy)= 0$ for all choices of $\bfy$. Thus, we have $(\bfx,\bfx') \notin E(G^{(m)})$. Since $(\bfx,\bfx') \notin E(G^{\boxtimes m})$ is arbitrary, this implies that $E(G^{(m)}) \subseteq E(G^{\boxtimes m})$.
\end{proof} 
\begin{proof}[Proof of Theorem \ref{thm:theorem4} (b)] 
$\Rightarrow:$ Since $G$ is nontrivial, there  exists $(x_1,x_1')\in E(G), (x_2,x_2')\notin E(G)$ and $x_2\ne x_2'$. Denote $\bfx=[x_1,x_2],\bfx'=[x_1',x_2']$ and note $(\bfx,\bfx')\in G^{\lor 2}$. By the assumption that $G^{\lor m} = G^{(m)}$ for $m\ge 1$ and setting $m=2$, there exists $\bfy = [y_1,y_2]$ such that  $p_{\bfX\bfY}(\bfx,\bfy)p_{\bfX\bfY}(\bfx',\bfy)>0$, which implies $p_{XY}(x_2,y_2)p_{XY}(x_2',y_2)>0$. Therefore, $(x_2,x_2')\notin E(G)$ is not due to condition $C1$. 

$\Leftarrow:$ By Proposition \ref{prop:AND_power_m-instance_graph_OR_power} and the fact $V(G^{\lor m})=V(G^{\boxtimes m})$, it suffices to show $E(G^{\lor m  })\subseteq E(G^{(m)})$.  Let $(\bfx,\bfx')\in E(G^{\lor m})$. Without loss of generality (w.l.o.g.), we may assume $\bfx,\bfx'$ is such that  
$$\begin{cases}
    (x_i,x_i')\in E(G)\text{~for~}i\in\{1,\dots,k_1\},\\
    (x_i,x_i')\notin E(G)\text{~and~} x_i\ne x_i'\text{~for~}i\in\{k_1+1,\dots,k_2\}, \text{~and}\\
    x_i=x_i'\text{~for~}i\in\{k_2+1,\dots,m\},\\
\end{cases}$$  for some $1\le k_1\le k_2\le m$. To show $(\bfx,\bfx')\in E(G^{(m)}),$ it suffices to find $\bfy=[y_1,\dots,y_m]$ such that  $p_{\bfX\bfY}(\bfx,\bfy)p_{\bfX\bfY}(\bfx',\bfy)>0$ and $f^{(m)}(\bfx,\bfy)\ne f^{(m)}(\bfx',\bfy).$

For $i\in\{1,\dots,k_1\}$, choose $y_i$   such that    $p_{XY}(x_i,y_i)p_{XY}(x_i',y_i)>0$ and $f(x_i,y_i)\ne f(x_i',y_i)$. Such $y_i$ exists as $(x_i,x_i')\in E(G)$. For $i\in\{k_1+1,\dots,k_2\}$, choose $y_i$  such that    $p_{XY}(x_i,y_i)p_{XY}(x_i',y_i)>0$ and $f(x_i,y_i)= f(x_i',y_i)$. Such $y_i$ exists because $(x_i,x_i')\notin E(G)$ and because  all $(x,x')\notin E(G)$ are due to condition $C2$.  For $i\in\{k_2+1,\dots,m\}$, choose $y_i$  such that  $p_{XY}(x_i',y_i)=p_{XY}(x_i,y_i) >0$. Such $y_i$ exist because $p_{XY}(x,y)$ satisfies Assumption \ref{assume:irreducible}.  Then, $p_{\bfX\bfY}(\bfx,\bfy)p_{\bfX\bfY}(\bfx',\bfy)=\prod_{i=1}^mp_{XY}(x_i,y_i)p_{XY}(x_i',y_i)>0$. We have $f^{(m)}(\bfx,\bfy)\ne f^{(m)}(\bfx',\bfy)$ because $f(x_i,y_i)\ne f(x_i',y_i)$ for $i\in\{1,\dots,k_1\}$.

Therefore, $(\bfx,\bfx')\in E(G^{(m)})$ holds. Since $(\bfx,\bfx')\in E(G^{\lor m})$ is arbitrary, we have $E(G^{\lor m})\subseteq E(G^{(m)})$.
\end{proof} 
As a consequence of Theorem \ref{thm:theorem4}, given a simple graph $G$, we can always construct functions $f$ and $g$ such that the $m$-instance $f$-confusion graphs of $f$ and $g$ are  $G^{\boxtimes m}$ and $G^{\lor m}$ respectively. 
\begin{theorem}\label{thm:construct_function_based_on_G}
Let $G$ be a simple graph. 
\begin{enumerate}[label=(\alph*)]
\item There exists a function $f(\cdot,\cdot)$ and a joint p.m.f. $p_{XY}(\cdot,\cdot)$ such that  the   $f$-confusion graph over $m$-instances is $G^{\boxtimes m}$.
\item There exists a function $g(\cdot,\cdot)$ and a joint p.m.f. $p_{XY}(\cdot,\cdot)$ such that  the $g$-confusion graph over $m$-instances  is $G^{\lor m}$.
\end{enumerate}
\end{theorem}
\begin{proof}[Proof of Theorem 2 (a)]
Let $\calX = V(G)$ and $\calY = E(G)$. Choose any valid p.m.f $p_{XY}(\cdot,\cdot)$ such that  $p_{XY}(x,y)>0$ if and only if vertex $x$ is incident with edge $y$.  Then, we set $f(x,y) = x$. Denote $f$-confusion graph by $G_f$.

For $x, x' \in \calX$ and $x \neq x'$, $(x,x')\notin E(G_f)$ implies that $p_{XY}(x,y)p_{XY}(x',y)=0$ for all $y \in \calY$, i.e., $(x,x')\notin E(G_f)$ is due to C1.
By Theorem \ref{thm:theorem4}, we have that $G^{(m)} = G^{\boxtimes m}$ for $m\ge 1$.
\end{proof} 
\begin{proof} [Proof of Theorem 2 (b)]
Let $\calX = V(G)$ and $\calY=\{\{i,j\} : i,j\in V(G) \text{~and~} i\neq j\}$. Choose any valid p.m.f. $p_{XY}(\cdot,\cdot)$ such that  $p_{XY}(x,y)>0$ if and only if vertex $x$ is an element of set $y$.     The function evaluation is such that  $$ g(x,y) = \begin{cases}
    x,\text{ if }y\in E(G),\\
    \boxplus,\text{ if }y\notin E(G),
\end{cases}$$ 
where symbol $\boxplus \notin \calX$. Denote $g$-confusion graph by $G_g$. Now we show that $G_g=G$ is $g$-confusion graph. For each $x\ne x'\in \calX$, there is exactly one choice of $y$, that is $\{x,x'\}$, such that  $p_{XY}(x,y)p_{XY}(x',y)>0$. $g(x,\{x,x'\})= g(x',\{x,x'\})$ if and only if $\{x,x'\}\notin E(G)$. This implies $G_g= G$ and all $(x,x')\notin E(G_g)$ is due to C2 as $x\ne x'\in \calX$ are arbitrary.  By Theorem \ref{thm:theorem4}, we have that $G^{(m)} = G^{\lor m}$ for $m\ge 1$.
\end{proof}


\begin{remark}\label{remark:free_to_choose_p_x}
    We point out that for a given $f$-confusion graph, for all $x\in\calX,y \in \calY$, we can always choose $p_{XY}(x,y)$ in such a way that $p_X(x)$ is uniform. Since $p_{XY}(x,y)= p_X(x) \cdot p_{Y|X}(y|x)$, we can set $p_X(x)$ be uniform and then choose $ p_{Y|X}(y|x)$ accordingly. 
\end{remark} 

We note here that the work of Charpenay et al. \cite{CharpenayTR23ITW}, considers a related problem in the classical zero error setting. 
For $m$-instances, they consider a model where Alice has access to side information $g(Y_t)$, $t \leq m$ in addition to $X$. 
%
%
Their work \cite{CharpenayTR23ITW}, considers a pairwise shared side information condition (Definition IV.1 of \cite{CharpenayTR23ITW}) and shows that in this case the $m$-instance $f$-confusion graph can be written as a
disjoint union of OR products. If $g(y)$ is a constant for all $y\in \calY$, this corresponds to our model. In this case, the pairwise shared side information condition implies that all $\bfx,\bfx'\notin E(G^{(m)})$ and $\bfx\ne \bfx'$ are due to condition C2, i.e., the if-part of Theorem \ref{thm:theorem4} (b) holds. In contrast, we provide an if-and-only-if characterization that covers both the strong product and OR product scenarios.

\section{Orthogonal representation of graphs}\label{sec:orth_rep}

Orthogonal representations of graphs generalize  vertex coloring and are the appropriate notion to consider within quantum codes for our problem setting.

\begin{definition}
An orthogonal representation\footnote{
Our definition of the orthogonal representation aligns with Lov\'asz's number defined in \cite{lovasz1979shannon}.  
The orthogonal representation is sometimes defined such that  \textbf{adjacent} vertices are assigned orthogonal vectors. The other definition is related to ours by taking graph complement. }
of a graph $G$ is a mapping $\phi:V(G)\mapsto \bbC^m$ for some $m\ge 1$ such that  each $\phi(v)$ is a complex unit-norm vector and distinct \textbf{non-adjacent}  vertices are assigned orthogonal vectors. The orthogonal rank of $G$, denoted by $\xi(G)$, is the minimum dimension $m$ such that  there exists an orthogonal representation of $G$.
\end{definition}  
We will use the following results (proved in Appendix \ref{app:appendix_C} for completeness).
\begin{proposition}\label{prop:xi}
\begin{enumerate}[label=(\alph*)]
    \item 
\label{prop:subgraph_xi}
If $G\subseteq H$, then $\xi(H)\le \xi(G)$.
\item \label{prop:xi_product}
$\xi(G\lor   H)\le  \xi(G\boxtimes  H)\le \xi(G)\xi(H)$.
\item \label{prop:xi_chi_compare}
$\xi(G)\le \chi(\overline{G})$. \label{prop:quan_less_than_class}
\end{enumerate} 
\end{proposition}


\begin{definition}{\it Lov\'asz number.} 
Let $G$ be     a finite simple graph. Its Lov\'asz number $\vartheta(G)$ \cite{lovasz1979shannon} is
$$
\vartheta(G) = \min_{\phi,c}\max_{i\in V(G)}\frac{1}{|\braket{c, \phi(i)}|^2}
$$
where the minimization is over all orthogonal representations $\phi$ such that every vertex is mapped to a real vector, and over all unit-norm real vectors $c$ ($c$ is called the handle).
\end{definition}  

A projector $P$ is a Hermitian matrix in $\bbC^{d\times d}$ such that $P^2=P.$
\begin{definition}{\it Projective rank \cite{Roberson_13}.}  
A $d/r$-projective representation of a graph
$G$ is an assignment of rank $r$ projectors in $\bbC^{d\times d}$ to the vertices of $G$ such that distinct \textbf{non-adjacent}
vertices are assigned orthogonal projectors. We say that the value of a $d/r$-representation
is the rational number $d/r$. The projective rank of $G$, denoted by $\xi_f(G)$, is defined by
\begin{align*}
    \xi_f(G) =\inf_{d,r}\bigg{\{}\frac{d}{r}:G\text{ has a }d/r\text{-projective representation}\bigg{\}}.
\end{align*}
\end{definition} 
We note that a $d/1$-projective representation is an orthogonal representation, so $\xi_f(G)\le \xi(G)$.
\begin{lemma}\label{lemma:lovasz}
\begin{enumerate}[label=(\alph*)]
\item \label{prop:lovasz_number_subgraph}
If $G\subseteq H$, then $\vartheta(G) \ge \vartheta (H)$.
    \item  \label{lemma:lovasz_number_indep_compare}
$\alpha(G)\le \vartheta(G)$ (Lemma 3 of \cite{lovasz1979shannon}). 
\item \label{lemma:lovasz_number_AND_product}
$\vartheta(G\lor H) =\vartheta(G\boxtimes H) = \vartheta(G)\vartheta(H)$ (Theorem 7 of \cite{lovasz1979shannon})\footnote{Lov\'asz noted that $\vartheta(G\lor H) = \vartheta(G)\vartheta(H)$ follows from the arguments in \cite{lovasz1979shannon}. We provide an explicit proof in Appendix \ref{app:vartheta_multiplicity} for completeness.}.
\item $\vartheta(G) \le \xi_f({G})\le \xi({G})$ (Section 6 of  \cite{Mancinska16}).\label{lemma:lovasz_number_xi_compare}
\end{enumerate}
Using the results in parts (b), (c) and (d), we have
\begin{align}
&  \alpha(G^{\boxtimes m}) \le\vartheta(G ^{\boxtimes m}) =   \vartheta(G )^m  = \vartheta(G ^{(m)}) = 
\vartheta(G ^{\lor m})\le \xi(  G ^{\lor m}). \label{eq:sandwich}
\end{align}
$\vartheta(G ^{(m)}) = \vartheta(G )^m  $ comes from the following. Proposition \ref{prop:AND_power_m-instance_graph_OR_power} states that $G^{\boxtimes m} \subseteq G^{(m)} \subseteq G^{\lor m}$ for $m\ge 1$. Part (b) of the above Lemma gives 
\begin{align*}
    \vartheta(G ^{\boxtimes m}) \ge \vartheta(G ^{(m)})  \ge  
\vartheta(G ^{\lor m}).
\end{align*}
Part (c)  the above Lemma gives $ \vartheta(G ^{\boxtimes m})=\vartheta(G ^{\lor m}) = \vartheta(G )^m$. Then, the conclusion follows.
\end{lemma} 





\begin{lemma}\label{lemma:lemma4}
For a function $f:\calX\times\calY \mapsto\calZ$ with p.m.f. $p_{XY}(\cdot,\cdot)$, let $G$ denote the $f$-confusion graph. The optimal quantum rate for computation over $m$ instances is $\frac{\log_2\xi(\overline{G^{(m)}})}{m}$ ($\overline{G^{(m)}}$ denotes the complement of the $f$-confusion graph over $m$-instances $G^{(m)}$).
\end{lemma}
\begin{proof}
An orthogonal representation $\phi:\overline{G^{(m)}}\mapsto \bbC^{\xi(\overline{G^{(m)}})}$  induces a quantum code as follows. Suppose Alice and Bob get  $\bfx $ and $\bfy$ respectively.  Alice sends $\ket{\phi(\bfx)}\bra{\phi(\bfx)}$ to Bob.  Bob chooses his POVM to be $\{\Pi^{\bfz,\bfy}  :\bfz\in \calZ^m\}$ where $\Pi^{\bfz,\bfy}$ is the projector onto $S_{\bfz,\bfy}:=\text{span}\{\ket{\phi(\bfx)}:p_{\bfX\bfY}(\bfx,\bfy)>0, f^{(m)}(\bfx,\bfy) 
 = \bfz \}$. 
 If $\sum_{\bfz}\Pi^{\bfz,\bfy}\ne I$, Bob adds $I-\sum_{\bfz}\Pi^{\bfz,\bfy}$ to complete a POVM. If $\bfx,\bfx'$ is such that $p_{\bfX\bfY}(\bfx,\bfy)p_{\bfX\bfY}(\bfx',\bfy)>0$, and $f^{(m)}(\bfx,\bfy)=\bfz,f^{(m)}(\bfx',\bfy)=\bfz'$ for some $\bfz\ne \bfz'$, then $\ket{\phi(\bfx)} \in S_{\bfz,\bfy}$ and $\ket{\phi(\bfx')}\in S_{\bfz',\bfy}$. Since $\phi$ is an orthogonal representation of $\overline{G}$, then  $\Pi^{\bfz,\bfy}\perp\Pi^{\bfz',\bfy}$, which implies that  the code is zero-error.   The code has rate $\frac{1}{m}\log_2\xi(\overline{G^{(m)}})$ as the orthogonal representation $\phi$ is $\xi(\overline{G^{(m)}})$-dimensional.


Conversely, a code   of rate $\frac{1}{m}\log_2d$ induces an orthogonal representation $\psi:V(G^{(m)})\mapsto \bbC^{d}$. Let $\rho_{\bfx}$ be a $d\times d$  density operator associated with input $\bfx$. If $\rho_{\bfx}$ is a pure state $\ket{v}\bra{v}$ for some $\ket{v}\in \bbC^d$, then we set  
 $\psi(\bfx) = \ket{v}$. If $\rho_{\bfx}$ is a mixed state with its density operator written as  
 $\rho_{\bfx} =\sum_{i=1}^k \lambda_i\ket{v_i}\bra{v_i}\text{ with }\lambda_i\ne 0\text{ for }i=1,\dots,k.$
Then let $l\in[k]$ be arbitrary and $\psi(\bfx ) = \ket{v_l}$. 


Now we check it is indeed an orthogonal representation. Let    $(\bfx,\bfx')\in E(G^{(m)})$ be such that $p_{\bfX\bfY}(\bfx,\bfy)p_{\bfX\bfY}(\bfx',\bfy)>0$ and $f^{(m)}(\bfx,\bfy)\ne f^{(m)}(\bfx',\bfy)$.  This code  is zero-error, so the POVM $\{\Lambda^\bfz_{f(\cdot,\cdot), \bfy}\}_{\bfz \in \calZ^m}$ must distinguish $\rho_{\bfx},\rho_{\bfx'}$ perfectly. This implies that $\rho_{\bfx}\perp\rho_{\bfx'}$. We have 
\begin{align*}
   &  \sum_{i=1}^k \lambda_i\ket{v_i}\bra{v_i}  \perp  \sum_{i=1}^k \lambda_i'\ket{v_i'}\bra{v_i'}  \Rightarrow \ket{v_i}\perp\ket{v_j'}\text{ for all }i,j
   \Rightarrow \psi(\bfx) \perp \psi(\bfx').
\end{align*}
Since $(\bfx,\bfx')\in E(G)$ is arbitrary, it implies that $\psi$ is an orthogonal representation.  
\end{proof}

\begin{theorem}\label{thm:witsenhausen_quantum_rate}  The rate in the quantum setting is given by 
\begin{align*}
    R_{\text{quantum}} &= \inf_m \frac{\log_2\xi(\overline{G^{(m)}})}{m}.
\end{align*}
\end{theorem} 
\begin{proof} 
It can be seen that $\log_2\xi(\overline{G^{(m+n)}})\le\log_2\xi(\overline{G^{(m)}})+\log_2\xi(\overline{G^{(n)}})$. Indeed, fix a code for $G^{(m)}$ and $G^{(n)}$ respectively. If Alice and Bob have inputs $\bfx = (\bfx_1,\bfx_2)\in \calX^{m+n}$ and $\bfy = (\bfy_1,\bfy_2) \in \calY^{m+n}$ where $\bfx_1\in \calX^m,\bfx_2\in \calX^n,\bfy_1\in\calY^m,\bfy_2\in\calY^n$, they can run the codes for $G^{(m)}$ and $G^{(n)}$ separately associated with inputs $(\bfx_1,\bfy_1)$ and $(\bfx_2,\bfy_2)$ respectively. This results in an $(m+n)$-instance code. 
By Fekete lemma \cite{Fekete_23}, the limit exists and is $\inf_m \frac{1}{m}\text{log}\, \xi (\overline{G^{(m)}})$.
\end{proof}


We emphasize that single-letter characterizations for the quantum rate and the FLC/VLC rates are not known in general. However,  for specific classes of $f$-confusion graphs, it is possible to obtain a single-letter characterization. 



In Lemma \ref{lemma_single_letter} below, if a) holds, then we have a single-letter characterization for quantum rate (asymptotic in $m$), but may not have a single-letter characterization for FLC rate. If (b) holds, i.e., if the confusion graph is perfect, then both quantum and FLC rate are equal and have a single-letter characterization. 

\begin{definition}{\it Perfect Graph, (Chapter 8 of \cite{west_2000}).}
A graph $G$ is perfect if $\chi(H)=\omega(H)$ for every induced subgraph $H$ of $G$.
\end{definition}

\begin{lemma}\label{lemma_single_letter}
    \begin{enumerate}[label=(\alph*)]
        \item If $\vartheta(\overline{G}) = \xi(\overline{G})$, then $\frac{1}{m}\log \xi(\overline{G^{(m)}}) =\xi (\overline{G})$ for all $m\ge 1$.  \label{lemma_single_letter_part_a}
        \item If $G$ is perfect, then $\xi(\overline{G^{(m)}})^{\frac{1}{m}} =\xi (\overline{G}) = \chi(G)=\chi(G^{(m)})^{\frac{1}{m}}$ for all $m\ge 1$.\label{lemma_single_letter_part_b} 
    \end{enumerate}
\end{lemma} The proof is given in Appendix \ref{app:proof_lemma_single_letter}. 
Suppose $G$ satisfies $\vartheta(\overline{G}) = \xi(\overline{G})$ or  is perfect. Taking infimum over $m$, we obtain a single-letter characterization  $\inf_m \frac{\log_2\xi(\overline{G^{(m)}})}{m} =\xi (\overline{G})$.  We point out that condition \ref{lemma_single_letter_part_b} implies condition \ref{lemma_single_letter_part_a} in the above Lemma. This is because if $G$ is perfect, then
\begin{align*}
    \omega(G) = \alpha(\overline{G}) \overset{(a)}{\le}  \vartheta(\overline{G})\overset{(b)}{\le}\xi(\overline{G})\overset{(c)}{\le}\chi(G) =\omega(G).
\end{align*}
where $(a),(b)$ are from  Lemma \ref{lemma:lovasz} \ref{lemma:lovasz_number_indep_compare} and \ref{lemma:lovasz_number_xi_compare} respectively, and $(c)$ is from Proposition \ref{prop:xi} \ref{prop:xi_chi_compare}. On the other hand, the condition in Lemma \ref{lemma_single_letter} \ref{lemma_single_letter_part_a} does not imply the condition in Lemma \ref{lemma_single_letter} \ref{lemma_single_letter_part_b} as there exist graphs such that $\vartheta(\overline{G}) = \xi(\overline{G})$, which are not perfect. One such example will be discussed in Section \ref{subsubsec:G13}.  

\section{Analyzing the behavior of quantum advantage in FLC}\label{sec:analysis_bahavior_quantum_advantage}
For any $m$-instance $f$-confusion graph $G^{(m)}, m \geq 1$, it is evident that the quantum rate is at most the classical FLC rate ({\it cf.} Proposition \ref{prop:xi} \ref{prop:quan_less_than_class}). We refer to the ratio of classical FLC and quantum rates as the quantum advantage in this section. As discussed in Section \ref{subsec:related_work}, there are several results in the literature that discuss the quantum advantage in the single-instance case. However, depending on the problem, a quantum advantage in the single-instance setting may not necessarily translate into a quantum advantage when considering $R_{\text{quantum}}$ and $R_{\text{FLC}}$.

In this section, we consider several examples of problems within function computation with side information and explore the behavior of the quantum advantage in the different settings (single-instance and asymptotic rate). For all the settings, we consider the quantum advantage under the two extremes of the corresponding $f$-confusion graph, namely $G^{\boxtimes m}$ and $G^{\lor m}$.

\subsection{No quantum advantage in both single-instance and multiple-instance settings.\label{subsec:C5}}

Consider the functions $\tilde f$ and $\tilde g$ described in the  Table \ref{table:simple_eg1} and \ref{table:simple_eg2}. The $m$-instance $f$-confusion graphs for $\tilde f$ and $\tilde g$ are strong and OR products of $C_5$ respectively.   The next proposition shows that there is no quantum advantage in both single-instance and multiple-instance cases for both $\tilde f$ and $\tilde g$. 
\begin{proposition}\label{prop:C5_AND_rates} The following statements hold. (i) $\xi(\overline{C_5}) = \chi(C_5)=3$, and (ii)
$R_{\text{quantum}}(\tilde f)  = R_{\text{FLC}}(\tilde f) = \frac{1}{2}\log_2 5$, (iii)
$R_{\text{quantum}}(\tilde g)  = R_{\text{FLC}}(\tilde g) =  \log_2 \frac{5}{2}$. 
\end{proposition}
\begin{proof}
    See Appendix \ref{app:pf_prop10}.
\end{proof}

\subsection{Quantum advantage in both single-instance and multiple-instance settings.}
\subsubsection{Example associated with graph OR product\label{subsubsec:G13}}
\begin{definition}{\it The thirteen vertex graph $G_{13}$.} 
\label{defn:G_thirteen}
    Let 
    \begin{align*}
    &S = \{A = (1,0,0), B = (0,1,0), C = (0,0,1), L = (0,1,\\&1), M = (0,1,-1), N = (1,0,1), P = (1,0,-1), Q = (1,1,\\&0), R = (1,-1,0), Y = (1,1,-1),X = (1,-1,1),
    Z = \\&(-1,1,1),W = (1,1,1)\}\subseteq \bbC^3. 
    \end{align*}
      $G_{13}$, the graph in Fig  \ref{fig:G13_4_coloring}, is the graph with vertex set $S$. Two vertices $u=(x_1,y_1,z_1),v=(x_2,y_2,z_2)$ are adjacent if $x_1x_2+y_1y_2+z_1z_2=0.$
\end{definition}
The graph $G_{13}$ was introduced in \cite{YuO12}.
Consider a function whose $m$-instance $f$-confusion graph is $G_{13}^{\lor m}$ for $m\ge 1$.
\begin{figure}[t]
    \centering
     
\begin{tikzpicture} 
[
  node distance=0.5cm,
  every node/.style={circle, draw, minimum size=4mm, inner sep=0pt, font = \normalsize},
  every edge/.style={draw, -} 
]

\node[draw=red, fill = red!20] (W) at (\xW,\yW) {W};
\node[draw=green, fill = green!20] (X) at (\xX,\yX) {X};
\node[draw=blue, fill = blue!20] (Y) at (\xY,\yY) {Y};
\node[draw=blue, fill = blue!20] (Z) at (\xZ,\yZ) {Z};

\node (M)[draw=green, fill = green!20] at (\xM,\yM) {M};
\node (L)[draw=red, fill = red!20] at (\xL,\yL) {L};
\node (N)[draw=green, fill = green!20] at (\xN,\yN) {N};
\node (P)[draw=blue, fill = blue!20] at (\xP,\yP) {P};
\node (Q)[draw=yellow, fill = yellow!20] at (\xQ,\yQ) {Q};
\node (R)[draw=green, fill = green!20] at (\xR,\yR) {R};

\node (A)[draw=blue, fill = blue!20] at (\xA,\yA) {A};
\node (B)[draw=yellow, fill = yellow!20] at (\xB,\yB) {B};
\node (C)[draw=red, fill = red!20] at (\xC,\yC) {C}; 


\draw (A) edge  (B);
\draw (A) edge  (C);
\draw (B) edge  (C);

\draw (L) edge  (M);
\draw (A) edge  (M);
\draw (A) edge  (L);

\draw (N) edge  (P);
\draw (B) edge  (N);
\draw (B) edge  (P);

\draw (Q) edge  (R);
\draw (Q) edge  (C);
\draw (R) edge  (C);

\draw (X) edge  (P);
\draw (Y) edge  (R);
\draw (N) edge  (Z);
\draw (M) edge  (Z);
\draw (L) edge  (Y);
\draw (Q) edge[bend left = 10]  (Z); 
\draw (W) edge  (M);
\draw (X) edge[bend right = 10]  (L);
\draw (R) edge  (W);
\draw (X) edge  (Q);
\draw (P) edge  (W);
\draw (N) edge[bend left = 10]  (Y); 
\end{tikzpicture}
    \caption{$G_{13}$ with a $4$-coloring.}
    \label{fig:G13_4_coloring}
\end{figure}
\begin{proposition}\label{prop:G_13_quantum_rate}
$$\xi(\overline{G_{13}^{\lor m}})  = \xi(\overline{G_{13}})^m = 3^m, \chi_f(G_{13}) = 35/11. $$    
\end{proposition}
\begin{proof}
     The work of \cite{GuptaSXCA_23} demonstrates that $\chi_f(G_{13}) = \frac{35}{11}$ via a numerical calculation. As $\chi(G_{13}) \geq \lceil \chi_f(G_{13}) \rceil$, this implies that $\chi(G_{13}) \ge 4$. In fact, $\chi(G_{13}) = 4$ as we have a 4-coloring as depicted in Fig. \ref{fig:G13_4_coloring}. Since the labeling in Definition \ref{defn:G_thirteen} provides an orthogonal representation of $\overline{G_{13}}$, this implies that there is a quantum advantage in the single-instance case. 
     
     Next, we compare the asymptotic rates. Since making all vectors in $S$ unit-norm yields an orthogonal representation of $\overline{G_{13}}$ in $\bbC^3$, we have that $\xi(\overline{G_{13}^{\lor m}}) = \xi(\overline{G_{13}}^{\boxtimes m}) \le \xi(\overline{G_{13}})^m \le 3^m$. Since $A,B,C$ are mutually orthogonal, they form a triangle ($K_3$) in $G_{13}$. Therefore, there is a $\overline{K_3^{\lor m}} = \overline{K_{3^m}}$ in $\overline{G_{13}^{\lor m}}$. It is an independent set of $3^m$ vertices. Any orthogonal representation of $\overline{G_{13}^{\lor m}}$ has to assign these vertices mutually orthogonal vectors. Therefore, $\xi(\overline{G_{13}^{\lor m}})\ge3^m$. 
\end{proof}
  We have that 
\begin{align*}
   & \chi(G^{\lor m})^{\frac{1}{m}} \ge \chi_f(G^{\lor m})^{\frac{1}{m}} =  \chi_f(G)  \\&= (35/11)  >3=\xi(\overline{G^{\lor m}})^{\frac{1}{m}}
\end{align*}
where the first equality holds from   (\ref{eqn:chi>=chi_f}), the first equality holds from Proposition \ref{prop:chi_frac} \ref{prop:OR_chi_f}.
Therefore, there is a quantum advantage in computing this function over $m$ instances for $m\ge 1$.

\begin{remark}
    The above example shows that there exists $G$ such that $\inf_m\frac{1}{m}\,\text{log}_2\,\xi(\overline{G^{\lor m}}) < \chi_f(G)$. It implies that the quantum rate is strictly smaller than the fractional chromatic number. Also,  \cite{ManinskaR16} (Section 4.1) showed that $\vartheta(\overline{G_{13}})= 3=\xi(\overline{G_{13}})$ and $\omega(G_{13})=3\ne 4=\chi(G_{13})$. This implies that $G_{13}$ satisfies  Lemma \ref{lemma_single_letter} (a) but does not satisfy Lemma \ref{lemma_single_letter} (b).
\end{remark} 

\subsubsection{Example associated with graph strong product\label{subsubsec:Hn}}
We consider a family of graphs $\{H_n:n=4p^l-1, l\ge 1, p\ge 11\text{ is a prime number}\}$ from \cite{BrietBL_15} (Definition 2.2). $H_n$ is defined as follows. The vertex set $V(H_n)$ consists of vectors in $\{-1,1\}^n$ with an even number of ``$-1$'' entries. The edge set consists of the pairs with inner product  $-1$. In \cite{BrietBL_15}, the authors show that $H_n$ has an exponential quantum advantage between entanglement-assisted and classical source coding. We point out that their result also implies an exponential separation in our setting.
\begin{proposition}[Corollary 5.8 of \cite{BrietBL_15}]\label{prop:lower_bound_H_n}
    \begin{align*}
       \chi(H_n^{\lor m})^{\frac{1}{m}} \ge \chi(H_n^{\boxtimes m})^{\frac{1}{m}} \ge 2^{0.154n - 1}\text{ for }m\ge 1.
    \end{align*}
\end{proposition}
\begin{proposition}(Lemma 5.2 of \cite{BrietBL_15})\label{prop:xi_of_H_n}
    \begin{align*}
 \xi(\overline{H_n}) \le n+1\text{ for }m\ge 1.
    \end{align*}
\end{proposition}
We have $\xi(\overline{H_n^{\boxtimes m}})^{\frac{1}{m}} \le \xi(\overline{H_n^{\lor m}})^{\frac{1}{m}} \le \xi(\overline{H_n})$
by Proposition \ref{prop:xi}  \ref{prop:xi_product}. 
Combining this and Proposition \ref{prop:lower_bound_H_n} and noting that $2^{0.154n-1} > n+1$ for $n=4p^l-1,l\ge 1,p\ge 11$, we have,
\begin{align*}
&\chi(H_n^{\lor m})^{\frac{1}{m}} \ge\chi(H_n^{\boxtimes m})^{\frac{1}{m}} \ge 2^{0.154n - 1} >n+1
\ge \xi(\overline{H_n^{\lor m}})^{\frac{1}{m}} \ge \xi(\overline{H_n^{\boxtimes m}})^{\frac{1}{m}}\text{ for }m\ge 1.
\end{align*}
This gives the desired result. 

We point out that a qualitatively similar result can be obtained by considering the work of \cite{Bar-YossefJK08}. The $f$-confusion graph used in \cite{Bar-YossefJK08}, denoted by $F_n$ where $n=4p$ and $p$ is an odd prime number, is defined as follows (Theorem 6.2 of \cite{Bar-YossefJK08}). The vertex set of the $f$-confusion graph $F_n$ consists of all $n$-bit binary strings. $(x,x')\in E(F_n)$ if and only if their Hamming distance\footnote{Hamming distance between two strings of equal length is the number of positions where they differ.} is $n/2$.   $H_{n-1}$ is the induced subgraph of $F_n$ on the set of binary strings with even Hamming weight\footnote{Hamming weight of a binary string is the number of $1$'s in the string.} and the first entry of every binary string is $0$ (Definition 2.2 of \cite{BrietBL_15}).  Proposition \ref{prop:AND_power_m-instance_graph_OR_power} states that for any $m$-instance $f$-confusion graph $G$, $G^{\boxtimes m}\subseteq G^{(m)}\subseteq G^{\lor m}$ for all $m\ge 1$. Therefore, $H_{n-1}^{\boxtimes m} $ is a subgraph (not necessarily induced) of $F_n^{(m)}$ as $H_{n-1}^{\boxtimes m}$ is a subgraph of $F_n^{\boxtimes m}$.  Proposition \ref{prop:lower_bound_H_n} implies the classical rate is lower bounded by $0.154 (n-1) - 1=\Omega(n)$. In \cite{Bar-YossefJK08}, a $O(\log_2 n)$-bits quantum code is presented for the single-instance case, which implies the quantum rate is at most $O(\log_2 n)$. It can be observed that this quantum advantage continues to hold even when considering $m$-instances.

\subsection{Quantum advantage in single-instance setting, but no quantum advantage in multiple-instance setting.\label{sec:line_graph_eg}}

For constructing this scenario, we require a few auxiliary definitions, which utilize directed graphs. We emphasize that we only consider simple graphs.
\begin{definition}{\it Line graph of a directed graph.} \label{defn:directed_line_graph}
    Let $G = (V,E)$ be a directed graph. Its line graph $\mathfrak{L}_D(G)$ is an undirected graph defined as follows.  The vertex set $V(\mathfrak{L}_D(G))$ is $E$. $(u_1,v_1),(u_2,v_2)\in E$ are adjacent if and only if $(u_1,v_1),(u_2,v_2)$ form a directed walk of length 2, i.e., one of the following is true.
    \begin{align*}
    \big ( u_1 = v_2\text{ and }u_2 = v_1 \big )\text{ or }\\
    \big ( u_1 \ne v_2\text{ and }u_2 = v_1 \big )\text{ or }\\
    \big ( u_1 = v_2\text{ and }u_2 \ne v_1 \big ).
    \end{align*}
    If $G$ is an undirected graph, we view $G$ as a directed graph where each undirected edge is replaced with two directed edges with opposite orientations and $\frL_D(G)$ is then defined as above.
\end{definition}

\begin{proposition}\label{prop:chi_and_xi_of_line_graph}
    Let $G$ be an undirected graph. Then $\chi(\frL_D(G) )\le \chi(G)$ and $\xi(\overline{\frL_D(G)}) \le \xi(\overline{G})$. 
\end{proposition}
\begin{proof}
Let $\psi:V(G) \mapsto [\chi(G)]$ be a $\chi(G)$-coloring of $G$. Define
\begin{align*}
    \phi:V(\frL_D(G)) \mapsto [\chi(G)], \phi(u,v) \mapsto \psi(u).
\end{align*}
Suppose $(u_1,v_1),(u_2,v_2)$ are adjacent in $\frL_D(G)$. By symmetry, we assume that $v_1 = u_2$. We have that $\phi(u_1,v_1)=\psi(u_1),\phi(u_2,v_2)=\psi(u_2)=\psi(v_1)$. 
By Definition \ref{defn:directed_line_graph}, $(u_1,v_1)$ are also adjacent in the undirected graph $G$. Since $\psi$ is a $\chi(G)$-coloring of $G$, we have $\psi(u_1)\ne \psi(v_1)$. Therefore, $\phi$ is a proper coloring of $\frL_D(G)$.

The second claim can be proved similarly. Specifically, now let $\psi:V(\overline{G}) \mapsto [\xi(\overline{G})]$ be a $\xi(\overline{G})$-dimensional orthogonal representation of $\overline{G}$. Define
\begin{align*}
    \phi:V(\overline{\frL_D(G)}) \mapsto [\xi(\overline{G})], \phi(u,v) \mapsto \psi(u).
\end{align*}
Suppose $(u_1,v_1),(u_2,v_2)\notin E(\overline{\frL_D(G)})$. It follows that $(u_1,v_1),(u_2,v_2)\in E(\frL_D(G))$. By symmetry, we assume that $v_1 = u_2$. We have that $\phi(u_1,v_1)=\psi(u_1),\phi(u_2,v_2)=\psi(u_2)=\psi(v_1)$. 
By Definition \ref{defn:directed_line_graph}, $(u_1,v_1)$ are also adjacent in the undirected graph $G$. Since $\psi$ is a $\xi(\overline{G})$-dimensional orthogonal representation of $\overline{G}$, we have $\psi(u_1)\perp \psi(v_1)$. Therefore, $\phi$ is a $\xi(\overline{G})$-dimensional orthogonal representation of $\frL_D(G)$.
\end{proof}


Denote $H=\frL_D(G_{13}).$ We will show  that 
\begin{align} 
     &\chi  (H) > \xi(\overline{H}) \text{ and } \inf_m\chi( {H^{\lor m}})^{1/m} = \inf_m \xi(\overline{H^{\lor m}})^{1/m} = \nonumber
     \\& \inf_m\chi( {H^{\boxtimes m}})^{1/m} = \inf_m \xi(\overline{H^{\boxtimes m}})^{1/m} =3. \label{eqn:calL_D(G13)_rates}
\end{align} 
We discuss the consequence of the above equations before we prove them. Pick a function and a valid p.m.f. such that the single-instance $f$-confusion graph is $H$. It has a quantum advantage in the single-instance case as $\chi(H) > \xi(\overline{H})$. Proposition \ref{prop:AND_power_m-instance_graph_OR_power} states that $H^{\boxtimes m}\subseteq H^{(m)}\subseteq H^{\lor m}$ for $m\ge 1$. By Proposition \ref{prop:xi} \ref{prop:subgraph_xi}, we have 
\begin{align*}
   \chi( {H^{\boxtimes m}})^{1/m}  \le \chi( {H^{(m)}})^{1/m} \le  \chi( {H^{\lor m}})^{1/m}, \text{~and}\\
   \xi( \overline{H^{\boxtimes m}})^{1/m}  \le \xi( \overline{H^{(m)}})^{1/m} \le  \xi( \overline{H^{\lor m}})^{1/m}.
\end{align*}
Then,  \eqref{eqn:calL_D(G13)_rates}  implies that there is a quantum advantage in the single-instance setting, but the advantage disappears in the asymptotic setting. We will prove  (\ref{eqn:calL_D(G13)_rates}) in the rest of this subsection.

\begin{claim}\label{claim:claim2}
    \begin{align*}
        \xi(\overline{H^{\boxtimes m}})=\xi(\overline{H^{\lor m}})    =3^m\text{ for }m\ge 1.
    \end{align*}
\end{claim}
\begin{proof}
Let $m\ge 1$. Since $\xi(\overline{G_{13}}) = 3$ and $H = \frL_D(G_{13})$, Proposition \ref{prop:chi_and_xi_of_line_graph} gives $\xi(\overline{H}) \le 3$. Therefore,
    \begin{align*}
    \xi(\overline{H^{\boxtimes m}})\le \xi(\overline{H^{\lor m}})=\xi(\overline{H}^{\boxtimes m}) \le \xi(\overline{H})^{m}     \le 3^m
    \end{align*}
where the first inequality holds by Proposition \ref{prop:xi}  \ref{prop:subgraph_xi}, the first equality holds by Proposition \ref{prop:chi}  \ref{prop:de_morgan}, the second inequality holds by Proposition \ref{prop:xi} \ref{prop:xi_product}. 

Next we show that $\xi(\overline{H^{\boxtimes m}})\ge 3^m$ for $m\ge 1$. Since $A,B,C$ form a triangle in $G_{13}$, $(A,B)(B,C)(C,A)$ form a triangle in $\frL_D(G_{13})$.  Therefore, there is a $\overline{K_3^{\boxtimes m}} = \overline{K_{3^m}}$ in $H^{\boxtimes m}$. It is an independent set of $3^m$ vertices. Any orthogonal representation of $\overline{H^{\boxtimes m}}$ has to assign these vertices mutually orthogonal vectors. Therefore, $\xi(\overline{H^{\boxtimes m}})\ge3^m$. 
\end{proof}

\begin{claim}
    \begin{align*}
     \inf_m\chi( {H^{\boxtimes m}})^{1/m} =   \inf_m\chi( {H^{\lor m}})^{1/m} = \chi_f(H) = 3
    \end{align*}
\end{claim}
\begin{proof}
In the proof of Claim \ref{claim:claim2}, we showed $H^{\boxtimes m}$ contains a subgraph $K_{3^m}$. Therefore, $\omega(H^{\boxtimes m}) \ge 3^m$ and then  $\inf_m\chi( {H^{\boxtimes m}})^{1/m} \ge 3$. 
$\inf_m\chi( {H^{\boxtimes m}})^{1/m} \le   \inf_m\chi( {H^{\lor m}})^{1/m}$ because $\chi( H^{\boxtimes m})\le \chi(H^{\lor m})$ for all $m\ge 1$.
$\inf_m\chi( {H^{\lor m}})^{1/m} = \chi_f(H)$ holds by Proposition \ref{prop:chi_frac} \ref{prop:OR_Witsenhausen_rate}. $\chi_f(H) = 3$ holds by the computation discussed in Remark \ref{remark:chi_f_line_graph_computation} below. Thus, we have equality.   
\end{proof}
\begin{remark}\label{remark:chi_f_line_graph_computation}
We computed $\chi_f(H)$ in two different ways.  In the first way, we used the built-in method from SageMath\cite{sagemath} to calculate the value of $\chi_f(H)$, which yields $\chi_f(H)=3$.  In the second way, we used the LP in Definition \ref{defi:Fraction_chromatic_number} in Appendix \ref{app:chi_f_LP}. We found a point $(X_I)_{I\in \calI}$ that is feasible in the LP with objective function value $3$.   This implies that $\chi_f(H) \leq 3$. On the other hand, the above argument shows that $\chi(\overline{H^{\lor m}}) \ge 3^m$, which implies $\chi_f(H) \ge 3$. Python code recreating our results is available at \cite{meng2024github}.
\end{remark}


\begin{claim} \label{claim:chi(L(G13)=4} For $H = \frL_D(G_{13})$, we have $\chi(H) =4$.
\end{claim} 

Since $\chi( {G_{13}}) = 4$ \cite{ManinskaR16} and $H = \frL_D(G_{13})$, Proposition \ref{prop:chi_and_xi_of_line_graph} gives $\chi( {H}) \le 4$. To show $\chi(H) >3$, we need the following tools.
\begin{definition}{\it Directed walk.} 
Let $D$ be a directed graph. A $k$-walk of $D$ is a collection of edges $e_1,\dots, e_k$ such that  the head of $e_{i-1}$ is the tail of $e_i$ for $i=2,\dots,k$. 
\end{definition}
\begin{definition}{\it  Edge chromatic number of directed graphs.}  
Let $D$ be a directed graph. A (proper) $k$ edge coloring of $D$ is  a mapping $c: E(D)\mapsto [k]$ such that  there is no $2$-walk whose two edges have the same color. The smallest $k$ such that  there is a $k$-edge-coloring of $D$ is the edge chromatic number $\chi_e(D)$.
\end{definition}

\begin{proposition}
Let $D$ be a directed graph. Then,
    $$\chi_e(D) = \chi \big (\frL_D(D)\big).
    $$
\end{proposition}
\begin{proof}
The result follows from the following equivalences.
$c$ is an edge coloring of $D$ \newline $\Leftrightarrow$ 
 for every 2-walk $(a,b),(b,c)$, we have $c(a,b)\ne c(b,c)$\newline $\Leftrightarrow$ $c$ is a vertex coloring of $\frL_D(D)$.
\end{proof}
 
Viewing each edge of $G_{13}$ as two directed edges with opposite directions, we obtain the directed graph in Fig. \ref{fig:G13}.  We will use the following proposition from \cite{ManinskaR16}; a proof is given for completeness.
\begin{proposition}[Pg. 8 of \cite{ManinskaR16}]\label{prop:G13_strcture}
    Let the graph be $G_{13}$ and $\{M,L,N,P,Q,R,X,Y,Z,W\}$ be vertices in $G_{13}$ (see Fig. \ref{fig:G13}). Let  $u\in\{M,L\},v\in \{N,P\}, w\in \{Q,R\}$ be arbitrary and $\overline{u},\overline{v},\overline{w}$ be the other vertex in $ \{M,L\},\{N,P\} ,\{Q,R\}$ respectively. 
    Then, exactly one of the following will hold.
    \begin{enumerate}
        \item $u,v,w$ have a common neighbor in $\{X,Y,Z,W\}$.
        \item $\overline{u},\overline{v},\overline{w}$ have a common neighbor in $\{X,Y,Z,W\}$.
    \end{enumerate}  
\end{proposition}
\begin{proof}
For $u,v,w$ described above, denote $N(u,v,w)$ and $N(\overline{u},\overline{v},\overline{w})$ as the common neighbor sets of $u,v,w$ and $\overline{u},\overline{v},\overline{w}$  in $\{X,Y,Z,W\}$ respectively. There are eight cases, but we only need to check four cases as the role of $u,v,w$ and $\overline{u},\overline{v},\overline{w}$ are symmetric.
These are checked in Table \ref{table:check_G13}.  
\begin{table}[t]
\centering
\begin{tabular}{|ccc|ccc|c|c|}
\hline
$u$ & $v$ & $w$ & $\overline{u}$ & $\overline{v}$ & $\overline{w}$ & $N(u,v,w)$ & $N(\overline{u},\overline{v},\overline{w})$ \\ \hline
$M$ & $N$ & $Q$ & $L$ & $P$ & $R$ & $\{Z\}$ & $\emptyset$ \\ \hline
$M$ & $N$ & $R$ & $L$ & $P$ & $Q$ & $\emptyset$ & $\{X\}$ \\ \hline
$M$ & $P$ & $Q$ & $L$ & $N$ & $R$ & $\emptyset$ & $\{Y\}$ \\ \hline
$M$ & $P$ & $R$ & $L$ & $N$ & $Q$ & $\{W\}$ & $\emptyset$ \\ 
\hline
\end{tabular}
\vspace{3mm}
\caption{Each possible case in Proposition \ref{prop:G13_strcture}.}
\label{table:check_G13}
\end{table}
\end{proof}

\begin{figure}[t]
    \centering
     
\begin{tikzpicture} 
[
  node distance=0.5cm,
  every node/.style={circle, draw, minimum size=4mm, inner sep=0pt, font = \normalsize},
  every edge/.style={draw, postaction={decorate}, decoration={markings, mark=at position 0.8 with {\arrow{Stealth}}}}
] 

\node (W) at (\xW,\yW) {W};
\node (X) at (\xX,\yX) {X};
\node (Y) at (\xY,\yY) {Y};
\node (Z) at (\xZ,\yZ) {Z};

\node (M) at (\xM,\yM) {M};
\node (L) at (\xL,\yL) {L};
\node (N) at (\xN,\yN) {N};
\node (P) at (\xP,\yP) {P};
\node (Q) at (\xQ,\yQ) {Q};
\node (R) at (\xR,\yR) {R};

\node (A) at (\xA,\yA) {A};
\node (B) at (\xB,\yB) {B};
\node (C) at (\xC,\yC) {C}; 


\draw (A) edge  (B);
\draw (B) edge  (A);
\draw (A) edge  (C);
\draw (C) edge  (A);
\draw (B) edge  (C);
\draw (C) edge  (B);

\draw (L) edge  (M);
\draw (M) edge  (L);
\draw (A) edge  (M);
\draw (M) edge  (A);
\draw (A) edge  (L);
\draw (L) edge  (A);

\draw (N) edge  (P);
\draw (P) edge  (N);
\draw (B) edge  (N);
\draw (N) edge  (B);
\draw (B) edge  (P);
\draw (P) edge  (B);

\draw (Q) edge  (R);
\draw (R) edge  (Q);
\draw (Q) edge  (C);
\draw (C) edge  (Q);
\draw (R) edge  (C);
\draw (C) edge  (R);

\draw (X) edge  (P);
\draw (P) edge  (X);
\draw (Y) edge  (R);
\draw (R) edge  (Y);
\draw (N) edge  (Z);
\draw (Z) edge  (N);
\draw (M) edge  (Z);
\draw (Z) edge  (M);
\draw (L) edge  (Y);
\draw (Y) edge  (L);
\draw (Q) edge[bend left = 10]  (Z);
\draw (Z) edge[bend right = 10]  (Q); 
\draw (W) edge  (M); 
\draw (M) edge  (W);
\draw (X) edge[bend right = 10]  (L);
\draw (L) edge[bend left = 10]  (X);
\draw (R) edge  (W);
\draw (W) edge  (R);
\draw (X) edge  (Q);
\draw (Q) edge  (X);
\draw (P) edge  (W);
\draw (W) edge  (P);
\draw (N) edge[bend left = 10]  (Y);
\draw (Y) edge[bend right = 10]  (N); 
\end{tikzpicture} 
    \caption{$G_{13}$ with each undirected edge viewed as two opposite directed edges.}
    \label{fig:G13}
\end{figure}
\begin{proposition}
    $$
 \chi(\frL_D(G_{13}))= \chi_e(G_{13}) >3
    $$
\end{proposition}
\begin{proof}
We prove this by contradiction.  Suppose $G_{13}$ has a 3-edge-coloring with colors in $\{1,2,3\}$. Let $S=\{A,B,C\}$ and consider the induced subgraph $G_{13}[S]$ (we refer to Fig. \ref{fig:coloring_ABC}).  $A$ has two incoming edges and two outgoing edges. Since there are three colors, there are two edges having the same color among these four edges. The two edges with the same color cannot be exactly one incoming edge and one outgoing edge because that forms a 2-walk whose two edges have the same color. Therefore, these  two edges must be either both incoming edges or both outgoing edges. Since reversing the direction of each edge of $G_{13}$ gives an isomorphism of $G_{13}$, we may assume w.l.o.g. that these two edges are both outgoing edges.   

By symmetry of $1,2,3$, we may assume $(A,B),(A,C)$ have color $1$ and $(B,A)$ has color $2$ and $(C,A)$ has color $3$. Since $(A,B)$ has color 1 and $(C,A)$ has color 3, $(B,C)$ must have color 2. For a similar reason, $(C,B)$ must have color 2. Then, the induced subgraph $G_{13}[S]$ are colored as Fig. \ref{fig:coloring_ABC}.  $B$ has incoming edges of colors $1$ and $3$. It implies that all outgoing edges of $B$ must have color $2$. Similarly,
all outgoing edges of $A$ and $C$ must have color $1$ and $3$ respectively. 
\begin{figure}[t]
    \centering
    \scalebox{0.7}{%
    \begin{tikzpicture}
\node[node_1] (A) at (\xA,\yA) {A};
\node[node_1] (B) at (\xB,\yB) {B};
\node[node_1] (C) at (\xC,\yC) {C};

\node[node_2] (BC) at (\xBC,\yBC) {2};
\node[node_2] (CB) at (\xCB,\yCB) {3};
\node[node_2, rotate = -60] (AC) at (\xAC,\yAC) {1};
\node[node_2, rotate = -60] (CA) at (\xCA,\yCA) {3};
\node[node_2, rotate = 60] (BA) at (\xBA,\yBA) {2};
\node[node_2, rotate = 60] (AB) at (\xAB,\yAB) {1};

\draw  (A) edge[edge_1]  (AB);
\draw  (AB) edge[edge_2]  (B);
\draw  (B) edge[edge_1]  (BA);
\draw  (BA) edge[edge_2]  (A);

\draw  (C) edge[edge_1]  (CB);
\draw  (CB) edge[edge_2]  (B);
\draw  (B) edge[edge_1]  (BC);
\draw  (BC) edge[edge_2]  (C);

\draw  (A) edge[edge_1]  (AC);
\draw  (AC) edge[edge_2]  (C);
\draw  (C) edge[edge_1]  (CA);
\draw  (CA) edge[edge_2]  (A);

\end{tikzpicture}%
    }
    \caption{Edge coloring of $G_{13}[A,B,C]$. The color is the label in the middle of each edge.}
    \label{fig:coloring_ABC}
\end{figure} 
Now consider the subgraph induced by vertices $A,B,C,M,N,L,P,Q,R$. Since all outgoing edges of $A,B,C$ must have color $1,2,3$ respectively. We must color these outgoing edges with the corresponding colors like Fig. \ref{fig:subgraph_phase2}. Consider Fig. \ref{fig:subgraph_phase2}. Denote the color of  $(L,M)$ (the red "?") $a\in \{1,2,3\}$. $a\ne 1$ because otherwise $(A,L),(L,M)$ is a 2-walk of color $1$. 

If $a=2$, then $M$ has 2 incoming edges of color $1$ and $2$ respectively. This implies that all outgoing edges of $M$ are of color $3$. In particular, $(M,L)$ is of color $3$ and $L$ has 2 incoming edges of color $1$ and $3$ respectively. It means that all outgoing edges of $L$ are of color $2$. The case $a=3$ is similar. Therefore, the following holds. 
\begin{enumerate}
    \item For each $v\in \{M,L\}$, the outgoing edges of $v$ have the same color, call it $\psi(v)$.
    \item Let $\overline{v}$ be the other vertex in $\{M,L\}$ and the corresponding color be $\psi(\overline{v})$. We have that $\{\psi(v)\}\cup \{\psi(\overline{v})\}\cup\{1\}= \{1,2,3\}$.
\end{enumerate} 
Similar arguments hold for the case of $\{N,P\}$ and $\{Q,R\}$. Let $\psi(X)$ be the color of the outgoing edges for $X\in \{M,N,L,P,Q,R\}$.  Let $u\in \{M,L\},v\in\{N,P\},w\in\{Q,R\}$ such that $\psi(u) = 2, \psi(v)=3,\psi(w)=1$ and $\overline{u},\overline{v},\overline{w}$ be the other vertex in $ \{M,L\},\{N,P\} ,\{Q,R\}$ respectively. Note that $\psi(\overline{u}) =3, \psi(\overline{v}) =1, \psi(\overline{w}) =2$. By Proposition \ref{prop:G13_strcture}, we have that exactly one of the following holds.
    \begin{enumerate}
        \item $u,v,w$ have a common out-neighbor in $\{X,Y,Z,W\}$.
        \item $\overline{u},\overline{v},\overline{w}$ have a common out-neighbor in $\{X,Y,Z,W\}$.
    \end{enumerate}  
Let $z\in\{X,Y,Z,W\}$ be the common out-neighbor above. Since the colors of the incoming edges from either $u,v,w$ or $\overline{u},\overline{v},\overline{w}$ equals $\{1,2,3\}$, outgoing edges of $z$ have no color to choose. This gives the desired contradiction.
\begin{figure}[t]
    \centering
    \scalebox{0.55}{%
      
\begin{tikzpicture}
\node[node_1] (A) at (\xA,\yA) {A};
\node[node_1] (B) at (\xB,\yB) {B};
\node[node_1] (C) at (\xC,\yC) {C}; 

\node[node_2] (BC) at (\xBC,\yBC) {2};
\node[node_2] (CB) at (\xCB,\yCB) {3};
\node[node_2, rotate = -60] (AC) at (\xAC,\yAC) {1};
\node[node_2, rotate = -60] (CA) at (\xCA,\yCA) {3};
\node[node_2, rotate = 60] (BA) at (\xBA,\yBA) {2};
\node[node_2, rotate = 60] (AB) at (\xAB,\yAB) {1};

\draw  (A) edge[edge_1]  (AB);
\draw  (AB) edge[edge_2]  (B);
\draw  (B) edge[edge_1]  (BA);
\draw  (BA) edge[edge_2]  (A);

\draw  (C) edge[edge_1]  (CB);
\draw  (CB) edge[edge_2]  (B);
\draw  (B) edge[edge_1]  (BC);
\draw  (BC) edge[edge_2]  (C);

\draw  (A) edge[edge_1]  (AC);
\draw  (AC) edge[edge_2]  (C);
\draw  (C) edge[edge_1]  (CA);
\draw  (CA) edge[edge_2]  (A);

\node[node_1] (M) at (\xMM,\yMM) {M};
\node[node_1] (L) at (\xLL,\yLL) {L};

\node[node_2] (ML) at (\xML,\yML) {$\hphantom{1}$};
\node[node_2, text=red,font=\Large] (LM) at (\xLM,\yLM) {?};
\node[node_2, rotate = 60] (AL) at (\xAL,\yAL) {1};
\node[node_2, rotate = 60] (LA) at (\xLA,\yLA) {$\hphantom{1}$};
\node[node_2, rotate = -60] (MA) at (\xMA,\yMA) {$\hphantom{1}$};
\node[node_2, rotate = -60] (AM) at (\xAM,\yAM) {1};

\draw  (A) edge[edge_1]  (AL);
\draw  (AL) edge[edge_2]  (L);
\draw  (L) edge[edge_1]  (LA);
\draw  (LA) edge[edge_2]  (A); 

\draw  (A) edge[edge_1]  (AM);
\draw  (AM) edge[edge_2]  (M);
\draw  (M) edge[edge_1]  (MA);
\draw  (MA) edge[edge_2]  (A); 

\draw  (M) edge[edge_1]  (ML);
\draw  (ML) edge[edge_2]  (L);
\draw  (L) edge[edge_1]  (LM);
\draw  (LM) edge[edge_2]  (M);

\node[node_1] (N) at (\xNN,\yNN) {N};
\node[node_1] (P) at (\xPP,\yPP) {P};

\node[node_2] (NB) at (\xNB,\yNB) {$\hphantom{1}$};
\node[node_2] (BN) at (\xBN,\yBN) {2};
\node[node_2, rotate = 60] (BP) at (\xBP,\yBP) {2};
\node[node_2, rotate = 60] (PB) at (\xPB,\yPB) {$\hphantom{1}$};
\node[node_2, rotate = -60] (NP) at (\xNP,\yNP) {$\hphantom{1}$};
\node[node_2, rotate = -60] (PN) at (\xPN,\yPN) {$\hphantom{1}$};

\draw  (B) edge[edge_1]  (BN);
\draw  (BN) edge[edge_2]  (N);
\draw  (N) edge[edge_1]  (NB);
\draw  (NB) edge[edge_2]  (B); 

\draw  (B) edge[edge_1]  (BP);
\draw  (BP) edge[edge_2]  (P);
\draw  (P) edge[edge_1]  (PB);
\draw  (PB) edge[edge_2]  (B); 

\draw  (P) edge[edge_1]  (PN);
\draw  (PN) edge[edge_2]  (N);
\draw  (N) edge[edge_1]  (NP);
\draw  (NP) edge[edge_2]  (P);

\node[node_1] (Q) at (\xQQ,\yQQ) {Q};
\node[node_1] (R) at (\xRR,\yRR) {R};

\node[node_2] (RC) at (\xRC,\yRC) {$\hphantom{1}$};
\node[node_2] (CR) at (\xCR,\yCR) {3};
\node[node_2, rotate = -60] (CQ) at (\xCQ,\yCQ) {3};
\node[node_2, rotate = -60] (QC) at (\xQC,\yQC) {$\hphantom{1}$};
\node[node_2, rotate = 60] (RQ) at (\xRQ,\yRQ) {$\hphantom{1}$};
\node[node_2, rotate = 60] (QR) at (\xQR,\yQR) {$\hphantom{1}$};

\draw  (C) edge[edge_1]  (CR);
\draw  (CR) edge[edge_2]  (R);
\draw  (R) edge[edge_1]  (RC);
\draw  (RC) edge[edge_2]  (C); 

\draw  (C) edge[edge_1]  (CQ);
\draw  (CQ) edge[edge_2]  (Q);
\draw  (Q) edge[edge_1]  (QC);
\draw  (QC) edge[edge_2]  (C); 

\draw  (Q) edge[edge_1]  (QR);
\draw  (QR) edge[edge_2]  (R);
\draw  (R) edge[edge_1]  (RQ);
\draw  (RQ) edge[edge_2]  (Q);

\end{tikzpicture}%
    }
    \caption{Partial coloring of the subgraph induced by $A,B,C,M,N,L,P,Q,R$. An edge with an empty label in the middle means the edge is not colored yet.}
    \label{fig:subgraph_phase2}
\end{figure}
\end{proof}

\begin{remark}
    We also numerically verified that $\chi(\frL_D(G_{13}))=4$ (code available at \cite{meng2024github}).  
\end{remark}

\section{VLC rate and graph entropies}\label{sec:VLC_rate_graph_entropy} 

In Section \ref{subsec:problem_formulation}, we formally defined a classical VLC for our problem and its corresponding rate. In this section, we present further results on the VLC rate. In particular, we first discuss the connection between the asymptotic optimal VLC rate ({\it cf.} Definition \ref{def:VLC_rate}) and appropriately defined graph entropies. Following this, we discuss previous definitions of VLCs for the problem of source coding with side information \cite{AlonO_96} and 
point out how our definition captures the more general function computation with side information problem. Lastly, we discuss some properties of graph entropy that are useful for our rate analyses.

In the discussion below, we use standard definitions within source coding ({\it cf.} Chapter 5 of \cite{CoverT05}).\\  Let $H_D(X),H_D(X|Y),I_D(X,Y)$ be the $D$-ary entropy, conditional entropy and mutual information. If we omit the subscript $D$, it means the $D=2$. 

\begin{definition}{ \it $D$-ary  Chromatic entropy.}\label{def:chromatic_entropy}
Let $G=(V,E)$ be a graph and $p_X(\cdot)$ be a distribution over $V$.  The $D$-ary chromatic
entropy $H_{\chi,D}(G,X)$ is defined by
\begin{align*}
    H_{\chi,D}(G,X)=\inf \{H_D(c(X))|\text{ $c$ is a coloring of $G$}\}.
\end{align*}
If $D=2$, we omit subscript $2$ and denote $H_{\chi}(G,X)$ (Section II of \cite{AlonO_96}).

\end{definition}

\begin{theorem}\label{thm:VLC_optimal_rate}
Let $G^{(m)}$ be the $m$-instance $f$-confusion graph and let $p_\bfX(\cdot)$ (induced by $p_{\bfX\bfY}(\cdot,\cdot)$) be the distribution over $V(G^{(m)})$, and the alphabet be $D$-ary.  The optimal VLC rate  is
    \begin{align*}
        R_{VLC}    = \lim_m\frac{1}{m} H_{\chi,D}(G^{(m)},\bfX).
    \end{align*}
\end{theorem}
 
The literature on classical source coding also uses the terminology of prefix-free and non-singular encodings of sources \cite{CoverT05}. To avoid any confusion with our definitions that are stated in the context of function computation with side information we refer to these encodings as regular. In particular, we have the following definitions.
\begin{definition}
Let $X$ be a random variable taking value on a finite set $\calX$, and $\phi_{e}:\calX\mapsto\{0,\dots, D-1\}^*$ be an $D$-ary encoding map. We say that $\phi_{e}$ is regular non-singular if  $x_1\ne x_2 \implies \phi_{e}(x_1)\ne \phi_{e}(x_2)$, and $\phi_{e}$ is regular prefix-free if $\phi_e(x_1)$ is not a prefix of $\phi_e(x_2)$ for all distinct $x_1,x_2\in \calX$.

Furthermore, we denote $l_{r.n.s.}(X,D),l_{r.p.f.}(X,D)$ the minimum expected length of any regular non-singular/regular prefix-free encoding of $X$ over a $D$-ary alphabet. 
\end{definition}



\begin{lemma}  The following holds. \label{lemma:11_pf_encoding_bounds}
\begin{enumerate}[label=(\alph*)]
    \item $H_D(X)\le l_{r.p.f.}(X,D)\le H_D(X)+1$ (see pg. 87 of \cite{CoverT05}).\label{lemma:11_pf_encoding_bounds_part_1}
    \item $H_D(X)  - \log_D\left(\frac{D}{D-1}\right) - \log_D(1+\ln |\calX|) \le l_{r.n.s.}(X,D)$. This bound essentially follows from \cite{Dunham80} who proved it for the binary case. We prove a $D$-ary version for completeness in Appendix H.  \label{lemma:11_pf_encoding_bounds_part_2}
\end{enumerate} 

Our proof of Theorem \ref{thm:VLC_optimal_rate} is similar to the one in \cite{AlonO_96} and appears below. 
\begin{proof}[Proof of Theorem \ref{thm:VLC_optimal_rate}] 
    Let $\phi_e$ be the $m$-instance code in Definition \ref{def:VLC_rate} achieving $R^{(m)}_{VLC}$. Then, $\phi_e$ is a coloring of $G^{(m)}$ by definition. The "color" $\phi_e(\bfX)$ is a random variable taking values in the image set of $\phi_e$ over $V(G^{(m)})$. Its  distribution function is   $\Pr (\,\phi_e(\bfX)=c\,) = \sum_{\substack{\bfx\in \phi_e^{-1}(c)}} p_{\bfX}(\bfx)$ where $\phi_e^{-1}(c)$ is the set of $\bfx$ such that $\phi_e (\bfx)=c$, and $p_{\bfX}(\cdot)$ is the marginal distribution of $p_{\bfX\bfY}(\cdot,\cdot)$.
    
    The identity map on  $\phi_e(\bfX)$ is a regular non-singular encoding of itself. We consider the identity map so that Lemma \ref{lemma:11_pf_encoding_bounds} \ref{lemma:11_pf_encoding_bounds_part_2} can be applied. This gives 
    \begin{align*} 
    R^{(m)}_{VLC}\ge &\frac{1}{m}l_{r.n.s.}(\phi_e(\bfX),D)\\
    \ge & \frac{1}{m}\big ( H_D(\phi_e(\bfX)) - \log_D\left(\frac{D}{D-1}\right)  - \log_D(1+\ln |\text{Im}(\phi_e)|)\big )\\
    \ge & \frac{1}{m}\big ( H_D(\phi_e(\bfX))- \log_D\left(\frac{D}{D-1}\right)   - \log_D(1+\ln |V(G)|^m)    \big )\\
    \ge & \frac{1}{m}\big ( H_{\chi,D}(G^{(m)},\bfX)- \log_D\left(\frac{D}{D-1}\right)   - \log_D(1+\ln |V(G)|^m)    \big )
    \end{align*}
where the first inequality holds because identify map on $\phi_e(\bfX)$ is a regular non-singular encoding of itself, the second inequality holds by the lower bound of Lemma \ref{lemma:11_pf_encoding_bounds}  \ref{lemma:11_pf_encoding_bounds_part_2} and $\text{Im}(\phi_e)$ denotes the image set of $\phi_e$, the third inequality holds because the number of colors in $\phi_e$ is at most size of vertex set, i.e. $|\text{Im}(\phi_e)|\le |V(G^{(m)})|=|V(G )|^m$, the fourth inequality holds because $\phi_e(\bfX)$ is a coloring, and $H_D(\phi_e(\bfX))\ge H_D(G^{(m)},\bfX)$.

On the other hand, let $C$ be any coloring of $G$ such that $H_D(C(\bfX))$ achieves $H_{\chi,D}(G^{(m)},\bfX)$. Any regular prefix-free encoding of   $C(\bfX)$ corresponds to a  code in Definition \ref{def:VLC_rate}. Therefore, 
    \begin{align*}
        R^{(m)}_{VLC}\le  \frac{1}{m}l_{r.p.f.}(C(\bfX),D)\le \frac{1}{m}\big ( H_{ D}(C(\bfX))+1\big ) = \frac{1}{m}\big ( H_{\chi,D}(G^{(m)},\bfX)+1\big ).
    \end{align*}
Combining the inequalities and taking limits we have,
\begin{align*}
  &R_{VLC}  =\lim_m \frac{1}{m} H_{\chi,D}(G^{(m)},\bfX).
\end{align*}
Furthermore, $\lim_m\frac{1}{m} H_{\chi,D}(G^{(m)},\bfX) = \inf_m\frac{1}{m} H_{\chi,D}(G^{(m)},\bfX)$ because of Fekete Lemma \cite{Fekete_23}. This means the following. Let $G^{(m)},G^{(n)}$ be associated with $\bfX_1,\bfX_2$ respectively where $\bfX_1,\bfX_2$ are independent. Suppose $C_1,C_2$ are colorings of $G^{(m)}$ and $G^{(n)}$  such that $H_D(C_1(\bfX_1)),H_D(C_1(\bfX_2))$ achieves $H_{\chi,D}(G^{(m)},\bfX_1),H_{\chi,D}(G^{(n)},\bfX_2)$  respectively. Let $\bfX=(\bfX_1,\bfX_2)$ where $\bfX_1,\bfX_2$ are the first $m$ and last $n$ components of $\bfx$ respectively. Then, define
\begin{align*}
    C_1\times C_2(\bfX) := (C_1(\bfX_1),C_2(\bfX_2)),
\end{align*}
which is a valid coloring for $G^{(m+n)}$. 
Use $p(C_1\times C_2(\bfx)),p(C_1(\bfx)), p(C_2(\bfx))$ as a shorthand for $p_{C_1\times C_2(\bfX)}(C_1\times C_2(\bfx)),p_{C_1(\bfX)}(C_1(\bfx)), p_{C_2(\bfX)}(C_2(\bfx))$ respectively. Then 
\begin{align*}  
    H_{\chi, D}(G^{(m+n)},\bfX )\le&  H_D(C_1\times C_2(\bfX)) \\
    =&\sum_{C_1\times C_2(\bfx)} -p(C_1\times C_2(\bfx))\log_D \big( C_1\times C_2(\bfx) \big)\\
    =&\sum_{C_1(\bfx_1)}\sum_{C_2(\bfx_2)} -p\big (C_1(\bfx_1)\big )p\big ( C_2(\bfx_2)\big )\Big (\log_D \big( C_1 (\bfx_1) )\big)+\log_D \big( C_2(\bfx_2) \big)\Big )\\
    =&\sum_{C_1(\bfx_1)}-p\big (C_1(\bfx_1)\big ) \log_D \big( C_1( \bfx_1) )\big) +  \sum_{C_2(\bfx_2)} - p\big ( C_2(\bfx_2)\big ) \log_D \big( C_2(\bfx_2) \big) \\
    =& H_D(C_1(\bfX_1)) + H_D(C_2(\bfX_2))\\ 
    =&H_{\chi,D}(G^{(m)},\bfX_1)+H_{\chi,D}(G^{(n)},\bfX_2). 
\end{align*}
Then, Fekete Lemma implies $\inf_m\frac{1}{m} H_{\chi,D}(G^{(m)},\bfX)=\lim_m\frac{1}{m} H_{\chi,D}(G^{(m)},\bfX)$.
\end{proof}
\end{lemma}  
\noindent From the proof above, we can see that the non-singular and prefix-free codes have the same asymptotic rate.   


We emphasize that we consider only prefix-free codes even in the single-instance case because they can always be extended to instantaneously decodable codes ({\it cf.} Section \ref{subsubsec:VLC_setting}). The upcoming Example \ref{example:C5_VLC} shows that non-singular codes cannot be extended in this manner in general. We point out that the work of \cite{AlonO_96} places different requirements on a prefix-free variable length code in the source coding with side information context ({\it cf.} Remark \ref{remark:AlonO_prefix-free}). Example \ref{example:C5_VLC} also shows that their definition may not allow for extension to instantaneously decodable codes, and hence our definition is the correct one to use. 

One way to arrive at a prefix-free variable length code is to simply choose a coloring of the $m$-instance graph $G^{(m)}$ and then encode the colors using a regular prefix-free encoding. Example \ref{example:C5_VLC_f} demonstrates that there are variable length prefix-free codes ({\it cf.} Definition \ref{def:pf_code}) that do not belong to these class of encodings. 


\begin{table*}[h]
\centering 
\begin{minipage}{.32\textwidth}
  \centering
  \begin{tabular}{c|c|ccccc}
    \multicolumn{2}{c|}{} & \multicolumn{5}{c}{\( x \)} \\
    \cline{3-7}
    \multicolumn{2}{c|}{\multirow{-2}{*}{\( P_{XY}(x,y) \)}} & 1 & 2 & 3 & 4 & 5 \\
    \hline
    \multirow{5}{*}{\( y \)} 
    & 1 & $\frac{1}{2}$ & $\frac{1}{2}$ & * & * & * \\
    & 2 & * & $\frac{1}{2}$ & $\frac{1}{2}$ & * & * \\
    & 3 & * & * & $\frac{1}{2}$ & $\frac{1}{2}$ & * \\
    & 4 & * & * & * & $\frac{1}{2}$ & $\frac{1}{2}$ \\
    & 5 & $\frac{1}{2}$ & * & * & * & $\frac{1}{2}$ \vspace{3mm}
  \end{tabular} 
  \caption{$P_{XY}(x,y)$ for the scenario of $\tilde f$}
  \label{table:simple_eg1_pxy}
\end{minipage}%
\quad\quad
\begin{minipage}{.32\textwidth}
  \centering
  \begin{tabular}{c|c|ccccc}
    \multicolumn{2}{c|}{} & \multicolumn{5}{c}{\( x \)} \\
    \cline{3-7}
    \multicolumn{2}{c|}{\multirow{-2}{*}{\( P_{XY}(xy) \)}} & 1 & 2 & 3 & 4 & 5 \\
    \hline
    \multirow{5}{*}{\( y \)} 
    & 1 & $\frac{1}{3}$ & $\frac{1}{3}$ & $\frac{1}{3}$ & * & * \\
    & 2 & * & $\frac{1}{3}$ & $\frac{1}{3}$ & $\frac{1}{3}$ & * \\
    & 3 & * & * & $\frac{1}{3}$ & $\frac{1}{3}$ & $\frac{1}{3}$ \\
    & 4 & $\frac{1}{3}$ & * & * & $\frac{1}{3}$ & $\frac{1}{3}$ \\
    & 5 & $\frac{1}{3}$ & $\frac{1}{3}$ & * & * & $\frac{1}{3}$ \vspace{3mm}
  \end{tabular}
  \caption{$P_{XY}(x,y)$ for the scenario of  $\tilde g$}
  \label{table:simple_eg2_pxy}
\end{minipage}  
\end{table*}  
\begin{example}\label{example:C5_VLC}
Consider the computation of the function $\tilde g$ ({\it cf.} Table \ref{table:simple_eg2}) where the probability distribution is given in Table \ref{table:simple_eg2_pxy}. Suppose that Alice uses the encoding in Fig. \ref{fig:c5_vlc_encoding}. 
It is also a non-singular encoding according to our Definition \ref{def:ns_code}.
\begin{figure}[h]
    \centering
    \begin{tikzpicture} 
[
  node distance=0.5cm,
  every node/.style={ draw, minimum size=4mm, inner sep=1pt, font = \normalsize, text width = 2 cm, align = center},
  every edge/.style={draw, -} 
]

\node[draw=black] (1) at (0,3) {$x=1$ \\ $\phi_e(x) = 1$};
\node[draw=black] (2) at (-2.8531,0.9271) {$x=2$ \\ $\phi_e(x) = 0$};
\node[draw=black] (3) at (-1.7634,-2.4271) {$x=3$ \\ $\phi_e(x) = 1$};
\node[draw=black] (4) at (1.7634,-2.4271) {$x=4$ \\ $\phi_e(x) = 00$}; 
\node[draw=black] (5) at (2.8531,0.9271) {$x=5$ \\ $\phi_e(x) = 01$}; 

\draw  (1) edge (2);
\draw  (3) edge (2);
\draw  (3) edge (4);
\draw  (4) edge (5);
\draw  (5) edge (1);
\end{tikzpicture}
    \caption{Variable length code employed by Alice.}
    \label{fig:c5_vlc_encoding}
\end{figure}
Since adjacent vertices are not prefixes of each other, this is an encoding that meets the criterion of \cite{AlonO_96} and is a non-singular code. 
In the single-instance case, this is a valid code. However, if we consider extension of the code, the encoding causes error. Suppose Bob observes $[2 ~ 4]$. Consider the following two cases. 
\begin{enumerate}
    \item Suppose that Alice observes $[2~ 5]$. She uses the extension of $\phi_e$ and sends $(\phi_e(2),\phi_e(5) ) = 001$. The function value should be $[\tilde g(2,2)~\tilde g(5,4)]=[1~0]$. 
    \item Suppose that Alice observes $[4~ 1]$. She uses the extension of $\phi_e$ and sends $(\phi_e(4),\phi_e(1) ) = 001$. The function value should be $[\tilde g(4,2)~\tilde g(1,4)]=[1~1]$. 
\end{enumerate}
We observe that the function evaluations in the cases above are different, but the messages sent to Bob are all the same, i.e. "001". Therefore, Bob cannot unambiguously distinguish between the cases.
\end{example}

\begin{example}  \label{example:C5_VLC_f} Consider the computation of the function $\tilde f$ ({\it cf.} Table \ref{table:simple_eg1}) where the probability distribution is given in Table \ref{table:simple_eg1_pxy}. Suppose  Alice still uses the same encoding as the previous example, i.e., the encoding in Fig. \ref{fig:c5_vlc_encoding}. Compared to Example \ref{example:C5_VLC}, the same encoding meets the prefix-free code criterion ({\it cf.} Section \ref{subsubsec:VLC_setting}) because  $(X,Y)$ is taking value in a smaller set. Therefore, this code can be extended. 

This encoding is a $4$-coloring of the graph. However, it is not a regular prefix-free encoding of the colors. For example, $\phi_e(2)=0$ is a prefix of $\phi_e(4)=00$ and $\phi_e(5)=01$.  
\end{example}

The following results will be used to bound the optimal VLC rates. See  Appendix \ref{app:proof_lemma_graph_entropy} for the proof of Lemma \ref{lemma:prefix_free_OR_prod}, Lemma \ref{lemma:ont_shot_entropy_OR_prod} and Lemma \ref{lemma:graph_entropy} \ref{lemma:chromatic_entropy_chi_upper_bound} and \ref{lemma:korner_entropy_chi_f_upper_bound}.
\begin{lemma}\label{lemma:prefix_free_OR_prod}
For a $f$-confusion graph $G$, suppose that all $(x,x')\notin E(G)$ are due to C2.  Let $\phi_e$ be an arbitrary prefix-free encoding function for $G^{(m)}$, and  $\text{Im}(\phi_e)$ be the image set of encoding $\phi_e$. If $(\zeta_1,\zeta_2)$ is such that  $\zeta_1\ne \zeta_2$ and $\zeta_1,\zeta_2\in \text{Im}(\phi_e)$,  then $\zeta_1$ is not a prefix of $\zeta_2$ and vice versa. 
\end{lemma} 
\begin{lemma}\label{lemma:ont_shot_entropy_OR_prod}
Let the alphabet be $D$-ary. For a $f$-confusion graph $G$, suppose that all $(x,x')\notin E(G)$ are due to C2. Then $R_{VLC}^{(1)} \ge H_{\chi,D}(G,X)$ where $X$ has distribution function $p_X(\cdot)$.
\end{lemma}

\begin{definition}{\it $D$-ary K\"orner's graph entropy.} 
\label{def:korner_entropy}
Let $\Gamma(G)$ denote the set of all independent sets of $G$. Let $X$ be distributed over the vertices of $G$. Then, the $D$-ary K\"orner's graph entropy $H_{\kappa}(G,X)$ is 
\begin{align*}
    H_{\kappa,D}(G,X)=\min_{X\in Z\in\Gamma(G)}I_D(X;Z)
\end{align*}
where   the minimization is taken over all conditional distributions $p_{Z|X}(z|x)$ with $Z$ taking value in $\Gamma(G)$. The constraint "$X\in Z\in\Gamma(G)$"  means that  the random variable $X$ belongs to the random independent set $Z$ with probability one. If $D=2$, then we just denote it $H_{\kappa}(G,X)$.
\end{definition} 
\begin{lemma} \label{lemma:graph_entropy}
\begin{enumerate}[label=(\alph*)]
    \item $\lim_{m\to \infty}\frac{1}{m}H_{\chi,D}(G^{\lor m},\bfX) = H_{\kappa,D}(G,X)$ (Theorem 5 of \cite{AlonO_96}). \label{lemma:chromatic_entropy_OR_prod_limit}
    \item If $\bfX$ is uniformly distributed over   $V(G^{(m)})$, then $H_{\chi,D}(G^{(m)},\bfX) \ge \log_D \left(\frac{|V(G^{(m)})|}{\alpha(G^{(m)})}\right)$ (Lemma 10 of \cite{AlonO_96}).\label{lemma:chromatic_entropy_ind_set_lower_bound}
    \item  $H_{\chi,D}(G^{(m)},\bfX) \le \log_D \chi(G^{(m)}).$   \label{lemma:chromatic_entropy_chi_upper_bound}
    \item  $H_{\kappa,D}(G^{(m)},\bfX) \le H_{\chi,D}(G^{(m)},\bfX)$ (Section V-B of \cite{AlonO_96}). \label{lemma:korner_vs_chromatic_entropy}
    \item $H_{\kappa,D}(G^{(m)},\bfX) \le \log_D \chi_f(G^{(m)})$. \label{lemma:korner_entropy_chi_f_upper_bound}
\end{enumerate}
\end{lemma} 
For  Lemma \ref{lemma:graph_entropy} \ref{lemma:chromatic_entropy_OR_prod_limit}, \ref{lemma:chromatic_entropy_ind_set_lower_bound}, \ref{lemma:korner_vs_chromatic_entropy}, \cite{AlonO_96} proved the binary case, i.e., $D=2$. The binary case immediately implies the $D$-ary case because $H(X) = \frac{\ln D}{\ln 2}H_D(X),I(X;Y) = \frac{\ln D}{\ln 2}I_D(X;Y).$

\section{Analyzing the behavior of quantum advantage with VLC}\label{sec:analysis_bahavior_quantum_advantage_vlc} 

In this section, we compare variable length prefix-free codes with quantum codes. The alphabet is binary throughout. 
The comparison is for both single-instance and multiple instances. In the single-instance setting, we compare $\log_2\xi(\overline{G})$ and $R^{(1)}_{VLC}$ ({\it cf.} Definition \ref{def:VLC_rate}). In the multiple-instance setting, we compare $R_{\text{quantum}}$ ({\it cf.} Theorem \ref{thm:theorem4}) and $R_{VLC}$ ({\it cf.} Definition \ref{def:VLC_rate}).  There are three possibilities.  The rate of VLC can be strictly less than, or strictly larger than, or equal to  the quantum rate. We call them VLC advantage, quantum advantage and equivalent rates respectively. We will revisit examples that are used in the quantum vs. FLC comparison where $p_X(\cdot)$ is uniform. While some of them have a similar result as the FLC case, there are  examples leading to different results.

\subsection{VLC advantage   in both single-instance and multiple-instance  setting\label{sec:VLC_k3}}
 
As pointed out in Remark \ref{remark:VLC_depends_on_p_X}, the VLC rate also depends on $p_{\bfX}(\cdot)$. On the other hand, the quantum rate and the FLC rate only depend upon $G^{(m)}$. Thus, the rate of the VLC rate can be very small for a ``favorable'' marginal distribution $p_{\bfX}(\cdot)$. We demonstrate this by means of an example below.


Let $\epsilon>0$ be sufficiently small, $\calX=\{x_0,x_1,x_2\}$, and $\calY=\{y_0\}$ and $f(x,y) = x$.  We set 
\begin{align*}
 p_{XY}(x,y) =&\begin{cases}
     1-\epsilon,\text{ if }(x,y)=(x_0,y_0),\\
     \epsilon/ 2 \text{ if }(x,y) = (x_1,y_0),\\
     \epsilon/2  \text{ if }(x,y) = (x_2,y_0),
 \end{cases}
\end{align*} 
The single-instance $f$-confusion graph is $K_3$, and the multiple-instance $f$-confusion graph is $K_3^{\boxtimes m}=K_3^{\lor m}$. Therefore, the single-instance and multiple-instance quantum rate is $\log_2 3 = 1.58$.

In comparison, in the single-instance setting, Alice can use the following encoding:
\begin{align*}
    x_0\mapsto 0, x_1\mapsto 10, x_2\mapsto 11. 
\end{align*}
The expected code length is $(1-\epsilon)\cdot 1 + \frac{\epsilon}{2}\cdot 2+ \frac{\epsilon}{2}\cdot 2 = 1+\epsilon < \log_2 3$ for sufficiently small $\epsilon$. In the multiple-instance case, suppose Alice sends $\bfx$ using the Huffman code, then the asymptotic rate is $H(1-\epsilon,\epsilon/2,\epsilon/2) = H_2(\epsilon) + \epsilon$. This can be made arbitrarily close to $0$ by letting $\epsilon \to 0$.

\subsection{VLC advantage in single-instance setting, but equivalent rate  in multiple-instance setting.\label{sec:VLC_C5_AND}}
Consider function $\tilde f$. Let $p_{XY}(\cdot,\cdot)$  be shown in Table \ref{table:simple_eg1_pxy}. Recall that the $f$-confusion graph is $C_5^{\boxtimes m}$. From Proposition \ref{prop:C5_AND_rates}, we have $\log_2\xi(\overline{C_5}) =\log_23\approx 1.585$, $R_{\text{quantum}}(\tilde f)  =\frac{1}{2}\log_2 5$. The following proposition proved in Appendix \ref{app:proof_VLC_rate_c5} shows $\tilde f$ is the example we need.
\begin{proposition}\label{prop:VLC_rate_c5_AND}
\begin{enumerate}
    \item $R^{(1)}_{VLC}(\tilde f) = 1.4$.
    \item $R_{VLC}(\tilde f)=\frac{1}{2}\log_2 5$.
\end{enumerate}
\end{proposition}


\subsection{Quantum advantage in single-instance setting, but equivalent rate  in multiple-instance setting.\label{sec:VLC_C5_OR}}
Consider function $\tilde g$. Let $p_{XY}(\cdot,\cdot)$  be shown in Table \ref{table:simple_eg2_pxy}. Then, $f$-confusion graph is $C_5^{\lor m}$. From Proposition \ref{prop:C5_AND_rates}, we have $\log_2\xi(\overline{C_5}) =\log_23\approx 1.585$,   $R_{\text{quantum}}(\tilde g) =  \log_2 \frac{5}{2}$. The following proposition proved in Appendix \ref{app:proof_VLC_rate_c5} shows  $\tilde g$ is the example we need.
\begin{proposition}\label{prop:VLC_rate_c5_OR}
\begin{enumerate}
    \item $ R^{(1)}_{VLC}(\tilde g) = 1.6$.
    \item  $R_{VLC}(\tilde g) =  \log_2 \frac{5}{2}$.
\end{enumerate}
\end{proposition}
\begin{remark}
   For Sections \ref{sec:VLC_C5_AND} and \ref{sec:VLC_C5_OR}, if we allow ternary ($D=3$) alphabet for classical communication and qutrits for quantum communication, then they would have same single-instance rate, which equals 1. This is because the rate is at least one. For quantum communication, $\xi(C_5) =3$ implies only one qutrit is needed. For VLC, the $3$-coloring of $C_5$ already induces an encoding of expected length 1.
   %
\end{remark} 
\subsection{Quantum advantage both in single-instance and in multiple-instance setting.\label{sec:VLC_quantum_in_single_and_multi_instance}}

In this section, we use the characterization of the asymptotic VLC rate corresponding to $G^{\lor m}$ as the {K\"orner} entropy ({\it cf.} Definition \ref{def:korner_entropy}). A numerical computation of  {K\"orner} graph entropy requires optimizing over the class of independent sets in $G$. For smaller graphs, such an enumeration is feasible. Moreover, the work of \cite{HarangiNB24} provides a Blahut-Arimoto-like algorithm for computing  K\"orner's graph entropy. We use their method in our calculations.

\subsubsection{Example associated with OR product.}  \label{sec:VLC_G13}
By Theorem \ref{thm:construct_function_based_on_G} and Remark \ref{remark:free_to_choose_p_x}, there exists a function and a joint distribution such that its $m$-instance $f$-confusion graph is $G_{13}^{\lor m}$ ({\it cf.} \ref{subsubsec:G13}) and $p_X(\cdot)$ is uniform over $\calX$.  From Proposition \ref{prop:G_13_quantum_rate}, we have that $  R_{\text{quantum}}(G_{13}) = \log_2\xi(G_{13}) = \log_2 3=1.585$. 

 By Theorem \ref{thm:VLC_optimal_rate} and  Lemma \ref{lemma:graph_entropy}     \ref{lemma:chromatic_entropy_OR_prod_limit}, $R_{VLC} = \lim_m\frac{1}{m} H_{\chi,D}(G_{13}^{\lor m},\bfX) = H_{\kappa}(G_{13},X)$. By Lemma \ref{lemma:ont_shot_entropy_OR_prod} and  Lemma \ref{lemma:graph_entropy}  \ref{lemma:korner_vs_chromatic_entropy}, $R_{VLC}^{(1)}(G_{13},X)\ge H_{\chi}(G_{13},X)\ge H_{\kappa}(G_{13},X)$. 


We used the method provided  in \cite{HarangiNB24} and obtained $H_{\kappa}(G_{13},X) \approx 1.646$ (code available at \cite{meng2024github}). 
This implies there is a   quantum advantage in both multiple-instance and single-instance setting. 

\subsubsection{Example associated with strong product.} \label{sec:VLC_Hn}
Consider a function and a joint distribution such that its single-instance $f$-confusion graph $H_n$ ({\it cf.} \ref{subsubsec:Hn}) and $p_X(\cdot)$ is uniform over $\calX$.  By Theorem \ref{thm:construct_function_based_on_G} and Remark \ref{remark:free_to_choose_p_x}, there exists such a function and a joint distribution.

\begin{proposition}[Lemma 5.5 of \cite{BrietBL_15}]\label{prop:upper_bound_alpha_H_n}
    \begin{align*}
       \alpha((H_n^{\lor m})^{\frac{1}{m}} \le \alpha((H_n^{ ( m)})^{\frac{1}{m}}\le \alpha((H_n^{\boxtimes m})^{\frac{1}{m}} < 2^{0.846n}\text{ for }m\ge 1.
    \end{align*}
\end{proposition}
Lemma 1 part \ref{lemma:chromatic_entropy_ind_set_lower_bound} and the above Proposition gives
\begin{align*}
\frac{1}{m}H_{\chi}(G^{(m)},\bfX)\ge \frac{1}{m} \log \frac{|V(G^{(m)})|}{\alpha(G^{(m)})} \ge \log \frac{2^{n}}{2^{0.864n}} = {0.154 n}.
\end{align*}
We already showed in Section \ref{subsubsec:Hn} that  
\begin{align*} n+1  \ge \xi(\overline{H_n^{\lor m}})^{\frac{1}{m}}\ge\xi(\overline{H_n^{(m)}})^{\frac{1}{m}} \ge \xi(\overline{H_n^{\boxtimes m}})^{\frac{1}{m}}\text{ for }m\ge 1.
\end{align*}
Then we have
\begin{align*}
    R_{VLC} \ge 0.154n > \log (n +1)\ge  \frac{\log_2\xi(\overline{H_n^{\boxtimes }})}{m}\ge R_{\text{quantum}}.
\end{align*}
This gives the following conclusion. If the marginal distribution $p_X(\cdot)$ is uniform, then both $H^{\lor m},H^{\boxtimes m}$ are $f$-confusion graphs where the quantum rate is exponentially smaller than the VLC rate.

\subsection{Quantum advantage   in single-instance setting, but VLC advantage in multiple-instance setting.\label{sec:VLC_line_graph}}

Consider a function and a joint distribution such that its $m$-instance $f$-confusion graph $\frL(G_{13})^{\lor m}$ ({\it cf.} \ref{sec:line_graph_eg}) and $p_X(\cdot)$ is uniform over $\calX$. From Section \ref{sec:line_graph_eg}, we have that $R_{\text{quantum}}(\frL(G_{13}))=\log_2\xi(\frL(G_{13}))=\log_2 3.$ 

Since the $f$-confusion graph is the OR product, all $x\ne x'$ and $(x,x')\notin E(G)$ are due to condition $C2$ by Theorem \ref{thm:theorem4} \ref{thm:theorem4_part_b}. Then,  $R_{VLC}^{(1)}(\frL(G_{13}),X)\ge H_{\chi}(\frL(G_{13}),X)$ by Lemma \ref{lemma:ont_shot_entropy_OR_prod}. We obtained a lower bound on $H_{\chi}(\frL(G_{13}),X)$ by numerical calculation (code available at \cite{meng2024github}), which is described as follows.  

We found $\alpha(\frL(G_{13}))=18$  by numerical calculation. Since every independent set of a graph is a clique of its complement graph, $\alpha(\frL(G_{13}))= \omega(\overline{\frL(G_{13})})$. We used Bron-Kerbosch algorithm \cite{BronKerbosch73} on $\overline{\frL(G_{13})}$ to enumerate all maximal clique sets and calculated $\omega(\overline{\frL(G_{13})})$. 

 Let $c$ be the coloring of  $\frL(G_{13})$ that achieves $H_{\chi}(G,X)$, and $C_1,C_2,\dots, C_k$ each color class (vertex set of same color). Recall that Claim \ref{claim:chi(L(G13)=4} states $\chi(\frL(G_{13})) = 4$. Then, the number of colors  $k\ge \chi(\frL(G_{13}))=4$.  Define  $\beta_i:=|C_i|,i\in k$. Since $p_X(\cdot)$ is uniform, we have 
\begin{align*}
  H_{\chi}(\frL(G_{13}),X)= \sum_{i=1}^k - \frac{|C_i|}{|V(\frL(G_{13}))|} \log_2 \Big  (\frac{|C_i|}{|V(\frL(G_{13}))|}\Big  )= \sum_{i=1}^k - \frac{\beta_i}{|V(\frL(G_{13}))|} \log_2 \Big (\frac{\beta_i}{|V(\frL(G_{13}))|}\Big ).
\end{align*} 
Since each color class is an independent set, $\beta_i\le \alpha(\frL(G_{13})) = 18$ for each $i$. The sum of size of all color classes must equal the number of vertices.  Therefore, we have
\begin{equation}\label{eq:line_graph_color}
        \sum_{i=1}^k\beta_i=|V(\frL(G_{13}))| = 48\text{ and }\beta_i\le 18, \forall i\in[k].
\end{equation}
Without loss of generality, we assume that $\beta_1\ge \beta_2 \ge \dots \ge \beta_k$. It turns out $(\beta_1, \beta_2 , \dots , \beta_k)$ can be viewed as an integer partition of 48 such that the number of parts $k$ is at least $\chi(\frL(G_{13})=4$  and each part  $\beta_i$  is at most $\alpha(\frL(G_{13})=18$. In particular, Sagemath has a built-in method to enumerate all such $(\beta_1, \beta_2 , \dots , \beta_k)$'s. Upon calculating the quantity 
\begin{align*}
    \sum_{i=1}^k - \frac{\beta_i}{|V(\frL(G_{13}))|} \log_2 \Big (\frac{\beta_i}{|V(\frL(G_{13}))|}\Big ),
\end{align*}
the minimum value obtained was $1.66$ (code available at \cite{meng2024github}). Therefore, we conclude that 
\begin{align*}
     H_{\chi}(\frL(G_{13}),X)\ge 1.66>\log_2 3.
\end{align*}
 By Theorem \ref{thm:VLC_optimal_rate} and  Lemma \ref{lemma:graph_entropy} \ref{lemma:chromatic_entropy_OR_prod_limit}, $R_{VLC} = \lim_m\frac{1}{m} H_{\chi,D}(\frL(G_{13})^{\lor m},\bfX) = H_{\kappa}(\frL(G_{13}),X)$. Using the method in \cite{HarangiNB24} we obtained $H_{\kappa}(\frL(G_{13}),X) = 1.45 < \log_2 3$ (code available at \cite{meng2024github}).
   
\section{Conclusions and Future Work}
\label{sec:conclusion}

In this work, we considered the problem of function computation with side information under the zero-error constraint. Alice and Bob have sources $X$ and $Y$ respectively with joint p.m.f. $p_{XY}(X,Y)$ and Bob wants to compute $f(X,Y)$. In all three settings (FLC, VLC and quantum), the $f$-confusion graph $G$ and its $m$-instance version $G^{(m)}$ plays a key role in the analysis of the optimal rate. We provided a characterization of necessary and sufficient conditions on the function $f(\cdot, \cdot)$ and joint p.m.f. $p_{XY}(\cdot,\cdot)$ such that $G^{(m)}$ equals  $G^{\boxtimes m}$ (strong product) or $G^{\lor m}$ (OR product). In general, $G^{(m)}$ lies in between these two extremes. We compared quantum rate versus FLC rate and  quantum rate versus VLC rate, and presented several examples of $f$-confusion graphs that demonstrate the varied behavior of the quantum advantage, when considering the single-instance and asymptotic rates. 

%
 
\newpage

\bibliographystyle{IEEEtran}
\bibliography{ref} 

\newpage
\thispagestyle{empty} 

\mbox{}

\appendices

 \section{Proofs of proposition \ref{prop:chi}} \label{app:appendix_A}

\subsection{Proof of Proposition \ref{prop:chi} \ref{prop:de_morgan}}
\begin{proof}
We only need to show that $E\big (    \overline{G\boxtimes H} \big ) = E\big (    \overline{G}\lor \overline{H}  \big )$. Toward this end, note that
\begin{align*}
&\big( (v_1,u_1),(v_2,u_2) \big)\in E\big (    \overline{G\boxtimes H} \big )\\
\Leftrightarrow&\big( (v_1,u_1),(v_2,u_2)\big) \notin E\big (    G\boxtimes H  \big )\text{~and~}(v_1,u_1)\ne(v_2,u_2)\\
\Leftrightarrow&\big( v_1\ne v_2\text{ and }(v_1,v_2)  \notin E(G) \big )\text{ or } \\
&\big( u_1\ne u_2\text{ and }(u_1,u_2) \notin E(H) \big ) \\
\Leftrightarrow&\big( (v_1,u_1),(v_2,u_2) \big)\in E\big (    \overline{G}\lor \overline{H}  \big ).
\end{align*}
We conclude that $E\big (    \overline{G\boxtimes H} \big ) = E\big (    \overline{G}\lor \overline{H}  \big )$.
\end{proof}

\subsection{Proof of Proposition \ref{prop:chi} \ref{prop:alpha_AND_OR_product}}
\begin{proof}

    Let $S_1$ and $S_2$ be independent sets in $G$ and $H$ respectively. Then, $S_1\times S_2$ is independent in $G\lor H$. Thus, $\alpha(G\lor H) \ge \alpha(G)\alpha(H)$. Now we show the other inequality. Let $S$ be a maximum independent set in $G\lor H$. Define 
    \begin{align*}
        &S_g := \{g\in V(G): \exists h\in V(H), (g,h) \in S\},
        \\  &S_h := \{h\in V(h): \exists g\in V(G), (g,h) \in S\}
    \end{align*}
    The role of $S_g,S_h$ is interchangeable, so it suffices to show $S_g$ is independent in $G$. Let distinct $g_1,g_2\in S_g$. Then there exist $h_1,h_2$ such that   $(g_1,h_1),(g_2,h_2)\in S$. Since $S$ is independent in $G\lor H$, $(g_1,h_1),(g_2,h_2)\notin E(G\lor H)$ and thus $(g_1,g_2)\notin E(G)$. It follows that $S_g$ is  independent.  Therefore, $\alpha(G\lor H) \le \alpha(G)\alpha(H)$ holds.
    
For $\alpha(G)\alpha(H)\le \alpha(G\boxtimes H),$ we note that if   $S_1$ and $S_2$ are independent in $G$ and $H$ respectively, then $S_1\times S_2$ is independent in $G\boxtimes H$. Then, $\alpha(G\boxtimes H) \ge \alpha(G) \alpha(H)$ holds.
\end{proof}

\subsection{Proof of Proposition \ref{prop:chi} \ref{prop:chi_product}}
\begin{proof}
$\chi(G\boxtimes H)\le\chi(G\lor H)$ holds because $G\boxtimes H\subseteq G\lor H$.

If $c:V(G)\mapsto [\chi(G)]$ and $d:V(H)\mapsto [\chi(H)]$ are proper colorings  of $G$ and $H$ respectively, then 
\begin{align*}c\times d:V(G)\times V(H)&\mapsto [\chi(G)]\times [\chi(H)],\\
(u,v) &\mapsto (c(u),d(v))
\end{align*}  is a proper coloring of $G\lor H$. Therefore, $\chi(G \lor H) \le \chi(G)\chi(H).$
\end{proof}

\section{Proof of Proposition \ref{prop:AND_power_m-instance_graph_OR_power}}

 \label{app:pf_prop5}

\begin{proof}Since the vertex sets of $G^{\boxtimes m}, G^{(m)}, G^{\lor m}$ are the same, namely $\calX^m$, it suffices to show  $E(G^{\boxtimes m})\subseteq E(G^{(m)})\subseteq E(G^{\lor m})$. 

Let $(\bfx,\bfx')\in E(G^{\boxtimes m})$ and $\bfx\ne \bfx'$. There exists $  i$ such that  $x_{  i}\ne x_{  i}'$ because  
  $\bfx\ne \bfx'$. For such $i$ indices, $(x_i,x_i')\in E(G)$, so there exists $y_i$ such that  $p_{XY}(x_i,y_i)p_{XY}(x_i',y_i)>0$ and $f(x_i,y_i)\ne f(x_i',y_i)$.  For the remaining $i$ indices such that  $x_i= x_i'$, there exists $y_i$ such that  $p_{XY}(x_i,y_i)>0$ by Assumption \ref{assume:irreducible}. Let $\bfy = \{y_i\}_{y=1}^m$. Then we have
\begin{align*}
p_{\bfX\bfY}(\bfx ,\bfy)p_{\bfX\bfY}(\bfx' ,\bfy) > 0\text{ and } f^{(m)}(\bfx ,\bfy) \neq f^{(m)}(\bfx' ,\bfy).
\end{align*}
Therefore, $(\bfx,\bfx')\in E(G^{(m)})$. It implies $E(G^{\boxtimes m})\subseteq E(G^{(m)})$ as  $(\bfx,\bfx')\in E(G^{(m)})$ is arbitrary.

Let $(\bfx,\bfx')\in E(G^{(m)})$ be arbitrary. Then, there exists $\bfy$ such that $ p_{\bfX\bfY}(\bfx ,\bfy)p_{\bfX\bfY}(\bfx' ,\bfy) > 0$ and   $f^{(m)}(\bfx ,\bfy) \neq f^{(m)}(\bfx' ,\bfy)$. It implies that  there exists $j \in [m]$ such that $p_{XY}(x_{j},y_j) p_{XY}(x_{j}', y_j) > 0$ and $f(x_{j},y_j) \neq f(x_{j}',y_j)$, i.e. $(x_j,x_j')\in E(G).$ This implies $(\bfx,\bfx')\in E(G^{\lor m})$. Thus, $E(G^{(m)}) \subseteq  E(G^{\lor m})$  as $(\bfx,\bfx')\in E(G^{\lor m})$ are   arbitrary.
\end{proof}

\section{Proofs of proposition \ref{prop:xi}}\label{app:appendix_C}

\subsection{Proof of Proposition \ref{prop:xi} 
 \ref{prop:subgraph_xi}}
\begin{proof}
Recall that we use $G\subseteq H$ to denote that $G$ is a spanning subgraph of $H$, i.e. $V(G) = V(H)$ and $E(G)\subseteq E(H)$. Therefore, every orthogonal representation of $G$ is an orthogonal representation of $H$.
\end{proof}

\subsection{Proof of Proposition \ref{prop:xi}  \ref{prop:xi_product}}
\begin{proof}
$\xi(G\lor H)\le \xi(G\boxtimes H)$ holds because of 
Proposition \ref{prop:xi} \ref{prop:subgraph_xi} and the fact that $G\boxtimes H\subseteq G\lor H$. 

Now we show that $\xi(G\boxtimes H)\le \xi(G)\xi(H).$ Suppose $\phi_G:V(G)\mapsto \bbC^{m_1},\phi_H:V(H)\mapsto \bbC^{m_2}$ are   orthogonal representations of $G,H$ respectively. We claim $$
\phi_G\times\phi_H:V(G)\times V(H) \mapsto \bbR^{m_1m_2}, (g,h)\mapsto \phi_G(g) \otimes \phi_H(h)
$$
where $\otimes$ denotes tensor product, is an orthogonal representation of $G\boxtimes H$. 

Indeed, let $(u_1,v_1),(u_2,v_2)$ be distinct and non-adjacent in $G\boxtimes H$. W.l.o.g., we assume $(u_1,u_2)\notin E(G)$ and $u_1\ne u_2$. Thus $\phi_G(u_1)\perp\phi_G(u_2)$. It follows that $\phi_G\times\phi_H(u_1,v_1)=\phi_G(u_1)\otimes\phi_H(v_1) \perp \phi_G(u_2)\otimes\phi_H(v_2)= \phi_G\times\phi_H(u_2,v_2)$.
\end{proof}

\subsection{Proof of Proposition \ref{prop:xi} \ref{prop:xi_chi_compare}}
\begin{proof}
Let $c:V(G)\mapsto[k]$ be a proper coloring of $\overline{G}$ and $e_i$ be the $i$-th elementary vector in $\bbC^{k}$ for $i\in[k]$. Then $c$ induces the following $k$-dimensional orthogonal representation of $G$.
$$
\phi:V(G)\mapsto \bbC^{k}, \phi(v)\mapsto e_i\text{ iff }c(v) = i.
$$
\end{proof}
 
\section{Proof of $\vartheta(G)\vartheta(H)=\vartheta(G\lor H)$}\label{app:vartheta_multiplicity}

From Theorem 7 of \cite{lovasz1979shannon}, we have that $\vartheta(G\boxtimes H) = \vartheta(G)\vartheta(H)$. Since $G\boxtimes H\subseteq G\lor H$, we have $\vartheta(G\lor H)\le \vartheta(G\boxtimes H) =\vartheta(G)\vartheta(H)$ by Lemma \ref{lemma:lovasz} \ref{prop:lovasz_number_subgraph}. Therefore, it suffices to show
\begin{align*}
    \vartheta(G)\vartheta(H)\le \vartheta(G\lor H).
\end{align*} 
We will use the following equivalent definition of $\vartheta(G)$.
\begin{theorem}[from \cite{lovasz1979shannon}]
    Let $G$ be a graph on vertices $[n]$. Let $\phi$ range over all orthogonal representation over $\overline{G}$ be s.t. $\phi(i)$ is a real vector for all $i\in V(G)$ and $d$ range over all real unit-norm vectors. Then 
    \begin{equation}\label{eqn:lovasz_number_thm5}
        \vartheta(G) = \max_{\phi,d}\sum_{i=1}^n(d^T\phi(i))^2
    \end{equation}
\end{theorem}
The proof  is mostly the same as the one in Theorem 7 of \cite{lovasz1979shannon}. We include it for completeness.
\begin{proof}[Proof of $ \vartheta(G)\vartheta(H)\le \vartheta(G\lor H)$]
Let $\phi$ and $c$ be an orthogonal representation of $G$ and a vector in  (\ref{eqn:lovasz_number_thm5}) that achieves $\vartheta(G)$, and $\psi$ and $d$ be an orthogonal representation of $H$ and a vector in   (\ref{eqn:lovasz_number_thm5})   that achieves $\vartheta(H)$. We have that $\phi \times\psi(i,j):= \phi(i) \otimes \psi(j)$
is an orthogonal representation of $\overline{G}\boxtimes \overline{H}$ by the same reason used in Proposition \ref{prop:xi} \ref{prop:xi_product}. Since $G\lor H =\overline{\overline{G}\boxtimes\overline{H}}$, we have
\begin{align*}
    \vartheta(G\lor H)  \ge &\sum_{i=1}^{|V(G)|}\sum_{j=1}^{|V(H)|} \Big ( (d\otimes c)^T (u_i \otimes v_j) \Big )^2\\
        = &\sum_{i=1}^{|V(G)|} \sum_{j=1}^{|V(H)|}(d^Tu_i)^2(c^Tv_j)^2\\
        =&\sum_{i=1}^{|V(G)|}  (d^Tu_i)^2\sum_{j=1}^{|V(H)|}(c^Tv_j)^2 = \vartheta(G)\vartheta(H)
\end{align*}
\end{proof}

\section{Proof of Lemma \ref{lemma_single_letter}}\label{app:proof_lemma_single_letter}
\begin{proof}
We have $\overline{G}^{\boxtimes m}= \overline{G^{\lor m}}\subseteq \overline{G^{(m)}}\subseteq \overline{G^{\boxtimes m}} =\overline{G}^{\lor m}$ due to Proposition \ref{prop:chi} \ref{prop:de_morgan} and Proposition \ref{prop:AND_power_m-instance_graph_OR_power}. Lemma \ref{lemma:lovasz} \ref{prop:lovasz_number_subgraph} and   \ref{lemma:lovasz_number_AND_product} implies that 
\begin{align*}
\vartheta(\overline{G})^m = \vartheta(\overline{G}^{\lor m}) \le  \vartheta(\overline {G^{(m)}})\le \vartheta(\overline{G}^{\boxtimes  m}) = \vartheta(\overline{G})^m.
\end{align*}
Lemma \ref{lemma:lovasz} \ref{lemma:lovasz_number_xi_compare} gives that $\vartheta(\overline {G^{(m)}})\le\xi(\overline {G^{(m)}})$.  Proposition \ref{prop:xi}     \ref{prop:xi_product} and \ref{prop:subgraph_xi} gives $\xi(\overline {G^{(m)}}) \le\xi(\overline{G}^{\boxtimes m})  \le  \xi(\overline {G})^{m}$. Combining everything, we have
\begin{equation}
    \vartheta(\overline{G})^m\le \xi(\overline {G^{(m)}}) \le  \xi(\overline {G})^{m}.\label{eq:sandwich_version2}    
\end{equation}
If $\vartheta(\overline{G})=\xi(\overline {G})$, then we have $\xi(\overline {G^{(m)}})^{\frac{1}{m}}=\xi(\overline {G})$ for all $m\ge 1$ by  (\ref{eq:sandwich_version2}). Therefore, \ref{lemma_single_letter_part_a} is proved. 

Lemma \ref{lemma:lovasz} \ref{lemma:lovasz_number_indep_compare} gives $\omega(G) = \alpha(\overline{G}) \le \vartheta(\overline{G})$. Proposition \ref{prop:xi} \ref{prop:xi_chi_compare} gives $\xi(\overline {G})\le \chi(G)$. Plugging them into   (\ref{eq:sandwich_version2}), we have
\begin{align*}
 \omega(G)^m   \le \vartheta(\overline{G})^m\le \xi(\overline {G^{(m)}}) \le  \xi(\overline {G})^{m}   \le \chi(G)^m.
\end{align*}
Similarly, for chromatic number, we have
\begin{align*}
 \omega(G)^m   \le \vartheta(\overline{G})^m\le \xi(\overline {G^{(m)}}) {\le}  \chi( {G^{(m)}})    \le \chi(G)^m.
\end{align*}
If $G$ is perfect, then $ \omega(G) = \chi(G)$. By the above inequalities, we have $\xi(\overline {G^{(m)}})^{\frac{1}{m}}=\xi(\overline {G})=\chi(G)=\chi( {G^{(m)}})^{\frac{1}{m}}$ for all $m\ge 1$.  
\end{proof}





\section{Proof of Proposition \ref{prop:C5_AND_rates}}
\label{app:pf_prop10} At various points in the discussion below, we use the well-known fact that the pentagon graph, $C_5$, is self-complementary, i.e., $\overline{C_5}$ is isomorphic to $C_5$.
\begin{proof} 
For part (i), note that by Proposition \ref{prop:xi} \ref{prop:xi_chi_compare}, we have $\xi(\overline{C_5}) \le \chi(C_5) =3 $. Moreover, we know that $\overline{C_5}$ is isomorphic to $C_5$ and furthermore that $\vartheta(C_5) \le \xi(C_5)$. It is well known that $\vartheta(C_5) = \sqrt{5}$ \cite{lovasz1979shannon}. Since $\xi(C_5)$ is an integer, we must have $\xi(C_5) =3$.

\noindent For part (ii), note that $G^{(m)} = C_5^{\boxtimes m}$. It follows that 
\begin{equation}\label{eqn:lowbound_xi(C5_OR_m)}
    \xi(\overline{C_5^{\boxtimes m}}) \ge \vartheta(\overline{C_5^{\boxtimes m}})=\vartheta({C_5^{\lor m}})= \vartheta({C_5})^m = 5^{m/2}.    
\end{equation}
The first inequality follows from Lemma \ref{lemma:lovasz} \ref{lemma:lovasz_number_xi_compare}. The first equality holds because $\overline{C_5^{\boxtimes m}} = \overline{C_5}^{\lor m}= {C_5^{\lor m}}$. The second equality holds by  \eqref{eq:sandwich}. The last equality holds as $\vartheta(C_5) =\sqrt{5}$ \cite{lovasz1979shannon}. 

Consider $G^{(2)} = C_5^{\boxtimes 2}$. We know $\chi( C_5^{\boxtimes 2})=5$ by \cite{Witsenhausen_76}. Therefore, for $m\ge 1$, we have 
\begin{equation}\label{eqn:upbound_chi(C5_OR_m)}
\begin{split}
        &\chi({C_5^{\boxtimes m}}) \le
        \begin{cases}
          \chi({C_5^{\boxtimes 2}})^{m/2}\text{ if }m\text{ even}\\
        \chi(C_5)\chi({C_5^{\boxtimes 2})^{(m-1)/2}}\text{ if }m\text{ odd}
        \end{cases}
        \\ 
        &\le 3\cdot 5^{(m-1)/2}
\end{split}
\end{equation}
where the first inequality holds by Proposition 
\ref{prop:chi} \ref{prop:chi_product} and the second inequality holds by $\chi(C_5^{\boxtimes 2}) =5$ and $\chi(C_5)=3$. By Proposition \ref{prop:xi} \ref{prop:xi_chi_compare}, we have $\xi(\overline{C_5^{\boxtimes m}})\le \chi({C_5^{\boxtimes m}})$. This combined with bounds in (\ref{eqn:lowbound_xi(C5_OR_m)}) and (\ref{eqn:upbound_chi(C5_OR_m)})  gives  
\begin{align*} 
   {5^{m/2}} \le  \xi(\overline{C_5^{\boxtimes m}}) \le  \chi({C_5^{\boxtimes m}})\le 3\cdot 5^{(m-1)/2}.
\end{align*} 
This implies that
\begin{align*} 
 & \frac{1}{2}\log_2 {5} \le    \frac{1}{m}
\log_2 \xi(\overline{C_5^{\boxtimes m}})\le \frac{1}{m}
\log_2 \chi({C_5^{\boxtimes m}}) \\ & 
 \le  \frac{1}{m} \log_2 3 + \frac{m-1}{2m}\log_2 5.
\end{align*}
Taking infimum over $m$, we conclude that $R_{\text{quantum}}(\tilde f) = R_{\text{FLC}}(\tilde f) =\frac{1}{2}\log_25$.

Now, we prove part (iii). From Theorem 6.15.3 of \cite{Roberson_13}, we have that $\xi_f(\overline{C_5}) = \frac{5}{2}$. The fact that $ \chi_f(C_5) = \frac{5}{2}$ is well known (see pg. 12 of \cite{ScheinermanU97FGT}), so $\xi_f(\overline{C_5})=\frac{5}{2} = \chi_f(C_5).$

From Theorem 27 of \cite{CubittLRSSW_14}, we have that $\xi_f(\overline{G \lor H} ) = \xi_f(\overline{G } ) \xi_f(\overline{ H} )$ for all graphs $G,H$. Thus, 
\begin{align*}
     \xi_f(\overline{C_5 } )= \xi_f(   \overline{ C_5^{\lor m} } )^{\frac{1}{m}}\overset{(a)}{\le} \xi (   \overline{ C_5^{\lor m} } )^{\frac{1}{m}}\overset{(b)}{\le} \chi (   { C_5^{\lor m} } )^{\frac{1}{m}} 
\end{align*}
where $(a)$ follows from Lemma \ref{lemma:lovasz}    \ref{lemma:lovasz_number_xi_compare}, and $(b)$ follows from Proposition \ref{prop:xi}  \ref{prop:xi_chi_compare}.
Taking logarithm and infimum over $m$ and using Proposition \ref{prop:chi_frac}  \ref{prop:OR_Witsenhausen_rate}, we have
\begin{align*}
    &\log_2\xi_f(\overline{C_5 } )=\inf_{m}\frac{1}{m} \log_2\xi_f(   \overline{ C_5^{\lor m} } )  \le\\
    &\inf_{m} \frac{1}{m}\log_2 \chi (   { C_5^{\lor m} } )  = \log_2\chi_f(C_5).
\end{align*}  
Taking infimum over $m$ and we conclude that $R_{\text{quantum}}(\tilde g) = R_{\text{FLC}}(\tilde g) = \log_2\frac{5}{2}$.
\end{proof}

\section{Linear Programming Formulation of $\chi_f(G)$}\label{app:chi_f_LP}
The fractional chromatic number can also be expressed in terms of a linear program with variables corresponding to the independent sets of $G$ (pg. 30 of \cite{ScheinermanU97FGT}). 
\begin{definition}\label{defi:Fraction_chromatic_number}
Let $G=(V,E)$ be a simple graph and $\calI$ be the collection of its independent set. Its fractional chromatic number $\chi_f(G)$ is defined as
\begin{align*}
    \min& \sum x_{I}\\
    \text{ such that  }& \sum_{I\in \calI, v\in I} x_I \ge 1,  \forall v\in V,\\
    &  x_I \ge 0, \forall I\in \calI.
\end{align*}
\end{definition} 
\section{Proof of Lemma \ref{lemma:11_pf_encoding_bounds}  \ref{lemma:11_pf_encoding_bounds_part_2}, Lemma \ref{lemma:prefix_free_OR_prod}, Lemma \ref{lemma:ont_shot_entropy_OR_prod} and Lemma \ref{lemma:graph_entropy}   \ref{lemma:chromatic_entropy_chi_upper_bound} and \ref{lemma:korner_entropy_chi_f_upper_bound}} \label{app:proof_lemma_graph_entropy} 
\begin{proof}[Proof of Lemma \ref{lemma:11_pf_encoding_bounds}  \ref{lemma:11_pf_encoding_bounds_part_2}]
Without loss of generality, we  let $\text{Pr}(X=x_i) = p_i,i=1,\dots,n$, where $n=|\calX|$, and assume that $p_1\ge p_2\ge p_3\ge \dots \ge p_n >0$. The set
of available codewords is $\{0,1,\dots,D-1\}^*$. Let $l_i$ denote the length of the codeword for $x_i$. Then, an optimal regular non-singular code is such that $l_1\le l_2\le l_3\le \dots\le l_n$. Then,
\begin{align*}
    &l_i = 1, \text{ if }i\in \{1,2,\dots, D\}\\
    &l_i = 2, \text{ if }i\in \{D+1,D+2,\dots, D+D^2\}\\
    &l_i = 3, \text{ if }i\in \{D+D^2+1,D+D^2+2,\dots, D+D^2+D^3\}\\
    &\dots
\end{align*}
Suppose  $i\in \{D+\dots+D^j+1, \dots ,D+\dots+D^j+D^{j+1} \}$. We can write $i= D+\dots+D^j+k$ for some $k\in\{1,2,\dots, D^{j+1}\}$. We have $l_i = j+1$. Note that 
\begin{align*}
    &\ceil{\log_D[(D-1)\cdot \frac{i}{D}+1]}\\ 
    =&\ceil{ \log_D \left[\frac{D-1}{D} (D+\dots+D^j+k )+1\right]}\\ 
    =&\ceil{ \log_D \left[\frac{D-1}{D} \left(D \frac{D^j-1}{D-1}+k \right)+1 \right]} \\
    =&\ceil{ \log_D \left[D^j-1+\frac{D-1}{D} k +1 \right]} \\
    =&\ceil{ \log_D \left[D^j +\frac{D-1}{D} k \right]} \\
    =&j+1
\end{align*}
where the last equality holds because $k\in\{1,2,\dots, D^{j+1}\}$. Therefore, we can write $l_i = \ceil{\log_D[(D-1)\cdot \frac{i}{D}+1]}$, and the optimal regular non-singular code has expected length $\sum_{i=1}^n p_i \ceil{\log_D[(D-1)\cdot \frac{i}{D}+1]}$. Then, we have
    \begin{align*}
        &H_D(X) - \sum_{i=1}^n p_i \ceil{\log_D[(D-1)\cdot \frac{i}{D}+1]}\\
        \le&H_D(X) - \sum_{i=1}^n p_i {\log_D[(D-1)\cdot \frac{i}{D}+1]}\\
        =& \sum_{i=1}^n p_i\log_D\left(\frac{1}{p_i}\right) - \sum_{i=1}^n p_i {\log_D[(D-1)\cdot \frac{i}{D}+1]}\\
        =& \sum_{i=1}^n p_i\log_D\left(\frac{1}{p_i}\cdot \frac{D}{(D-1)\cdot i + D}\right)\\
        \
        \overset{(a)}{\le }&   \log_D \left(\sum_{i=1}^n \frac{D}{(D-1)\cdot i + D}\right)\\
        =&  1+  \log_D \left(\sum_{i=1}^n \frac{1}{(D-1)\cdot i + D} \right)\\ 
        =&   \log_D\left(\frac{D}{D-1}\right) +\log_D \left(\sum_{i=1}^n \frac{1}{ i + \frac{D}{D-1}}\right)\\ 
        \le &   \log_D\left(\frac{D}{D-1}\right) +\log_D \left(\sum_{i=1}^n \frac{1}{ i } \right)\\ 
        \overset{(b)}{\le} &   \log_D\left(\frac{D}{D-1}\right) +\log_D(1+\ln n)
    \end{align*}
    where $(a)$ holds by concavity, and $(b)$ holds by $\sum_{i=1}^n \frac{1}{ i }=1+\sum_{i=1}^{n-1}\frac{1}{n+1}\le1+\int_{1}^n\frac{1}{x}dx=1+\ln n $.
\end{proof}
\begin{proof}[Proof of Lemma \ref{lemma:prefix_free_OR_prod}]
Pick arbitrary $(\zeta_1,\zeta_2)$   such that  $\zeta_1\ne \zeta_2$ and $\zeta_1,\zeta_2\in \text{Im}(\phi_e)$. Since $\zeta_1,\zeta_2\in \text{Im}(\phi_e)$, there exist $\bfx_1,\bfx_2$ such that $\phi_e(\bfx_1)=\zeta_1,\phi_e(\bfx_2)=\zeta_2$. 
\begin{itemize} 

    \item Case 1: $(\bfx_1,\bfx_2)\in E(G)$. This means that  $\zeta_1=\phi_e(\bfx_1),\zeta_2=\phi_e(\bfx_2)$ cannot be a prefix of each other.

    \item Case 2:  $\bfx_1\ne\bfx_2$ and $(\bfx_1,\bfx_2)\notin E(G)$. Note that all $x\ne x'$ and $(x,x')\notin E(G)$ are due to condition C2. Thus, there exists $\bfy$ such that $p_{\bfX\bfY}(\bfx_1,\bfy)p_{XY}(\bfx_2,\bfy)>0$, (see proof of Theorem \ref{thm:theorem4} (b) for the construction of $\bfy$). Since $\phi_e$ is a prefix-free encoding function,  $\zeta_1=\phi_e(\bfx_1),\zeta_2=\phi_e(\bfx_2)$ cannot be a prefix of each other. 
\end{itemize}
\end{proof}
\begin{proof}[Proof of Lemma \ref{lemma:ont_shot_entropy_OR_prod}]
Let $\phi_e$ be the optimal encoding function for the single-instance setting. Then Lemma \ref{lemma:prefix_free_OR_prod} shows that if $\zeta_1\ne \zeta_2$ and $\zeta_1,\zeta_2\in \text{Im}(\phi_e)$, then $\zeta_1,\zeta_2$ are not prefix of each other. Therefore, identity map on $\phi_e(X)$ is a regular prefix-free encoding of itself.   Then, Lemma \ref{lemma:11_pf_encoding_bounds} \ref{lemma:11_pf_encoding_bounds_part_1} gives
 \begin{align*}  
   R_{VLC}^{(1)}=&\sum_xp_X(x) |\phi_e(X)|\\
   =&\sum_\zeta \sum_{x\in \phi_e^{-1}(\zeta)}p_X(x) |\phi_e(X)| \\
   = &   \sum_\zeta\Pr({\phi_e(X)}=\zeta) |\zeta| \\
   \ge&    l_{r.p.f.}(\phi_e(X),D)\\
   \ge& H_D(\phi_e(X))\\
   \ge&  H_{\chi,D} (G,X)
 \end{align*}
where $\phi_e^{-1}(\zeta)=\{x: \phi_e(x)=\zeta\}$, and the first inequality holds because identity map on $\phi_e(X)$ is a regular prefix-free encoding of itself, second inequality holds by Lemma \ref{lemma:11_pf_encoding_bounds} \ref{lemma:11_pf_encoding_bounds_part_1}, and third inequality holds by noting $\phi_e$ is a valid coloring. 
\end{proof}

\begin{proof}[Proof of  Lemma \ref{lemma:graph_entropy}    \ref{lemma:chromatic_entropy_chi_upper_bound}]
    Let $C$ be a coloring of $G^{(m)}$ that uses $\chi(G)$ colors.  Then, random color $C(\bfX)$ is supported on a set of size $\chi(G)$. Then
    \begin{align*}
       H_{\chi,D}(G^{(m)},\bfX) \le H_D(C(\bfX))\le \log_D\chi(G).
    \end{align*}
\end{proof}



\begin{proof}[Proof of  Lemma \ref{lemma:graph_entropy}   \ref{lemma:korner_entropy_chi_f_upper_bound}]
Fix  $\epsilon>0$. Let  $g$ be a $a:b$-coloring of $G^{(m)}$ that is $\epsilon$-close to $\chi_f(G^{(m)})$, i.e., $a/b-\chi_f(G^{(m)})\le \epsilon$. Thus, $\forall (\bfx_1,\bfx_2)\in E(G^{(m)}), g(\bfx_1)\cap g(\bfx_2) = \emptyset$ and $|g(\bfx)| = b$ for all $\bfx \in V(G^{(m)})$. This implies that   $S_i:=\{\bfx\in V(G^{(m)})|i\in g(\bfx)\}$ is an independent set for each $i\in [a]$. Let $Z$ denote a random variable taking values on the independent sets in $G^{(m)}$ and consider the following conditional distribution.
\begin{align}
    \text{Pr}(Z = z|\bfx) &= \begin{cases}
        \frac{1}{b} & \text{~if~} z = S_i \text{~for some~} i \in [a] \text{~and~} \bfx \in S_i,\\
        0 & \text{otherwise.}
        \end{cases}
\end{align}

We note that $Z$ is supported on $\{S_1,S_2,\dots, S_a\}$. This implies that $H_D(Z)\le \log_Da$. Given $\bfx$, we have $    H_D(Z|\bfx) = \log_D b$ as the distribution $\text{Pr}(Z = z|\bfx)$ is uniform for each $\bfx$. Then, we calculate the mutual information
\begin{align*}
    I_D(\bfX;Z) =& H_D(Z)-H_D(Z|\bfX)\\
    \le & \log_D a -H_D(Z|\bfX)\\
    = & \log_D a -\sum_{\bfx}p(\bfx)H_D(Z|\bfx)\\
    = & \log_D a -\sum_{\bfx}p(\bfx)\log_D b\\
    =& \log_D a -\log_D b\\ 
    \le &\log_D [\chi_f(G^{(m)}) + \epsilon].
\end{align*}
Since $\epsilon$ is arbitrary, we can take the infimum $\epsilon\to 0$ and the proof is complete.
\end{proof} 

\section{Proof of Proposition \ref{prop:VLC_rate_c5_AND} and  Proposition \ref{prop:VLC_rate_c5_OR}}\label{app:proof_VLC_rate_c5}
\begin{proof}[Proof of Proposition \ref{prop:VLC_rate_c5_AND}]
For $\tilde f$, consider the following   encoding in Fig. \ref{fig:c5_vlc_encoding}.
It can be verified that this encoding satisfies the conditions of a variable length prefix-free code ({\it cf.} Section \ref{subsec:problem_formulation}). The expected length is 
\begin{align*}
\frac{2}{5}\cdot 1 + \frac{1}{5}\cdot 1 + \frac{1}{5}\cdot 2 + \frac{1}{5}\cdot 2= \frac{7}{5} = 1.4.
\end{align*}
Therefore, $R_{VLC}^{(1)}(\tilde f) \le 1.4$. 

Toward a contradiction, we assume there is another encoding $\psi_e$ with strictly smaller expected length.  In $\phi_e$, we observe that there are three $x$s mapped to strings of length $1$, and other $x$s mapped to strings of length 2.   Since $\psi_e$ has smaller expected length, it must map at least four $x$s into strings of length 1. Since $C_5$ can be viewed as a regular pentagon $5-1-2-3-4-5$, and the symmetry group of the regular pentagon contains five rotations, we may assume that the $x$s in $\{1,2,4,5\}$ are mapped into a string of length $1$. 

If the $x$'s in $\{2,4\}$ are mapped to different strings, then $\{\psi_e(2),\psi_e(4)\} = \{0,1\}$. Then $x=3$ cannot be assigned a string that satisfies the prefix-free condition. Therefore, the $x$s in $\{2,4\}$ are mapped to the same string, say the string $0$. However, the $x$s in $\{1,5\}$ are adjacent, so we must have $\{\psi_e(1),\psi_e(5)\} = \{0,1\}$. This implies one of the $x$s in $\{1,5\}$ must have the same string with an adjacent $x'$ in $\{2,4\}$. This gives the required contradiction. Therefore, $R_{VLC}^{(1)}(\tilde f) = 1.4.$
 
For the multiple-instance case, we know that $p_{\bfX}(\cdot)$ is uniform over $C_5^{\boxtimes m}$ because $p_X(x)$ is uniform over $V(C_5)$. Note $\lim_m\frac{1}{m}\log \alpha(C_5^{\boxtimes m})=\frac{1}{2}\log 5$ from   \cite{lovasz1979shannon} and $\frac{1}{m} \log \chi(C_5^{\boxtimes m})  = \frac{1}{2}\log 5$ from \cite{Witsenhausen_76}. Use Lemma \ref{lemma:graph_entropy} \ref{lemma:chromatic_entropy_ind_set_lower_bound} for the lower bound  and  Lemma \ref{lemma:graph_entropy} \ref{lemma:chromatic_entropy_chi_upper_bound} for the upper bound. We have
\begin{align*}
     \frac{1}{m} \log \frac{|V(C_5^{\boxtimes m})|}{\alpha(C_5^{\boxtimes m})} \le R^{(m)} \le \frac{1}{m} \log \chi(C_5^{\boxtimes m}). \\
\end{align*}  
Taking infimum over $m$, we have $R_{VLC} = \frac{1}{2}\log 5$.
\end{proof}
\begin{proof}[Proof of Proposition \ref{prop:VLC_rate_c5_OR}]
For $\tilde g$, we note that every $x\ne x'$ and $(x,x')\notin E(C_5)$ is due to C2.  By Lemma \ref{lemma:prefix_free_OR_prod}, any prefix-free encoding $\phi_e$ must satisfy \textbf{Condition A}: if $\zeta_1\ne \zeta_2$ and $\zeta_1,\zeta_2\in \text{Im}(\phi_e)$, then   $\zeta_1,\zeta_2$ are not prefix of each other. Consider the following encoding.
\begin{align*}
\phi_e(x) =\begin{cases}
    0,\text{ if }x=1\text{ or }3,\\
    10,\text{ if }x=2\text{ or }4,\\
    11,\text{ if }x=5.
\end{cases}
\end{align*} 
It can be verified that this encoding satisfies the conditions of a variable length prefix-free code ({\it cf.} Section \ref{subsec:problem_formulation}). The expected length is  
\begin{align*}
\frac{2}{5}\cdot 1 + \frac{2}{5}\cdot 2 + \frac{1}{5}\cdot 2  = \frac{8}{5} = 1.6.
\end{align*}
Therefore, $R_{VLC}^{(1)}(\tilde g) \le 1.6$. 

Toward a contradiction, we assume there is another prefix-free encoding $\psi_e$ with strictly smaller expected length.  In $\phi_e$, we observe that there are three $x$s mapped to strings of length $2$, and other $x$s are mapped to strings of length 1.   Since $\psi_e$ has smaller expected length, it must map at least three $x$s into strings of length 1. Let them be $\{x_0,x_1,x_2\}$. Since $C_5$ can be viewed as a regular pentagon $5-1-2-3-4-5$, and the symmetry group of the regular pentagon contains 5 rotations, we may assume the following two cases.

\begin{itemize}
    \item Case $\{x_0,x_1,x_2\}=\{1,2,3\}$: Since $(1,2),(2,3)\in E(C_5)$, it must be that $\psi_e(1)=\psi_e(3)\ne \psi_e(2)$. Without loss of generality, $\psi_e(1)=\psi_e(3)=0,\psi_e(2)=1.$ Since $(3,4)\in E(C_5)$, $\psi_e(4)$ cannot be $0$ and cannot use $0$ as a prefix. Then, $1$ must be a prefix of $\psi_e(4)$ by \textbf{Condition A}.
\begin{itemize}
    \item  Subcase $|\psi_e(4)|=1$: Since $|\psi_e(4)|=1$ and $1$ is a prefix of $\psi_e(4)$, $\psi_e(4)=1$.  Since $(5,4)\in E(C_5)$, $\psi_e(5)$  cannot use $1$ as a prefix by \textbf{Condition A}. Similarly, $(1,5)\in E(C_5)$ implies $\psi_e(5)$  cannot use $0$ as a prefix. Then, there is no choice for $\psi_e(5)$. This is a contradiction.
    \item  Subcase $|\psi_e(4)|>1$: In this case, we have that $\psi_e(2)\ne \psi_e(4)$ because they have different length. Since  $\psi_e(4)$ contains $1=\phi_e(2)$ as a prefix, this contradicts  \textbf{Condition A}. 

\end{itemize}
    \item  
Case $\{x_0,x_1,x_2\}=\{1,2,4\}$:

Since   $(1,2) \in E(C_5)$,  $\psi_e(1),\psi_e(2)$   must be different strings of length 1, and $\psi_e(4)$ must equal to exactly one of $\psi_e(1),\psi_e(2)$. Since a regular  pentagon has 5 axial symmetries, and the role of bit $0$ and $1$ are symmetric, we may assume that $\psi_e(1)=\psi_e(4)=0,\psi_e(2)=1$.  Since $(3,4)\in E(C_5)$, $\psi_e(3)$  cannot use $0$ as a prefix by \textbf{Condition A}. Similarly, $(2,3)\in E(C_5)$ implies $\psi_e(3)$  cannot use $1$ as a prefix. Then, there is no choice for $\psi_e(3)$. This is a contradiction.
\end{itemize}
 
Combining all cases, we conclude that $\psi_e$ cannot exist, i.e. $R_{VLC}^{(1)}(\tilde g) \ge 1.6$


Since  the confusion graph $C_5^{\lor m}$ is OR product, Lemma \ref{lemma:graph_entropy} \ref{lemma:chromatic_entropy_OR_prod_limit} shows  $R_{VLC}(\tilde g)=H_{\kappa}(C_5,X)$. Example 3 of \cite{AlonO_96} gives $H_{\kappa}(C_5,X) = \log_2 \frac{5}{2}$. 
\end{proof}

\section{Two-fold $f$-confusion graphs of $\tilde f,\tilde g,\tilde h$}\label{app:2_fold_graphs}
In the figures below, we show the two-fold $f$-confusion graphs corresponding to $\tilde{f}, \tilde{g}$ and $\tilde{h}$. We point out that for easier graph visualization, there are multiple copies of a given node that are shown in the figures. For instance, in the figure for $\tilde{f}$, there are three copies of $(1,1)$. In the actual graph, all these will correspond to only one vertex $(1,1)$.
\begin{figure}[h]
    \centering
    \begin{tikzpicture} 
[
  node distance=0.5cm,
  every node/.style={circle, draw, minimum size=4mm, inner sep=0pt, font = \normalsize},
  every edge/.style={draw, -} 
]

\node[draw=black] (11) at (1,1) {1,1};
\node[draw=black] (12) at (1,2) {1,2};
\node[draw=black] (13) at (1,3) {1,3};
\node[draw=black] (14) at (1,4) {1,4};
\node[draw=black] (15) at (1,5) {1,5};

\node[draw=black] (21) at (2,1) {2,1};
\node[draw=black] (22) at (2,2) {2,2};
\node[draw=black] (23) at (2,3) {2,3};
\node[draw=black] (24) at (2,4) {2,4};
\node[draw=black] (25) at (2,5) {2,5};

\node[draw=black] (31) at (3,1) {3,1};
\node[draw=black] (32) at (3,2) {3,2};
\node[draw=black] (33) at (3,3) {3,3};
\node[draw=black] (34) at (3,4) {3,4};
\node[draw=black] (35) at (3,5) {3,5};

\node[draw=black] (41) at (4,1) {4,1};
\node[draw=black] (42) at (4,2) {4,2};
\node[draw=black] (43) at (4,3) {4,3};
\node[draw=black] (44) at (4,4) {4,4};
\node[draw=black] (45) at (4,5) {4,5};

\node[draw=black] (51) at (5,1) {5,1};
\node[draw=black] (52) at (5,2) {5,2};
\node[draw=black] (53) at (5,3) {5,3};
\node[draw=black] (54) at (5,4) {5,4};
\node[draw=black] (55) at (5,5) {5,5};

\node[draw=black] (01) at (0,1) {5,1};
\node[draw=black] (02) at (0,2) {5,2};
\node[draw=black] (03) at (0,3) {5,3};
\node[draw=black] (04) at (0,4) {5,4};
\node[draw=black] (05) at (0,5) {5,5};

\node[draw=black] (61) at (6,1) {1,1};
\node[draw=black] (62) at (6,2) {1,2};
\node[draw=black] (63) at (6,3) {1,3};
\node[draw=black] (64) at (6,4) {1,4};
\node[draw=black] (65) at (6,5) {1,5};

\node[draw=black] (10) at (1,0) {1,5};
\node[draw=black] (20) at (2,0) {2,5};
\node[draw=black] (30) at (3,0) {3,5};
\node[draw=black] (40) at (4,0) {4,5};
\node[draw=black] (50) at (5,0) {5,5};

\node[draw=black] (16) at (1,6) {1,1};
\node[draw=black] (26) at (2,6) {2,1};
\node[draw=black] (36) at (3,6) {3,1};
\node[draw=black] (46) at (4,6) {4,1};
\node[draw=black] (56) at (5,6) {5,1};

\draw (11) edge (12);
\draw (12) edge (13);
\draw (13) edge (14);
\draw (14) edge (15);

\draw (21) edge (22);
\draw (22) edge (23);
\draw (23) edge (24);
\draw (24) edge (25);

\draw (31) edge (32);
\draw (32) edge (33);
\draw (33) edge (34);
\draw (34) edge (35);

\draw (41) edge (42);
\draw (42) edge (43);
\draw (43) edge (44);
\draw (44) edge (45);

\draw (51) edge (52);
\draw (52) edge (53);
\draw (53) edge (54);
\draw (54) edge (55);

\draw (11) edge (21);
\draw (21) edge (31);
\draw (31) edge (41);
\draw (41) edge (51);

\draw (12) edge (22);
\draw (22) edge (32);
\draw (32) edge (42);
\draw (42) edge (52);

\draw (13) edge (23);
\draw (23) edge (33);
\draw (33) edge (43);
\draw (43) edge (53);

\draw (14) edge (24);
\draw (24) edge (34);
\draw (34) edge (44);
\draw (44) edge (54);

\draw (15) edge (25);
\draw (25) edge (35);
\draw (35) edge (45);
\draw (45) edge (55);

\draw (11) edge (22);
\draw (21) edge (32);
\draw (31) edge (42);
\draw (41) edge (52);

\draw (12) edge (23);
\draw (22) edge (33);
\draw (32) edge (43);
\draw (42) edge (53);

\draw (13) edge (24);
\draw (23) edge (34);
\draw (33) edge (44);
\draw (43) edge (54);

\draw (14) edge (25);
\draw (24) edge (35);
\draw (34) edge (45);
\draw (44) edge (55);

\draw (21) edge (12);
\draw (31) edge (22);
\draw (41) edge (32);
\draw (51) edge (42);
 
\draw (22) edge (13);
\draw (32) edge (23);
\draw (42) edge (33);
\draw (52) edge (43);
 
\draw (23) edge (14);
\draw (33) edge (24);
\draw (43) edge (34);
\draw (53) edge (44);
 
\draw (24) edge (15);
\draw (34) edge (25);
\draw (44) edge (35);
\draw (54) edge (45);

\draw[dashed] (25) edge (16);
\draw[dashed] (35) edge (26);
\draw[dashed] (45) edge (36);
\draw[dashed] (55) edge (46);

\draw[dashed] (20) edge (11);
\draw[dashed] (30) edge (21);
\draw[dashed] (40) edge (31);
\draw[dashed] (50) edge (41);

\draw[dashed] (10) edge (21);
\draw[dashed] (20) edge (31);
\draw[dashed] (30) edge (41);
\draw[dashed] (40) edge (51);

\draw[dashed] (15) edge (26);
\draw[dashed] (25) edge (36);
\draw[dashed] (35) edge (46);
\draw[dashed] (45) edge (56);

\draw[dashed] (52) edge (61);
\draw[dashed] (53) edge (62);
\draw[dashed] (54) edge (63);
\draw[dashed] (55) edge (64);

\draw[dashed] (02) edge (11);
\draw[dashed] (03) edge (12);
\draw[dashed] (04) edge (13);
\draw[dashed] (05) edge (14);

\draw[dashed] (01) edge (12);
\draw[dashed] (02) edge (13);
\draw[dashed] (03) edge (14);
\draw[dashed] (04) edge (15);

\draw[dashed] (51) edge (62);
\draw[dashed] (52) edge (63);
\draw[dashed] (53) edge (64);
\draw[dashed] (54) edge (65);

\draw[dashed] (51) edge (61);
\draw[dashed] (52) edge (62);
\draw[dashed] (53) edge (63);
\draw[dashed] (54) edge (64);
\draw[dashed] (55) edge (65);

\draw[dashed] (01) edge (11);
\draw[dashed] (02) edge (12);
\draw[dashed] (03) edge (13);
\draw[dashed] (04) edge (14);
\draw[dashed] (05) edge (15);

\draw[dashed] (15) edge (16);
\draw[dashed] (25) edge (26);
\draw[dashed] (35) edge (36);
\draw[dashed] (45) edge (46);
\draw[dashed] (55) edge (56);

\draw[dashed] (10) edge (11);
\draw[dashed] (20) edge (21);
\draw[dashed] (30) edge (31);
\draw[dashed] (40) edge (41);
\draw[dashed] (50) edge (51);
\end{tikzpicture}
    \caption{Two-fold $f$-confusion graphs of $\tilde f$.}
\end{figure}
\begin{figure}[h]
    \centering
    \begin{tikzpicture} 
[
  node distance=0.5cm,
  every node/.style={circle, draw, minimum size=4mm, inner sep=0pt, font = \normalsize},
  every edge/.style={draw, -} 
]

\node[draw=black] (11) at (1,1) {1,1};
\node[draw=black] (12) at (1,2) {1,2};
\node[draw=black] (13) at (1,3) {1,3};
\node[draw=black] (14) at (1,4) {1,4};
\node[draw=black] (15) at (1,5) {1,5};

\node[draw=black] (21) at (2,1) {2,1};
\node[draw=black] (22) at (2,2) {2,2};
\node[draw=black] (23) at (2,3) {2,3};
\node[draw=black] (24) at (2,4) {2,4};
\node[draw=black] (25) at (2,5) {2,5};

\node[draw=black] (31) at (3,1) {3,1};
\node[draw=black] (32) at (3,2) {3,2};
\node[draw=black] (33) at (3,3) {3,3};
\node[draw=black] (34) at (3,4) {3,4};
\node[draw=black] (35) at (3,5) {3,5};

\node[draw=black] (41) at (4,1) {4,1};
\node[draw=black] (42) at (4,2) {4,2};
\node[draw=black] (43) at (4,3) {4,3};
\node[draw=black] (44) at (4,4) {4,4};
\node[draw=black] (45) at (4,5) {4,5};

\node[draw=black] (51) at (5,1) {5,1};
\node[draw=black] (52) at (5,2) {5,2};
\node[draw=black] (53) at (5,3) {5,3};
\node[draw=black] (54) at (5,4) {5,4};
\node[draw=black] (55) at (5,5) {5,5};

\node[draw=black] (01) at (0,1) {5,1};
\node[draw=black] (02) at (0,2) {5,2};
\node[draw=black] (03) at (0,3) {5,3};
\node[draw=black] (04) at (0,4) {5,4};
\node[draw=black] (05) at (0,5) {5,5};

\node[draw=black] (61) at (6,1) {1,1};
\node[draw=black] (62) at (6,2) {1,2};
\node[draw=black] (63) at (6,3) {1,3};
\node[draw=black] (64) at (6,4) {1,4};
\node[draw=black] (65) at (6,5) {1,5};

\node[draw=black] (10) at (1,0) {1,5};
\node[draw=black] (20) at (2,0) {2,5};
\node[draw=black] (30) at (3,0) {3,5};
\node[draw=black] (40) at (4,0) {4,5};
\node[draw=black] (50) at (5,0) {5,5};

\node[draw=black] (16) at (1,6) {1,1};
\node[draw=black] (26) at (2,6) {2,1};
\node[draw=black] (36) at (3,6) {3,1};
\node[draw=black] (46) at (4,6) {4,1};
\node[draw=black] (56) at (5,6) {5,1};

\draw (11) edge (12);
\draw (12) edge (13);
\draw (13) edge (14);
\draw (14) edge (15);

\draw (21) edge (22);
\draw (22) edge (23);
\draw (23) edge (24);
\draw (24) edge (25);

\draw (31) edge (32);
\draw (32) edge (33);
\draw (33) edge (34);
\draw (34) edge (35);

\draw (41) edge (42);
\draw (42) edge (43);
\draw (43) edge (44);
\draw (44) edge (45);

\draw (51) edge (52);
\draw (52) edge (53);
\draw (53) edge (54);
\draw (54) edge (55);

\draw (11) edge (21);
\draw (21) edge (31);
\draw (31) edge (41);
\draw (41) edge (51);

\draw (12) edge (22);
\draw (22) edge (32);
\draw (32) edge (42);
\draw (42) edge (52);

\draw (13) edge (23);
\draw (23) edge (33);
\draw (33) edge (43);
\draw (43) edge (53);

\draw (14) edge (24);
\draw (24) edge (34);
\draw (34) edge (44);
\draw (44) edge (54);

\draw (15) edge (25);
\draw (25) edge (35);
\draw (35) edge (45);
\draw (45) edge (55);

\draw (11) edge (22);
\draw (21) edge (32);
\draw (31) edge (42);
\draw (41) edge (52);

\draw (12) edge (23);
\draw (22) edge (33);
\draw (32) edge (43);
\draw (42) edge (53);

\draw (13) edge (24);
\draw (23) edge (34);
\draw (33) edge (44);
\draw (43) edge (54);

\draw (14) edge (25);
\draw (24) edge (35);
\draw (34) edge (45);
\draw (44) edge (55);

\draw (21) edge (12);
\draw (31) edge (22);
\draw (41) edge (32);
\draw (51) edge (42);
 
\draw (22) edge (13);
\draw (32) edge (23);
\draw (42) edge (33);
\draw (52) edge (43);
 
\draw (23) edge (14);
\draw (33) edge (24);
\draw (43) edge (34);
\draw (53) edge (44);
 
\draw (24) edge (15);
\draw (34) edge (25);
\draw (44) edge (35);
\draw (54) edge (45);

\draw[dashed] (25) edge (16);
\draw[dashed] (35) edge (26);
\draw[dashed] (45) edge (36);
\draw[dashed] (55) edge (46);

\draw[dashed] (20) edge (11);
\draw[dashed] (30) edge (21);
\draw[dashed] (40) edge (31);
\draw[dashed] (50) edge (41);

\draw[dashed] (10) edge (21);
\draw[dashed] (20) edge (31);
\draw[dashed] (30) edge (41);
\draw[dashed] (40) edge (51);

\draw[dashed] (15) edge (26);
\draw[dashed] (25) edge (36);
\draw[dashed] (35) edge (46);
\draw[dashed] (45) edge (56);

\draw[dashed] (52) edge (61);
\draw[dashed] (53) edge (62);
\draw[dashed] (54) edge (63);
\draw[dashed] (55) edge (64);

\draw[dashed] (02) edge (11);
\draw[dashed] (03) edge (12);
\draw[dashed] (04) edge (13);
\draw[dashed] (05) edge (14);

\draw[dashed] (01) edge (12);
\draw[dashed] (02) edge (13);
\draw[dashed] (03) edge (14);
\draw[dashed] (04) edge (15);

\draw[dashed] (51) edge (62);
\draw[dashed] (52) edge (63);
\draw[dashed] (53) edge (64);
\draw[dashed] (54) edge (65);

\draw[dashed] (51) edge (61);
\draw[dashed] (52) edge (62);
\draw[dashed] (53) edge (63);
\draw[dashed] (54) edge (64);
\draw[dashed] (55) edge (65);

\draw[dashed] (01) edge (11);
\draw[dashed] (02) edge (12);
\draw[dashed] (03) edge (13);
\draw[dashed] (04) edge (14);
\draw[dashed] (05) edge (15);

\draw[dashed] (15) edge (16);
\draw[dashed] (25) edge (26);
\draw[dashed] (35) edge (36);
\draw[dashed] (45) edge (46);
\draw[dashed] (55) edge (56);

\draw[dashed] (10) edge (11);
\draw[dashed] (20) edge (21);
\draw[dashed] (30) edge (31);
\draw[dashed] (40) edge (41);
\draw[dashed] (50) edge (51);

\draw (11) edge (32);
\draw (11) edge (42);
\draw (11) edge (23);
\draw (11) edge (24);

\draw[dashed] (11) edge (30);
\draw[dashed] (11) edge (40);
\draw[dashed] (11) edge (04);
\draw[dashed] (11) edge (03);

\draw (21) edge (33);
\draw (21) edge (34);
\draw (21) edge (13);
\draw (21) edge (14);

\draw (21) edge (52);
\draw (21) edge (42);
\draw[dashed] (21) edge (50);
\draw[dashed] (21) edge (40);

\draw (31) edge (43);
\draw (31) edge (44);
\draw (31) edge (23);
\draw (31) edge (24);

\draw (31) edge (12);
\draw (31) edge (52);
\draw[dashed] (31) edge (10);
\draw[dashed] (31) edge (50);

\draw (41) edge (53);
\draw (41) edge (54);
\draw (41) edge (33);
\draw (41) edge (34);

\draw (41) edge (12);
\draw (41) edge (22);
\draw[dashed] (41) edge (10);
\draw[dashed] (41) edge (20);

\draw (51) edge (43);
\draw (51) edge (44);
\draw[dashed] (51) edge (63);
\draw[dashed] (51) edge (64);

\draw (51) edge (12);
\draw (51) edge (22);
\draw[dashed] (51) edge (10);
\draw[dashed] (51) edge (20);

\draw (12) edge (31);
\draw (12) edge (41);
\draw (12) edge (33);
\draw (12) edge (43);

\draw (12) edge (24);
\draw (12) edge (25);
\draw[dashed] (12) edge (04);
\draw[dashed] (12) edge (05);

\draw (22) edge (41);
\draw (22) edge (51);
\draw (22) edge (43);
\draw (22) edge (53);

\draw (22) edge (14);
\draw (22) edge (15);
\draw (22) edge (34);
\draw (22) edge (35);

\draw (32) edge (51);
\draw (32) edge (11);
\draw (32) edge (53);
\draw (32) edge (13);

\draw (32) edge (24);
\draw (32) edge (25);
\draw (32) edge (44);
\draw (32) edge (45);

\draw (42) edge (21);
\draw (42) edge (11);
\draw (42) edge (23);
\draw (42) edge (13);

\draw (42) edge (34);
\draw (42) edge (35);
\draw (42) edge (54);
\draw (42) edge (55);

\draw (52) edge (31);
\draw (52) edge (21);
\draw (52) edge (33);
\draw (52) edge (23);

\draw (52) edge (44);
\draw (52) edge (45);
\draw[dashed] (52) edge (64);
\draw[dashed] (52) edge (65);

\draw (13) edge (32);
\draw (13) edge (42);
\draw (13) edge (34);
\draw (13) edge (44);

\draw (13) edge (25);
\draw (13) edge (21);
\draw[dashed] (13) edge (05);
\draw[dashed] (13) edge (01);

\draw (23) edge (42);
\draw (23) edge (52);
\draw (23) edge (44);
\draw (23) edge (54);

\draw (23) edge (15);
\draw (23) edge (11);
\draw (23) edge (35);
\draw (23) edge (31);

\draw (33) edge (52);
\draw (33) edge (12);
\draw (33) edge (54);
\draw (33) edge (14);

\draw (33) edge (25);
\draw (33) edge (21);
\draw (33) edge (45);
\draw (33) edge (41);

\draw (43) edge (22);
\draw (43) edge (12);
\draw (43) edge (24);
\draw (43) edge (14);

\draw (43) edge (35);
\draw (43) edge (31);
\draw (43) edge (55);
\draw (43) edge (51);

\draw (53) edge (32);
\draw (53) edge (22);
\draw (53) edge (34);
\draw (53) edge (24);

\draw (53) edge (45);
\draw (53) edge (41);
\draw[dashed] (53) edge (65);
\draw[dashed] (53) edge (61);

\draw (14) edge (33);
\draw (14) edge (43);
\draw (14) edge (35);
\draw (14) edge (45);

\draw (14) edge (22);
\draw (14) edge (21);
\draw[dashed] (14) edge (02);
\draw[dashed] (14) edge (01);

\draw (24) edge (43);
\draw (24) edge (53);
\draw (24) edge (45);
\draw (24) edge (55);

\draw (24) edge (12);
\draw (24) edge (11);
\draw (24) edge (32);
\draw (24) edge (31);

\draw (34) edge (53);
\draw (34) edge (13);
\draw (34) edge (55);
\draw (34) edge (15);

\draw (34) edge (22);
\draw (34) edge (21);
\draw (34) edge (42);
\draw (34) edge (41);

\draw (44) edge (23);
\draw (44) edge (13);
\draw (44) edge (25);
\draw (44) edge (15);

\draw (44) edge (32);
\draw (44) edge (31);
\draw (44) edge (52);
\draw (44) edge (51);

\draw (54) edge (33);
\draw (54) edge (23);
\draw (54) edge (35);
\draw (54) edge (25);

\draw (54) edge (42);
\draw (54) edge (41);
\draw[dashed] (54) edge (62);
\draw[dashed] (54) edge (61);

\draw[dashed] (15) edge (02);
\draw[dashed] (15) edge (03); 

\draw[dashed] (55) edge (62);
\draw[dashed] (55) edge (63);

\draw[dashed] (15) edge (46);
\draw[dashed] (15) edge (36); 

\draw[dashed] (25) edge (56);
\draw[dashed] (25) edge (46); 

\draw[dashed] (35) edge (56);
\draw[dashed] (35) edge (46); 

\draw[dashed] (45) edge (16);
\draw[dashed] (45) edge (56); 

\draw[dashed] (55) edge (16);
\draw[dashed] (55) edge (26); 6+2+6

\end{tikzpicture}
    \caption{Two-fold $f$-confusion graphs of $\tilde g$.}
\end{figure}
\begin{figure}[h]
    \centering
    \begin{tikzpicture} 
[
  node distance=0.5cm,
  every node/.style={circle, draw, minimum size=4mm, inner sep=0pt, font = \normalsize},
  every edge/.style={draw, -} 
]

\node[draw=black] (11) at (1,1) {1,1};
\node[draw=black] (12) at (1,2) {1,2};
\node[draw=black] (13) at (1,3) {1,3};
\node[draw=black] (14) at (1,4) {1,4};
\node[draw=black] (15) at (1,5) {1,5};

\node[draw=black] (21) at (2,1) {2,1};
\node[draw=black] (22) at (2,2) {2,2};
\node[draw=black] (23) at (2,3) {2,3};
\node[draw=black] (24) at (2,4) {2,4};
\node[draw=black] (25) at (2,5) {2,5};

\node[draw=black] (31) at (3,1) {3,1};
\node[draw=black] (32) at (3,2) {3,2};
\node[draw=black] (33) at (3,3) {3,3};
\node[draw=black] (34) at (3,4) {3,4};
\node[draw=black] (35) at (3,5) {3,5};

\node[draw=black] (41) at (4,1) {4,1};
\node[draw=black] (42) at (4,2) {4,2};
\node[draw=black] (43) at (4,3) {4,3};
\node[draw=black] (44) at (4,4) {4,4};
\node[draw=black] (45) at (4,5) {4,5};

\node[draw=black] (51) at (5,1) {5,1};
\node[draw=black] (52) at (5,2) {5,2};
\node[draw=black] (53) at (5,3) {5,3};
\node[draw=black] (54) at (5,4) {5,4};
\node[draw=black] (55) at (5,5) {5,5};

\node[draw=black] (01) at (0,1) {5,1};
\node[draw=black] (02) at (0,2) {5,2};
\node[draw=black] (03) at (0,3) {5,3};
\node[draw=black] (04) at (0,4) {5,4};
\node[draw=black] (05) at (0,5) {5,5};

\node[draw=black] (61) at (6,1) {1,1};
\node[draw=black] (62) at (6,2) {1,2};
\node[draw=black] (63) at (6,3) {1,3};
\node[draw=black] (64) at (6,4) {1,4};
\node[draw=black] (65) at (6,5) {1,5};

\node[draw=black] (10) at (1,0) {1,5};
\node[draw=black] (20) at (2,0) {2,5};
\node[draw=black] (30) at (3,0) {3,5};
\node[draw=black] (40) at (4,0) {4,5};
\node[draw=black] (50) at (5,0) {5,5};

\node[draw=black] (16) at (1,6) {1,1};
\node[draw=black] (26) at (2,6) {2,1};
\node[draw=black] (36) at (3,6) {3,1};
\node[draw=black] (46) at (4,6) {4,1};
\node[draw=black] (56) at (5,6) {5,1};

\draw (11) edge (12);
\draw (12) edge (13);
\draw (13) edge (14);
\draw (14) edge (15);

\draw (21) edge (22);
\draw (22) edge (23);
\draw (23) edge (24);
\draw (24) edge (25);

\draw (31) edge (32);
\draw (32) edge (33);
\draw (33) edge (34);
\draw (34) edge (35);

\draw (41) edge (42);
\draw (42) edge (43);
\draw (43) edge (44);
\draw (44) edge (45);

\draw (51) edge (52);
\draw (52) edge (53);
\draw (53) edge (54);
\draw (54) edge (55);

\draw (11) edge (21);
\draw (21) edge (31);
\draw (31) edge (41);
\draw (41) edge (51);

\draw (12) edge (22);
\draw (22) edge (32);
\draw (32) edge (42);
\draw (42) edge (52);

\draw (13) edge (23);
\draw (23) edge (33);
\draw (33) edge (43);
\draw (43) edge (53);

\draw (14) edge (24);
\draw (24) edge (34);
\draw (34) edge (44);
\draw (44) edge (54);

\draw (15) edge (25);
\draw (25) edge (35);
\draw (35) edge (45);
\draw (45) edge (55);

\draw (11) edge (22);
\draw (21) edge (32);
\draw (31) edge (42);
\draw (41) edge (52);

\draw (12) edge (23);
\draw (22) edge (33);
\draw (32) edge (43);
\draw (42) edge (53);

\draw (13) edge (24);
\draw (23) edge (34);
\draw (33) edge (44);
\draw (43) edge (54);

\draw (14) edge (25);
\draw (24) edge (35);
\draw (34) edge (45);
\draw (44) edge (55);

\draw (21) edge (12);
\draw (31) edge (22);
\draw (41) edge (32);
\draw (51) edge (42);
 
\draw (22) edge (13);
\draw (32) edge (23);
\draw (42) edge (33);
\draw (52) edge (43);
 
\draw (23) edge (14);
\draw (33) edge (24);
\draw (43) edge (34);
\draw (53) edge (44);
 
\draw (24) edge (15);
\draw (34) edge (25);
\draw (44) edge (35);
\draw (54) edge (45);

\draw[dashed] (25) edge (16);
\draw[dashed] (35) edge (26);
\draw[dashed] (45) edge (36);
\draw[dashed] (55) edge (46);

\draw[dashed] (20) edge (11);
\draw[dashed] (30) edge (21);
\draw[dashed] (40) edge (31);
\draw[dashed] (50) edge (41);

\draw[dashed] (10) edge (21);
\draw[dashed] (20) edge (31);
\draw[dashed] (30) edge (41);
\draw[dashed] (40) edge (51);

\draw[dashed] (15) edge (26);
\draw[dashed] (25) edge (36);
\draw[dashed] (35) edge (46);
\draw[dashed] (45) edge (56);

\draw[dashed] (52) edge (61);
\draw[dashed] (53) edge (62);
\draw[dashed] (54) edge (63);
\draw[dashed] (55) edge (64);

\draw[dashed] (02) edge (11);
\draw[dashed] (03) edge (12);
\draw[dashed] (04) edge (13);
\draw[dashed] (05) edge (14);

\draw[dashed] (01) edge (12);
\draw[dashed] (02) edge (13);
\draw[dashed] (03) edge (14);
\draw[dashed] (04) edge (15);

\draw[dashed] (51) edge (62);
\draw[dashed] (52) edge (63);
\draw[dashed] (53) edge (64);
\draw[dashed] (54) edge (65);

\draw[dashed] (51) edge (61);
\draw[dashed] (52) edge (62);
\draw[dashed] (53) edge (63);
\draw[dashed] (54) edge (64);
\draw[dashed] (55) edge (65);

\draw[dashed] (01) edge (11);
\draw[dashed] (02) edge (12);
\draw[dashed] (03) edge (13);
\draw[dashed] (04) edge (14);
\draw[dashed] (05) edge (15);

\draw[dashed] (15) edge (16);
\draw[dashed] (25) edge (26);
\draw[dashed] (35) edge (36);
\draw[dashed] (45) edge (46);
\draw[dashed] (55) edge (56);

\draw[dashed] (10) edge (11);
\draw[dashed] (20) edge (21);
\draw[dashed] (30) edge (31);
\draw[dashed] (40) edge (41);
\draw[dashed] (50) edge (51);

\draw[dashed] (11) edge (03);
\draw[dashed] (11) edge (30);
\draw[dashed] (53) edge (61);
\draw[dashed] (35) edge (16);
\draw (11) edge (23);
\draw (11) edge (32);

\draw (21) edge (33);
\draw (31) edge (43);
\draw (31) edge (23);
\draw (41) edge (53);
\draw (41) edge (33);
\draw (51) edge (43);
\draw[dashed] (51) edge (63);

\draw (12) edge (33);
\draw (12) edge (31);

\draw[dashed] (13) edge (01); 
\draw (13) edge (21);
\draw (13) edge (32);
\draw (13) edge (34);
 
\draw (33) edge (14); 
 
\draw (14) edge (35); 
\draw (15) edge (34); 
\draw[dashed] (15) edge (36); 
\draw[dashed] (31) edge (10);

\end{tikzpicture}
    \caption{Two-fold $f$-confusion graphs of $\tilde h$.}
\end{figure}

\end{document}